%% file: main.tex
\pdfoutput=1

\newcommand{\WITHAPPENDIX}{}

\ifdefined\WITHAPPENDIX
    \documentclass[acmsmall,nonacm]{acmart}
\else
    \documentclass[acmsmall]{acmart}
\fi

\ifdefined\WITHAPPENDIX
\else
    \setcopyright{rightsretained}
    \acmPrice{}
    \acmDOI{10.1145/3622870}
    \acmYear{2023}
    \copyrightyear{2023}
    \acmSubmissionID{oopslab23main-p651-p}
    \acmJournal{PACMPL}
    \acmVolume{7}
    \acmNumber{OOPSLA2}
    \acmArticle{294}
    \acmMonth{10}
    \received{2023-04-14}
    \received[accepted]{2023-08-27}
\fi

\usepackage{colortbl} %
\usepackage{stmaryrd} %
\usepackage[nameinlink]{cleveref}
\renewcommand{\cref}{\Cref}
\usepackage{braket}
\usepackage[]{todonotes}
\usepackage{xspace}
\usepackage{color,soul}
\usepackage[inline]{enumitem}
\usepackage[skins]{tcolorbox}
\usepackage{xcolor}
\usepackage{nicefrac}
\usepackage{xfrac}
\usepackage{csquotes}
\usepackage{wrapfig}
\usepackage{mathtools}
\usepackage{multirow}
\usepackage{thmtools}
\usepackage{thm-restate}
\usepackage{afterpage}

\input{macros}

\usepackage{array}
\newcolumntype{L}{>{$}l<{$}}
\newcolumntype{C}{>{$}c<{$}}
\newcolumntype{R}{>{$}r<{$}}

\usepackage{graphicx}
\graphicspath{ {./assets/} }

\usepackage[nameinlink]{cleveref}

\usepackage{tikz}
\usetikzlibrary{decorations.pathreplacing}
\usetikzlibrary{cd}
\usetikzlibrary{calc}
\usepackage{adjustbox}
\usepackage{pgfplots}
\pgfplotsset{compat=1.11}

\begin{document}

\ifdefined\WITHAPPENDIX
    \title[A Deductive Verification Infrastructure for Probabilistic Programs (Extended Version)]{A Deductive Verification Infrastructure for \\ Probabilistic Programs (Extended Version)}
\else
    \title[A Deductive Verification Infrastructure for Probabilistic Programs]{A Deductive Verification Infrastructure for \\ Probabilistic Programs}
\fi

\ifdefined\WITHAPPENDIX
    \titlenote{This is an extended and updated version of the publication at OOPSLA 2023 (\url{https://doi.org/10.1145/3622870}).}
\fi

\author{Philipp Schröer}
\email{phisch@cs.rwth-aachen.de}
\orcid{0000-0002-4329-530X}
\affiliation{%
  \institution{RWTH Aachen University}
   \country{Germany}
}

\author{Kevin Batz}
\email{kevin.batz@cs.rwth-aachen.de}
\orcid{0000-0001-8705-2564}
\affiliation{%
  \institution{RWTH Aachen University}
   \country{Germany}
}

\author{Benjamin Lucien Kaminski}
\email{kaminski@cs.uni-saarland.de}
\orcid{0000-0001-5185-2324}
\affiliation{%
    \institution{Saarland University}
    \country{Germany}
}
\affiliation{%
    \institution{University College London}
    \country{United Kingdom}
}

\author{Joost-Pieter Katoen}
\email{katoen@cs.rwth-aachen.de}
\orcid{0000-0002-6143-1926}
\affiliation{%
    \institution{RWTH Aachen University}
    \country{Germany}
}

\author{Christoph Matheja}
\email{chmat@dtu.dk}
\orcid{0000-0001-9151-0441}
\affiliation{%
    \institution{Technical University of Denmark}
    \country{Denmark}
}

\renewcommand{\shortauthors}{Schröer, Batz, Kaminski, Katoen, Matheja}

\begin{abstract}
    \input{sections/0_abstract}
\end{abstract}

\ccsdesc[500]{Theory of computation~Logic and verification}
\ccsdesc[500]{Theory of computation~Automated reasoning}
\ccsdesc[500]{Theory of computation~Hoare logic}
\ccsdesc[500]{Theory of computation~Axiomatic semantics}
\ccsdesc[500]{Theory of computation~Denotational semantics}
\ccsdesc[500]{Theory of computation~Invariants}
\ccsdesc[500]{Theory of computation~Program specifications}
\ccsdesc[500]{Theory of computation~Pre- and post-conditions}
\ccsdesc[500]{Theory of computation~Program verification}
\ccsdesc[500]{Theory of computation~Assertions}

\keywords{deductive verification, quantitative verification, probabilistic programs, weakest preexpectations, real-valued logics, automated reasoning}

\maketitle

\input{sections/1_introduction}

\input{sections/3_heylo}

\input{sections/4_heyvl}

\input{sections/5_case_studies}

\input{sections/6_implementation}
\input{sections/7_related_work}
\input{sections/8_conclusion}

\input{sections/9_data_availability}

\begin{acks}
    This work was partially supported by the Digital Research Centre Denmark (DIREC), the ERC Advanced Research Grant FRAPPANT (grant no. 787914), and the 2022 WhatsApp Privacy Aware Program Analysis Research Award.
\end{acks}

\bibliographystyle{ACM-Reference-Format}
\bibliography{literature}

\ifdefined\WITHAPPENDIX
\clearpage

\appendix
\input{appendices/proofs_heylo}
\input{appendices/proofs_heyvl}
\input{appendices/encodings}

\fi

\end{document}

%% file: macros.tex
\newcommand{\morespace}[1]{~{}#1{}~}
\newcommand{\qmorespace}[1]{\quad{}#1{}\quad}
\newcommand{\qqmorespace}[1]{\qquad{}#1{}\qquad}

\newcommand{\eeq}{\morespace{=}}
\newcommand{\qeq}{\qmorespace{=}}
\newcommand{\qqeq}{\qqmorespace{=}}

\newcommand{\lleq}{\morespace{\leq}}

\newcommand{\hheyloleq}{\morespace{\heyloleq}}
\newcommand{\hheylogeq}{\morespace{\heylogeq}}
\newcommand{\eexpleq}{\morespace{\expleq}}
\newcommand{\eexpgeq}{\morespace{\expgeq}}

\newcommand{\qiff}{\qmorespace{\textnormal{iff}}}
\newcommand{\qqiff}{\qqmorespace{\textnormal{iff}}}

\newcommand{\qimplies}{\qmorespace{\textnormal{implies}}}

\newcommand{\qand}{\qmorespace{\textnormal{and}}}
\newcommand{\qqand}{\qqmorespace{\textnormal{and}}}

\newenvironment{halfboxl}{%
\noindent
\begin{minipage}[t]{0.4875\linewidth}
}{%
\end{minipage}\hspace{0.025\linewidth}}
\newenvironment{halfboxr}{%
\begin{minipage}[t]{0.4875\linewidth}
}{%
\end{minipage}\bigbreak}

\usepackage{xcolor}
\definecolor{hintgray}{RGB}{92,92,92}

\definecolor{col1}{RGB}{254,195,10}
\definecolor{col2}{RGB}{251,0,6}
\definecolor{col3}{RGB}{252,11,135}
\definecolor{col4}{RGB}{17,139,2}
\definecolor{col5}{RGB}{14,131,254}
\definecolor{col6}{RGB}{82,0,135}

\definecolor{brLightGreen}{HTML}{addd8e}
\definecolor{brLightGray}{HTML}{f0f0f0}

\definecolor{orange}{RGB}{230,159,0}
\definecolor{skyblue}{RGB}{86,180,233}
\definecolor{brBlue}{RGB}{0,114,178}
\definecolor{bluishgreen}{RGB}{0,158,115}
\definecolor{vermillion}{RGB}{213,94,0}
\definecolor{reddishpurple}{RGB}{204,121,167}

\colorlet{heyvlColor}{vermillion}
\colorlet{prepostColor}{brBlue}

\newcommand{\explain}[1]{\tag*{\small\color{hintgray}(#1)}}

\newcommand{\ifThenElse}[3]{\begin{cases}{#2}, & \text{if } {#1}\\{#3}, & \text{otherwise} \end{cases}}
\newcommand{\ifThenElseDot}[3]{\begin{cases}{#2}, & \text{if } {#1}\\{#3}, & \text{otherwise}~. \end{cases}}

\newcommand{\triple}[3]{\ensuremath{\langle {\color{prepostColor}#1} \rangle~#2~\langle {\color{prepostColor}#3} \rangle}}

\newcommand{\Sup}{\reflectbox{\textnormal{\textsf{\fontfamily{phv}\selectfont S}}}\hspace{.2ex}}
\newcommand{\Inf}{\raisebox{.6\depth}{\rotatebox{-30}{\textnormal{\textsf{\fontfamily{phv}\selectfont \reflectbox{J}}}}\hspace{-.1ex}}}

\newcommand{\quant}[2]{#1.~#2}
\newcommand{\squant}[2]{\ensuremath{\Sup \quant{#1}{#2}}}
\newcommand{\iquant}[2]{\ensuremath{\Inf \quant{#1}{#2}}}

\newcommand{\true}{\mathsf{true}}
\newcommand{\false}{\mathsf{false}}

\newcommand{\impl}{\rightarrow}
\newcommand{\coimpl}{\leftsquigarrow}
\newcommand{\coneg}{{\sim}}

\newcommand{\hla}{\ensuremath{\varphi}}
\newcommand{\hlb}{\ensuremath{\psi}}
\newcommand{\hlc}{\ensuremath{\rho}}
\newcommand{\hld}{\ensuremath{\gamma}}

\newcommand{\colheylo}[1]{{\color{prepostColor}#1}}
\newcommand{\colheyvl}[1]{{\color{heyvlColor}#1}}
\newcommand{\hhla}{\colheylo{\hla}}
\newcommand{\hhlb}{\colheylo{\hlb}}
\newcommand{\hhlc}{\colheylo{\hlc}}

\newcommand{\expa}{\ensuremath{X}}
\newcommand{\expb}{\ensuremath{Y}}

\newcommand{\expleq}{\ensuremath{\preceq}}
\newcommand{\expgeq}{\ensuremath{\succeq}}

\newcommand{\substBy}[2]{[{#1} \mapsto {#2}]}

\newcommand{\iverson}[1]{\left[{#1}\right]}

\newcommand{\Terms}{\mathcal{T}}

\newcommand{\interpretsimple}[1]{\llbracket{#1}\rrbracket}

\newcommand{\eval}[1]{\llbracket{#1}\rrbracket}

\newcommand{\pexp}[2]{#1 \cdot \langle #2 \rangle }
\newcommand{\pexpand}{+}
\newcommand{\sfsymbol}[1]{\textsf{\upshape {#1}}}
\newcommand{\Vars}{\sfsymbol{Vars}\xspace}

\newcommand{\typeof}[2]{\ensuremath{#1 \colon {#2}}}
\newcommand{\Types}{\mathsf{Types}\xspace}
\newcommand{\pGCL}{\sfsymbol{pGCL}\xspace}
\newcommand{\HeyVL}{\sfsymbol{HeyVL}\xspace}
\newcommand{\HeyLo}{\sfsymbol{HeyLo}\xspace}

\newcommand{\funcfont}[1]{\sfsymbol{#1}}
\newcommand{\listtype}{\typefont{Lists}}
\newcommand{\listlen}{\funcfont{len}}

\newcommand{\typevar}{\ensuremath{\tau}}

\newcommand{\heyloleq}{\sqsubseteq}
\newcommand{\heylogeq}{\sqsupseteq}

\newcommand{\aexpr}{\ensuremath{a}} %
\newcommand{\bexpr}{\ensuremath{b}} %
\newcommand{\procname}{P} %

\newcommand{\heylovalidate}[1]{\ensuremath{\triangle\!\left( #1 \right)}}
\newcommand{\heylocovalidate}[1]{\ensuremath{\triangledown\!\left( #1 \right)}}

\newcommand{\funcsymb}{\ensuremath{f}}
\newcommand{\termvar}{\ensuremath{t}}

\newcommand{\typefont}[1]{\mathsf{#1}}
\newcommand{\bool}{\Bools}
\newcommand{\uint}{\Nats}
\newcommand{\ureal}{\PosReals}

\newcommand{\Bools}{\mathbb{B}}
\newcommand{\Nats}{\mathbb{N}}
\newcommand{\Ints}{\mathbb{Z}}
\newcommand{\Rats}{\mathbb{Q}}
\newcommand{\PosRats}{\mathbb{Q}_{\geq 0}}
\newcommand{\Reals}{\mathbb{R}}
\newcommand{\PosReals}{\mathbb{R}_{\geq 0}}
\newcommand{\PosRealsInf}{\mathbb{R}_{\geq 0}^{\infty}}

\newcommand{\States}{\mathsf{States}}
\newcommand{\State}{\sigma}
\newcommand{\Vals}{\mathsf{Vals}}

\newcommand{\symSkip}{\ensuremath{\texttt{skip}}}
\newcommand{\symAssign}{\morespace{\coloneqq}}
\newcommand{\symRasgn}{\colonapprox}
\newcommand{\symSemi}{\ensuremath{\texttt{;}~}}
\newcommand{\symIf}{\ensuremath{\texttt{if}}}
\newcommand{\symElse}{\ensuremath{\texttt{else}}}
\newcommand{\symAssert}{\ensuremath{\texttt{assert}}}
\newcommand{\symAssume}{\ensuremath{\texttt{assume}}}

\newcommand{\symDemonic}{\ensuremath{\symIf~(\sqcap)}}
\newcommand{\symCdot}{\ensuremath{\symIf~(\cdot)}}
\newcommand{\symAngelic}{\ensuremath{\symIf~(\sqcup)}}

\newcommand{\symHavoc}{\ensuremath{\texttt{havoc}}}

\newcommand{\symWhile}{\ensuremath{\texttt{while}}}

\newcommand{\symObserve}{\ensuremath{\texttt{observe}}}
\newcommand{\symTick}{\ensuremath{\texttt{reward}}}
\newcommand{\symProc}{\ensuremath{\texttt{proc}}}
\newcommand{\symcoProc}{\ensuremath{\texttt{coproc}}}
\newcommand{\symReturns}{\texttt{->}}
\newcommand{\symRequires}{\ensuremath{\texttt{pre}}}
\newcommand{\symEnsures}{\ensuremath{\texttt{post}}}
\newcommand{\symValidate}{\ensuremath{\texttt{validate}}}

\newcommand{\symVar}{\ensuremath{\texttt{var}}}

\newcommand{\blockStart}{\ensuremath{\{}}
\newcommand{\blockEnd}{\ensuremath{\}}}

\usepackage{pgfplots}
\pgfplotsset{compat=newest}
\usepackage{tikz}
\usetikzlibrary{arrows,decorations.pathmorphing,decorations.pathreplacing,positioning,fit,trees,shapes,shadows,automata,calc,shapes.multipart,backgrounds}

\usetikzlibrary{decorations.text}

\tikzstyle{myarrows}=[-latex, shorten <=10mm, shorten >=10mm, line width=4mm,draw=rwth-50]

\newcommand{\proc}[3]{\symProc~{\mathit{#1}}\,\texttt{(}{#2}\texttt{)}~\symReturns~\texttt{(}{#3}\texttt{)}}
\newcommand{\coproc}[3]{\symcoProc~{\mathit{#1}}\,\texttt{(}{#2}\texttt{)}~\symReturns~\texttt{(}{#3}\texttt{)}}
\newcommand{\Ensures}[1]{\symEnsures~{\color{prepostColor}#1}}
\newcommand{\Requires}[1]{\symRequires~{\color{prepostColor}#1}}

\newcommand{\varandtype}[2]{\ensuremath{{#1}\texttt{:}\,\mathsf{#2}}}
\newcommand{\varin}{\ensuremath{\mathit{in}}}
\newcommand{\varout}{\ensuremath{\mathit{out}}}

\newcommand{\stmt}{\ensuremath{S}}
\newcommand{\sstmt}{{\color{heyvlColor}\stmt}}

\newcommand{\stmtSkip}{\symSkip}
\newcommand{\stmtAsgn}[2]{\ensuremath{{#1} \symAssign {#2}}}
\newcommand{\stmtDeclInit}[3]{\symVar~\ensuremath{{#1}\colon {#2} \symRasgn {#3}}}

\newcommand{\stmtRasgn}[2]{\ensuremath{{#1} \symRasgn {#2}}}
\newcommand{\stmtSeq}[2]{\ensuremath{{#1}\symSemi {#2}}}
\newcommand{\stmtIf}[3]{\ensuremath{\symIf~({#1})~\{ {#2} \}~\symElse~\{ {#3} \}}}
\newcommand{\stmtIfStart}[1]{\ensuremath{\symIf~({#1})~\{ }}
\newcommand{\stmtProb}[3]{\ensuremath{\{ {#2} \} ~[{#1}]~ \{ {#3} \}}}
\newcommand{\stmtWhile}[2]{\ensuremath{\symWhile~({#1})~\{ {#2} \}}}
\newcommand{\headerWhile}[1]{\ensuremath{\symWhile~({#1})}}

\newcommand{\stmtHavoc}[1]{\ensuremath{\symHavoc~{#1}}}
\newcommand{\stmtAssert}[1]{\ensuremath{\symAssert~{#1}}}
\newcommand{\stmtAssume}[1]{\ensuremath{\symAssume~{#1}}}

\newcommand{\stmtDemonicStart}{\ensuremath{\symDemonic~\{ }}
\newcommand{\stmtCdotStart}{\ensuremath{\symCdot~\{ }}
\newcommand{\stmtElseStart}{\ensuremath{\symElse~\{ }}
\newcommand{\stmtAngelicStart}{\ensuremath{\symAngelic~\{ }}

\newcommand{\stmtDemonic}[2]{\ensuremath{\stmtDemonicStart{#1} \}~\stmtElseStart {#2} \}}}
\newcommand{\stmtCdot}[2]{\ensuremath{\stmtCdotStart{#1} \}~\stmtElseStart {#2} \}}}

\newcommand{\stmtAngelic}[2]{\ensuremath{\stmtAngelicStart{#1} \}~\stmtElseStart {#2} \}}}

\newcommand{\stmtObserve}[1]{\ensuremath{\symObserve~{#1}}}
\newcommand{\stmtTick}[1]{\ensuremath{\symTick~{#1}}}
\newcommand{\stmtValidate}{\ensuremath{\symValidate}}

\newcommand{\symDown}{}

\newcommand{\downHavoc}[1]{\ensuremath{\symDown\symHavoc~{#1}}}
\newcommand{\downAssert}[1]{\ensuremath{\symDown\symAssert~{#1}}}
\newcommand{\downAssume}[1]{\ensuremath{\symDown\symAssume~{#1}}}

\newcommand{\Havoc}[1]{\ensuremath{\symDown\symHavoc~{#1}}}
\newcommand{\Assert}[1]{\ensuremath{\symDown\symAssert~{#1}}}
\newcommand{\Assume}[1]{\ensuremath{\symDown\symAssume~{#1}}}

\newcommand{\Validate}{\ensuremath{\symDown\stmtValidate}}

\newcommand{\symUp}{\ensuremath{\texttt{co}}}

\newcommand{\upHavoc}[1]{\ensuremath{\symUp\symHavoc~{#1}}}
\newcommand{\upAssert}[1]{\ensuremath{\symUp\symAssert~{#1}}}
\newcommand{\upAssume}[1]{\ensuremath{\symUp\symAssume~{#1}}}

\newcommand{\coHavoc}[1]{\ensuremath{\symUp\symHavoc~{#1}}}
\newcommand{\coAssert}[1]{\ensuremath{\symUp\symAssert~{#1}}}
\newcommand{\coAssume}[1]{\ensuremath{\symUp\symAssume~{#1}}}

\newcommand{\coValidate}{\ensuremath{\symUp\stmtValidate}}

\newcommand{\monus}{\mathbin{\dot-}}

\newcommand{\CC}{\ensuremath{C}}
\newcommand{\loopbody}{\ensuremath{C_0}}
\newcommand{\loopinv}{\ensuremath{I}}
\newcommand{\WHILESYMBOL}{\ensuremath{\textnormal{\texttt{while}}}}

\newcommand{\WHILEDO}[2]{\ensuremath{\WHILESYMBOL \left(\, {#1} \,\right)\left\{\, {#2} \,\right\}}}

\newcommand{\exprTrue}{\ensuremath{\texttt{true}}}
\newcommand{\exprFalse}{\ensuremath{\texttt{false}}}
\newcommand{\exprFlip}[1]{\ensuremath{\texttt{flip}({#1})}}
\newcommand{\exprUnif}[2]{\ensuremath{\texttt{unif}({#1},~{#2})}}

\newcommand{\interpretsimpleState}[1]{\interpretsimple{#1}(\State)}

\newcommand{\interpretsimpleStateSubstBy}[3]{\interpretsimple{#1}(\State\substBy{#2}{#3})}

\newcommand{\evalState}[1]{\eval{#1}(\State)}

\newcommand{\exprIte}[3]{\mathrm{ite}({#1},~{#2},~{#3})}

\newcommand{\symWp}{\sfsymbol{wp}}
\newcommand{\symVc}{\sfsymbol{vp}}
\newcommand{\symVcB}{\sfsymbol{vc}}

\renewcommand{\wp}[1]{\symWp\llbracket{#1}\rrbracket}

\newcommand{\vc}[1]{\symVc\llbracket{#1}\rrbracket}
\newcommand{\vcB}[1]{\symVcB\llbracket{#1}\rrbracket}

\newcommand{\symWlp}{\sfsymbol{wlp}}
\newcommand{\wlp}[1]{\symWlp\llbracket{#1}\rrbracket}

\newcommand{\symErt}{\sfsymbol{ert}}
\newcommand{\ert}[1]{\symErt\llbracket{#1}\rrbracket}

\newcommand{\Expectations}{\mathbb{E}}
\newcommand{\OneBoundedExpectations}{\Expectations_{\leq 1}}

\newcommand{\embed}[1]{\scalebox{0.85}{\textnormal{\textsf{?}}}({#1})}
\newcommand{\coembed}[1]{\scalebox{0.85}{\textnormal{\textsf{co?}}}({#1})}

\newcommand{\wpPhi}{\prescript{\symWp}{}{\Phi}}
\newcommand{\wlpPhi}{\prescript{\symWlp}{}{\Phi}}

\newcommand{\intersem}[1]{\color{hintgray}{//~#1}} %

\newcommand{\Rplus}{\mathbb{R}^\infty_{\geq 0}}

\newcommand{\symDownTransWp}{enc^{\symDown}_{\symWp}}
\newcommand{\symDownTransErt}{enc^{\symDown}_{\symErt}}
\newcommand{\symUpTransWp}{enc^{\symUp}_{\symWp}}
\newcommand{\symDownTransWlp}{enc^{\symDown}_{\symWlp}}
\newcommand{\symUpTransWlp}{enc^{\symUp}_{\symWlp}}

\newcommand{\downTransWp}[1]{\symDownTransWp\lfloor{#1}\rfloor}
\newcommand{\downTransErt}[1]{\symDownTransErt\lfloor{#1}\rfloor}
\newcommand{\upTransWp}[1]{\symUpTransWp\lfloor{#1}\rfloor}
\newcommand{\downTransWlp}[1]{\symDownTransWlp\lfloor{#1}\rfloor}
\newcommand{\upTransWlp}[1]{\symUpTransWlp\lfloor{#1}\rfloor}

\newcommand{\pcc}{C}
\newcommand{\pSkip}{\symSkip}
\newcommand{\pDiverge}{\texttt{diverge}}
\newcommand{\pAssign}[2]{#1~\texttt{:=}~#2}
\newcommand{\pChoice}[3]{\{\,#1\,\}~[#2]~\{\,#3\,\}}
\newcommand{\pNd}[2]{\{\,#1\,\}~[]~\{\,#2\,\}}
\newcommand{\pIte}[3]{\symIf~(#1)~\{ #2 \}~\symElse~\{#3\}}

\newcommand{\symForeign}[1]{\textit{#1}}
\newcommand{\foreignWp}[3]{\symForeign{#1}(#2,#3)}

\newcommand{\downSpec}[2]{spec_{\symDown}(#1,~#2)}
\newcommand{\upSpec}[2]{spec_{\symUp}(#1,~#2)}

\newcommand{\const}[1]{const({#1})}

\newcommand{\symDownIterWlp}{iter_{\symWlp}}
\newcommand{\downIterWlp}[1]{\symDownIterWlp({#1})}
\newcommand{\symUpIterWp}{iter_{\symWp}}
\newcommand{\upIterWp}[1]{\symUpIterWp({#1})}

\newcommand{\symDownPsi}{\prescript{\symWlp}{}{\Psi}}
\newcommand{\symUpPsi}{\prescript{\symUp}{}{\Psi}}

\newcommand{\symDownExtendWlp}{extend_{\symWlp}}
\newcommand{\downExtendWlp}[1]{\symDownExtendWlp({#1})}
\newcommand{\downExtendWlpK}[2]{\symDownExtendWlp^{(#1)}({#2})}
\newcommand{\symUpExtendWp}{extend_{\symWp}}
\newcommand{\upExtendWp}[1]{\symUpExtendWp({#1})}
\newcommand{\upExtendWpK}[2]{\symUpExtendWp^{(#1)}({#2})}

\newcommand{\ohfiveExp}[1]{\mathit{exp}(0.5, {#1})}

\newcounter{challengecounter}
\newtcolorbox{challengebox}[3][]
{
  enhanced,
  width=\textwidth,
  drop shadow southwest,
  sharp corners,
  #1,
}

\newcommand{\lightgray}[1]{\textcolor{lightgray}{#1}}

\newcommand{\setcomp}[2]{\left\{#1 ~{}\middle|{}~ #2\right\}}
\newcommand{\tool}[1]{\textnormal{\textsc{#1}}}

\newcommand{\proofcase}[1]{\item[-]\textsc{Case} #1.}

\newlength{\appendixprogramsep}
\NewDocumentEnvironment{appendixprogram}{+b}{%
    \setlength{\jot}{0pt}%
    \setlength{\abovedisplayshortskip}{\abovedisplayskip}%
    \setlength{\belowdisplayshortskip}{\belowdisplayskip}%
    \allowdisplaybreaks[0]%
    \begin{align*}
        #1
    \end{align*}%
}{}

%% file: sections/0_abstract.tex
This paper presents a quantitative program verification infrastructure for discrete probabilistic programs. 
Our infrastructure can be viewed as the probabilistic analogue of Boogie: its central components are an intermediate verification language (IVL) together with a real-valued logic. 
Our IVL provides a programming-language-style for expressing verification conditions whose validity implies the correctness of a program under investigation. 
As our focus is on verifying quantitative properties such as bounds on expected outcomes, expected run-times, or termination probabilities, off-the-shelf IVLs based on Boolean first-order logic do not suffice. 
Instead, a paradigm shift from the standard Boolean to a \emph{real-valued} domain is required.

Our IVL features quantitative generalizations of standard verification constructs such as \texttt{assume}- and \texttt{assert}-statements. 
Verification conditions are generated by a weakest-precondition-style semantics, based on our real-valued logic. 
We show that our verification infrastructure supports natural encodings of numerous verification techniques from the literature. 
With our SMT-based implementation, we automatically verify a variety of benchmarks. 
To the best of our knowledge, this establishes the first deductive verification infrastructure for expectation-based reasoning about probabilistic programs.

%% file: sections/1_introduction.tex
\section{Introduction and Overview}

Probabilistic programs differ from ordinary programs by the ability to base decisions on samples from probability distributions.
They are found in randomized algorithms, communication protocols,  models of physical and biological processes, and -- more recently -- statistical models used in machine learning and artificial intelligence (cf.~\cite{gordonProbabilisticProgramming2014,bartheFoundationsProbabilisticProgramming2020}).
Typical questions in the design and analysis of probabilistic programs are concerned with \emph{quantifying} aspects of their \emph{expected} -- or average -- \emph{behavior}, e.g.\ the \emph{expected runtime} of a randomized algorithm, the \emph{expected number of retransmissions} in a protocol, or the \emph{probability} that a particle reaches its destination. 

Writing correct probabilistic programs is notoriously hard. 
They may contain subtle bugs occurring with low probability or undesirably favor certain results in the long run.
In fact, reasoning about the expected behavior of probabilistic programs is known to be strictly harder than for ordinary programs~\cite{kaminskiHardnessAnalyzingProbabilistic2019}.

There exists a plethora of research on verification techniques for probabilistic programs, ranging from program logics (cf.~\cite{mciverAbstractionRefinementProof2005,kaminskiWeakestPreconditionReasoning2018}) to highly specialized proof rules~\cite{harkAimingLowHarder2019,mciverNewProofRule2018}, often with little (if any) automation.
These techniques are based on different branches of mathematics -- e.g.\ domain theory or martingale analysis -- and their relationships are  non-trivial (cf.~\citet{DBLP:journals/toplas/TakisakaOUH21}).
This poses major challenges for comparing -- let alone \emph{combining} -- such different approaches.

In this paper, we build a \emph{verification infrastructure} for reasoning about the expected behavior of (discrete) probabilistic programs; \Cref{fig:intro:infrastructure} gives an overview.
\begin{figure}[t]
\begin{adjustbox}{max width=\textwidth}
  \input{assets/meme.tex}	
\end{adjustbox}
\caption{Architecture of our verification infrastructure.}
\label{fig:intro:infrastructure}
\end{figure}
Modern program verifiers for non-probabilistic programs 
often have a front-end that translates a given program and its specification into an intermediate language, such as \tool{Boogie}~\cite{leinoThisBoogie2008}, \tool{Why3}~\cite{DBLP:conf/esop/FilliatreP13}, or \tool{Viper}~\cite{viper}.  
Such intermediate languages enable the encoding of complex verification techniques, while allowing for the separate development of efficient back-ends, e.g.~verification condition generators.
 In this very spirit, we introduce a novel \emph{quantitative intermediate verification language} that enables
researchers to (i)~prototype and automate new verification techniques, %
(ii) combine proof rules, and %
(iii) benefit from back-end improvements.
Before we dive into details, we discuss five examples of probabilistic programs from the literature that have been verified with five different techniques -- all of them have been encoded in our language and verified with our tool.
\begin{example}[Rabin's Mutual Exclusion Protocol~\textnormal{\cite{kushilevitzRandomizedMutualExclusion1992}}]
This protocol controls processes competing for access to a critical section.
To determine which process gets access, every process will repeatedly toss a fair coin until it sees heads; the process that needed the largest number of tosses is then granted access.
\Cref{fig:intro:rabin} shows a probabilistic program modeling Rabin's protocol: 
$i$~is the number of remaining processes competing for access.
While more than 1 competitor remains, each competitor tosses one coin (inner loop).
If the coin shows heads (i.e.\ if $\exprFlip{0.5}$ samples a 1), that competitor is removed from the pool of remaining competitors (by subtracting $d = 1$ from $i$).
One can verify with the \emph{weakest liberal preexpectation calculus} by \citet{mciverAbstractionRefinementProof2005} that \emph{the probability to select exactly one process (plus the probability of nontermination) is at least $\nicefrac{2}{3}$} if there are initially at least 2 processes. 
\end{example}

\begin{figure}[t]
\begin{minipage}{0.5\textwidth}
\begin{align*}
    &\headerWhile{1 < i}~\{ \\
    &\quad \stmtAsgn{n}{i}\symSemi \\
    &\quad \headerWhile{0 < n}~\{ \\
    &\quad\quad \stmtAsgn{d}{\exprFlip{0.5}}\symSemi \\
    &\quad\quad \stmtAsgn{i}{i-d}\symSemi \\
    &\quad\quad \stmtAsgn{n}{n-1} \\
    &\quad\} \\
    &\}
\end{align*}
\vspace{-1.5em}
\caption{Model of Rabin's Protocol}
\label{fig:intro:rabin}
\bigskip
\begin{align*}
    &\headerWhile{0 < x}~\{ \\
    &\quad \stmtAsgn{i}{N + 1}\symSemi \\
    &\quad \headerWhile{0 < x < i}~\{ \\
    &\quad\quad \stmtRasgn{i}{\exprUnif{1}{N}} \\
    &\quad\} \\
    &\quad \stmtAsgn{x}{x-1} \\
    &\}
\end{align*}
\vspace{-1.5em}
\caption{The Coupon Collector's Problem}
\label{fig:intro:coupon-collector}
\end{minipage}%
\begin{minipage}{0.5\textwidth}
\begin{align*} %
    & \texttt{fn}~\mathit{lossy}(\varandtype{l}{List})~\blockStart \\
    & \quad \stmtIfStart{\mathit{len}(l) > 0} \\
    & \quad \quad \pChoice{~\mathit{lossy}(\mathit{tail}(l))~}{0.5}{~\pDiverge~} \\
    & \quad \blockEnd \\
    & \blockEnd
\end{align*}
\vspace{-2.5em}
\caption{Lossy list traversal}
\label{fig:intro:lossy}

\bigskip

\begin{align*}
    &\headerWhile{x > 0}~\{ \\
    &\quad \stmtAsgn{q}{x/(2 \cdot x + 1)}\symSemi \\
    &\quad \stmtProb{q}{\stmtAsgn{x}{x - 1}}{\stmtAsgn{x}{x + 1}} \\
    & \blockEnd%
\end{align*}
\vspace{-2em}
\caption{Variant of a random walk}
\label{fig:ast-rule}

\bigskip

\begin{align*}
    & \headerWhile{x \neq 0}~\blockStart \\
    & \quad \stmtProb{0.5}{\stmtAsgn{x}{0}}{\stmtAsgn{y}{y + 1}} \\
    & \quad \stmtAsgn{n}{n + 1} \\
    & \blockEnd
\end{align*}
\vspace{-1.5em}
\caption{Counterexample from~\cite{harkAimingLowHarder2019}}
\label{fig:ost-rule}
\end{minipage}%
\end{figure}%

\begin{example}[The Coupon Collector \textnormal{\cite{wiki:cc}}]
\Cref{fig:intro:coupon-collector} models the coupon collector problem -- a well-known problem in probability theory: 
Suppose any box of cereals contains one of $N$ different coupons. 
What is the average number of boxes one needs to buy to collect at least one of all $N$ different coupons, assuming that each coupon type occurs with the same probability?
Our formulation is taken from \cite{kaminskiWeakestPreconditionReasoning2018}; the authors develop an \emph{expected runtime calculus} and use \emph{invariant-based arguments} to show that the \emph{expected number of loop iterations}, which coincides with the average number of boxes one needs to buy, \emph{is bounded from above by $N \cdot H_N$}, where $H_N$ is the $N$-th harmonic number.
\end{example}

\begin{example}[Lossy List Traversal~\textnormal{\cite{batzQuantitativeSeparationLogic2019}}]
\label{ex:motiv-lossy-list}
\Cref{fig:intro:lossy} depicts a recursive function implementing a lossy list traversal; it flips a fair coin (using the probabilistic choice $\pChoice{\ldots}{0.5}{\ldots}$) and, depending on the outcome, either calls itself with the list's tail or diverges, i.e.\ enters an infinite loop.
Using the \emph{weakest preexpectation calculus}~\cite{kozenProbabilisticPDL1983,mciverAbstractionRefinementProof2005}, one can prove that this program terminates with probability \emph{at most} $0.5^{\mathit{len}(l)}$.
Analyzing the lossy list traversal is intuitive -- for every non-empty list, there is exactly one execution that does not diverge; its probability is $0.5^{\mathit{len}(l)}$. What is noteworthy, however, is that even for such a simple program, we need to reason about an exponential function.
This is common when verifying probabilistic programs: proving non-trivial bounds often requires non-linear arithmetic.%
\end{example}

\begin{example}[Fair Random Walk \textnormal{\cite{wiki:rw}}]
\Cref{fig:ast-rule} depicts a variant of a one-dimensional random walk of a particle with position~$x$ -- a well-studied model in physics.
Analyzing the program's termination behavior is hard because the probability $q$ of moving to the left or right changes in every loop iteration depending on the previous position $x$.
\citet{mciverNewProofRule2018} propose a proof rule based on \emph{quasi-variants} that allows proving that \emph{this program terminates almost-surely}, i.e. with probability one.
Fair random walks, i.e.\ if $q = \nicefrac{1}{2}$, are well-known to terminate almost-surely but still have \mbox{infinite expected runtime}.
\end{example}

\begin{example}[Lower Bounds on Expected Values \textnormal{\cite{harkAimingLowHarder2019}}]
\Cref{fig:ost-rule} shows another loop whose control flow depends on the outcome of coin flips.
\citet{harkAimingLowHarder2019} studied this example to demonstrate that induction-based proof rules for \emph{lower bounds}\footnote{Specifically: \emph{lower bound on partial correctness} plus \emph{proof of termination} gives \emph{lower bound on total correctness}.}, which are sound for classical verification, may become unsound when reasoning about probabilistic programs.
The authors used martingale analysis and the optional stopping theorem to develop a sound proof rule capable of proving that, whenever $x \neq 0$ initially holds, then the expected value of~$y$ after the program's termination is \emph{at least} $1 + y$.
\end{example}

\paragraph{Challenges}
We summarize the challenges of developing an infrastructure for automated verification of probabilistic programs revealed by the examples in \Cref{fig:intro:rabin,fig:intro:coupon-collector,fig:intro:lossy,fig:ast-rule,fig:ost-rule}:

First, there are many different verification techniques for probabilistic programs that are based on different concepts, e.g. quantitative invariants, quasi-variants, different notions of martingales, or stopping times of stochastic processes. 
Developing a language that is sufficiently expressive to encode these techniques 
while keeping it
amenable to automation is a major challenge.

Second, verification of probabilistic programs involves \emph{reasoning about both lower- and upper bounds} on expected values. This is different from classical program verification, which can be understood as proving that a given precondition implies a program's weakest precondition, i.e.\ $\texttt{pre} \Rightarrow \wp{C}(\texttt{post})$. In other words, $\texttt{pre}$ is a \emph{lower bound} (in the Boolean lattice) on $\wp{C}(\texttt{post})$. Proving \emph{upper bounds}, i.e. $\wp{C}(\texttt{post})\Rightarrow \texttt{pre}$, has received scarce attention.\footnote{Notable exceptions are Cousot's necessary preconditions~\cite{cousotAutomaticInferenceNecessary2013} and recent works on (partial) incorrectness logic~\cite{ohearnIncorrectnessLogic2020,zhangQuantitativeStrongestPost2022}.}

Third, in \Cref{fig:intro:coupon-collector,fig:intro:lossy,fig:ast-rule}, we noticed that %
verification of probabilistic programs often involves reasoning about \emph{unbounded} random variables and non-linear arithmetic involving exponentials, harmonic numbers, limits, and possibly infinite sums.

\paragraph{Our approach}
We address the first challenge by developing a quantitative IVL and a real-valued logic %
tailored to verification of probabilistic programs. 
The IVL features quantitative generalizations of standard verification constructs such as \texttt{assume}- and \texttt{assert}-statements.
Our quantitative constructs are inspired by G\"odel logics~\cite{baazInfinitevaluedGodelLogics1996,preiningGodelLogicsSurvey2010}. %
In particular, they have \emph{dual \emph{co}-constructs} for verifying upper- instead of lower bounds, thereby addressing the second challenge.
These dual constructs are not only interesting for quantitative reasoning, but indeed also for Boolean reasoning \`{a} la $\wp{C}(\texttt{post})\Rightarrow \texttt{pre}$.
To address the third challenge, we rely on modern SMT solvers' abilities to deal with custom theories, standard techniques for limiting the number of user-defined function applications, and custom optimizations.

\Cref{fig:lossy-heyvl} shows a program written in our quantitative IVL; it encodes the verification of \Cref{ex:motiv-lossy-list}. 
We use a \emph{co}procedure to prove that the \colheylo{quantitative precondition} $\colheylo{\ohfiveExp{\mathit{len}(l)}} = 0.5^{\mathit{len}(l)}$ is an \emph{upper} bound on the procedure's termination probability\footnote{Technically, $\colheylo{\ohfiveExp{\mathit{len}(l)}}$ upper-bounds the expected value of the random variable $\colheylo{1}$ after the procedure's termination.} given by the \colheylo{quantitative postcondition $1$}.
    We establish the above bound for the procedure body while assuming that it holds for recursive calls (cf.~\cite{olmedoReasoningRecursiveProbabilistic2016}).
	Our dual quantitative \texttt{assert}- and \texttt{assume}-statements encode the call in the usual way: we assert the procedure's \colheylo{pre} and assume its \colheylo{post}.

\begin{figure}[t]
\begin{align*} %
    & \coproc{\mathit{lossy}}{\varandtype{l}{List}}{} \\
    & \quad \Requires{\ohfiveExp{\mathit{len}(l)}} \\
    & \quad \Ensures{1} \\
    & \blockStart \\
    & \qquad \stmtIfStart{\mathit{len}(l) > 0} \\
    & \qquad \quad \stmtDeclInit{coin}{\bool}{\exprFlip{0.5}} ~~ \intersem{\text{coin flip}} \\
    & \qquad \quad \stmtIfStart{coin}  \\
    & \qquad \quad \quad \coAssert{\colheylo{\ohfiveExp{\mathit{len}(\mathit{tail}(l))}}}\symSemi\coValidate\symSemi  \coAssume{\colheylo{1}}~~\intersem{\text{call of } \mathit{lossy}(\mathit{tail}(l))} \\
    & \qquad \quad \blockEnd~\stmtElseStart \Assert{\embed{\false}}~\blockEnd~\intersem{\pDiverge} \\
    & \qquad \blockEnd \\
    & \blockEnd
\end{align*}	
\vspace{-1.5em}
\caption{Encoding of the lossy list traversal (see \Cref{fig:intro:lossy}) in our intermediate language.}
\label{fig:lossy-heyvl}
\end{figure}

\paragraph{Contributions} The main contributions of our work are:
\begin{enumerate}
  \item A \emph{novel intermediate verification language} ($\rightarrow$ \Cref{sec:heyvl}) for automating probabilistic program verification techniques featuring \emph{quantitative generalizations of standard verification constructs}, e.g.\ $\texttt{assert}$ and $\texttt{assume}$, and a \emph{formalization} of its %
      semantics based on a real-valued logic ($\rightarrow$~\Cref{sec:heylo}) with constructs inspired by G\"odel logics. %
  \item \emph{Encodings of verification techniques and proof rules} with different theoretical underpinnings (e.g. domain theory, martingales, and the optional stopping theorem) taken from the probabilistic program verification literature into our intermediate language ($\rightarrow$ \Cref{sec:encodings}).
  \item An SMT-backed \emph{verification infrastructure} that enables researchers to prototype and automate verification techniques for probabilistic programs by encoding to our intermediate language, an \emph{experimental evaluation} of its feasibility, and a prototypical \emph{frontend} for verifying programs written in the probabilistic guarded command language ($\rightarrow$ \Cref{sec:implementation}).
\end{enumerate}

\ifdefined\WITHAPPENDIX
\else
Proofs and further details about our encodings are available online in a technical report.
\fi

%% file: assets/meme.tex
\begin{tikzpicture}
		
\node (ert) {expected run-times};

\node[right=2cm of ert] (pc) {partial correctness};

\node[right=1.1cm of pc.east] (rpe) {expected resource consumption};

\node[below=0.2cm of ert.center] (martingales) {martingales};

\node[right=0.5cm of martingales.east] (past) {positive almost-sure termination};

\node[right=0.7cm of past.east] (ast) {almost-sure termination};

\node[below left =0.4cm of past.east] (park) {Park induction};

\node[left=1.3cm of park.west] (aert) {amortised analysis};

\node[right =1.3cm of park.east] (cwp) {conditional expected values};

\node[below =0.2cm of aert.west] (wp) {total correctness};

\node[right=1.3cm of wp.east] (kind) {k-induction};

\node[right=1.3cm of kind.east] (rpe) {probabilistic sensitivity};

\begin{scope}[on background layer]
	\node[minimum width=15cm,draw, fit=(pc)(martingales)(past)(ast)(park)(cwp)(wp)(kind)(rpe)(aert), inner sep=5pt, fill=brLightGray] (expectations) {};
\end{scope}

\node[minimum height=0.8cm,minimum width=15cm,inner sep=5pt,align=center,below=0.5cm of expectations, draw,fill=heyvlColor!60] (qivl) {\large \textbf{\textbf{Q}uantitative \textbf{I}ntermediate \textbf{V}erification \textbf{L}anguage} (\textsf{\textbf{HeyVL}})};

\node[xshift=1.75cm,minimum height=0.8cm, minimum width=3.5cm,inner sep=5pt,align=center,below=0.9cm of qivl.west, draw,fill=brLightGreen] (vcgen) { VC Generator };

\node[xshift=0.5cm,minimum height=0.8cm, minimum width=3.5cm,inner sep=5pt,align=center,right = 1.5cm of vcgen.east, draw,fill=prepostColor!60] (heylo) { Real-Valued Logic (\HeyLo) };

\node[xshift=0.75cm,minimum height=0.8cm, minimum width=3cm,inner sep=5pt,align=center,right = 1.5cm of heylo.east, draw,fill=brLightGray] (smt) { SMT Solver };
		
\draw[line width=2mm,-{Triangle[width=5mm,length=3mm]}] (expectations.south) -- (qivl);

\draw[line width=2mm,-{Triangle[width=5mm,length=3mm]}] ([xshift=-5.75cm] qivl.south) -- (vcgen);

\draw[line width=2mm,-{Triangle[width=5mm,length=3mm]}] ([xshift=-5.75cm] vcgen) -- (heylo);

\draw[line width=2mm,-{Triangle[width=5mm,length=3mm]}] ([xshift=-5.75cm] heylo) -- (smt);

\end{tikzpicture}

%% file: sections/3_heylo.tex
\section{\HeyLo: A Quantitative Assertion Language}
\label{sec:heylo}

When analyzing quantitative program properties such as runtimes, failure probabilities, or space usage, it is often more direct, more intuitive, and more practical to reason directly about \emph{values} like the runtime $n^2$, the probability~$\sfrac{1}{2^x}$, or a list's length, instead of \emph{predicates} like $\textit{rt} = n^2$, $\textit{prob} \leq \sfrac{1}{2^x}$, or $\text{length}(ls) > 0$ (cf., \cite{ngoBoundedExpectationsResource2018,kaminskiWeakestPreconditionReasoning2018}).

This section introduces \HeyLo{} -- a real-valued logic for quantitative verification of probabilistic programs, which aims to take the role that predicate logic has for classical verification.
By syntactifying real-valued functions, \HeyLo serves as 
(1) a language for specifying quantitative properties -- in particular those that~\citet{mciverAbstractionRefinementProof2005} (and many other authors) call \emph{expectations}\footnote{For historical reasons, the term \emph{expectations} refers to random variables on a program's state space.} --, and
(2)~a foundation for automation by reducing many verification problems to a decision problem for \HeyLo, e.g.\ validity or entailment checking.
To ensure that \HeyLo is expressive enough for (1), we design it reminiscently of the language by \citet{batzRelativelyCompleteVerification2021}, which is relatively complete for the verification of probabilistic programs.

To ensure that \HeyLo is suitable for (2), \HeyLo is \emph{first-order}, so as to simplify automation.
Moreover, verification problems can often be stated as inequalities between two functions.
To ensure that such inequalities can, in principle, be encoded into a \emph{single} decision problem for \HeyLo, we introduce \emph{quantitative (co)implications} -- which provide a syntax for comparing \HeyLo formulae -- and prove an analogue to the classical deduction theorem for predicate logic~\cite{kleeneIntroductionMetamathematics1952}.
Supporting comparisons between expectations via (co)implications is essential for encoding proof rules for probabilistic programs. 
The (co)implications are inspired by intuitionistic G\"odel logics \cite{baazInfinitevaluedGodelLogics1996,preiningGodelLogicsSurvey2010} and form Heyting algebras (cf.\ \Cref{thm:adjoint}), hence the name \HeyLo.

\subsection{Program States and Expectations}
Let $\Vars = \{x,y,\ldots\}$ be a countably infinite set of typed variables. 
We write $\typeof{x}{\typevar}$ to indicate that $x$ is of type $\typevar$, i.e.\ $\typevar$ is the set of values $x$ can take.
We assume the built-in types $\Bools = \{\true,\false\}$, $\Nats$, $\Ints$, $\Rats$, $\PosRats$, $\Reals$, $\PosReals$, and $\PosRealsInf = \PosReals \cup \{\infty\}$; 
 our verification infrastructure also supports user-defined mathematical types (cf. \cref{sec:domain-decl}).
  We collect all types in $\Types$ and all values in 
  $\Vals = \bigcup_{\typevar \in \Types} \typevar$. 
  A \emph{(program) state} $\State$ maps every variable $\typeof{x}{\typevar}$ to a value in $\typevar$.
The set %
of states is thus%
\begin{align*}
	\States \qeq \setcomp{
		\State\colon \Vars \to \Vals \quad{}
	}{
		\quad  \text{for all } x \in \Vars\colon \quad \typeof{x}{\typevar}\qimplies \sigma(x)\in\typevar~  
	}~.
\end{align*}%
\emph{Expectations} are the quantitative analogue to logical predicates: they map program states to $\PosRealsInf$ instead of truth values.
The complete lattice $(\Expectations,\,\expleq)$ of expectations is given by%
\begin{align*}%
	\Expectations \eeq \setcomp{ \expa }{ \expa \colon \States \to \Rplus }
	\qquad\textnormal{with}\qquad \expa \eexpleq \expb \quad\text{iff}\quad \text{for all $\State\in\States$}\colon{}~ \expa(\State) \lleq \expb(\State)~.
\end{align*}%

\subsection{Syntax of \HeyLo}
We start with the construction of \HeyLo's atoms. The set $\Terms$ of \emph{terms} is given by the grammar
\begin{align*}
    \termvar \quad{}\Coloneqq{}\quad c \morespace{\vert} x 
    \morespace{\vert}  \funcsymb(\termvar, \ldots, \termvar)~,
\end{align*}%
where $c$ is a \emph{constant} in $\Rats \cup \Bools$, $x$ is a \emph{variable} in $\Vars$, and
$\funcsymb$ is either one of the \emph{built-in function} symbols $+,\cdot,-,\monus,<,=,\wedge,\vee,\neg$ 
($\monus$ is subtraction truncated at $0$)
or a typed \emph{user-defined function} symbol $\typeof{\funcsymb}{\typevar_1 \times \ldots \times \typevar_n \to \typevar}$ for some $n\geq 0$ and types $\typevar_1,\ldots,\typevar_n,\typevar$ (cf. \cref{sec:domain-decl}).
Function symbols include, for example, the length of lists $\typeof{\listlen}{\listtype \to \Nats}$ and the exponential function $\typeof{\mathit{exp}}{\Reals \times \Ints \to \Reals}$ mapping $(r,n)$ to $r^n$.

We write $\typeof{\termvar}{\typevar}$ to indicate that term $\termvar$ is of type $\typevar$. 
Typing and subtyping of terms is standard. %
In particular, if $\typeof{\termvar}{\typevar_1}$ and $\typevar_1 \subseteq \typevar_2$, then $\typeof{\termvar}{\typevar_2}$. 
We only consider well-typed terms.

We denote terms of type $\PosRats$ (resp.\ $\Bools$) by $\aexpr$ (resp.\ $\bexpr$) and call them \emph{arithmetic} expressions (resp. \emph{Boolean expressions}).
The set of \emph{$\HeyLo$ formulae} is given by the following grammar:

\vspace{-1\baselineskip}%
\begin{halfboxl}
	\begin{align*}
		\hla \morespace{\Coloneqq}& \aexpr \explain{arithmetic expressions} \\
		\vert~& \hla + \hla \explain{addition} \\
		\vert~& \hla \sqcap \hla \explain{minimum} \\
		\vert~& \iquant{\typeof{x}{\typevar}}{\hla} \explain{infimum over $\typeof{x}{\typevar}$} \\
		\vert~& \hla \impl \hla \explain{implication} %
	\end{align*}
\end{halfboxl}%
\begin{halfboxr}
	\begin{align*}
		\vert~& \embed{\bexpr} \explain{Boolean embedding} \\
		\vert~& \hla \cdot \hla \explain{multiplication} \\
		\vert~& \hla \sqcup \hla \explain{maximum} \\
		\vert~& \squant{\typeof{x}{\typevar}}{\hla} \explain{supremum over $\typeof{x}{\typevar}$} \\
		\vert~& \hla \coimpl \hla \explain{coimplication} %
	\end{align*}
\end{halfboxr}
\vspace{-1\baselineskip}%
\noindent%
We explain the meaning of \HeyLo formulae in the next subsection. Free- and bound (by $\Sup$ or $\Inf$ quantifiers) variables of a \HeyLo formula $\hla$ are defined as usual. The order of precedence for arithmetic- and Boolean expressions is standard. For \HeyLo formulae, the order of precedence is,
\[
    \Inf, \Sup \qquad\lightgray{<}\qquad \impl, \coimpl \qquad\lightgray{<}\qquad \sqcup \qquad\lightgray{<}\qquad \sqcap \qquad\lightgray{<}\qquad + \qquad\lightgray{<}\qquad \cdot %
    \quad\qquad ,
\]
i.e.\ $\Inf$ and $\Sup$ are least binding and $\cdot$ is most binding.
We use parentheses %
to resolve ambiguities. 

\subsection{Semantics and Properties of \HeyLo}
\label{sec:heylo-sem}
\begin{figure}[t]
	\renewcommand{\arraystretch}{1.5}%
	\newcommand{\hint}[1]{\footnotesize{}#1}%
	\begin{minipage}{0.5\textwidth}
		\begin{center}
			\adjustbox{width=\textwidth}{%
			\begin{tabular}{@{\qquad}l@{\qquad}l@{\qquad}l}
				\toprule
				$\hlc$ & $\interpretsimpleState{\hlc}$ \\
				\midrule
				$\aexpr$ & $\interpretsimpleState{a}\vphantom{\ifThenElse{\interpretsimpleState{\bexpr} = \true}{\infty}{0}}$ \\
				$\hla + \hlb$ & $\interpretsimple{\hla}(\sigma) + \interpretsimple{\hlb}(\sigma)$ \\
				$\hla \sqcap \hlb$ & $\min \Set{\interpretsimple{\hla}(\sigma),~ \interpretsimple{\hlb}(\sigma)}$ \\
				$\iquant{\typeof{x}{\typevar}}{\hla}$ & $\inf \Set{ \interpretsimpleStateSubstBy{\hla}{x}{v} | v \in \typevar }$ & \\
				$\hla \impl \hlb$ & $\ifThenElse{\interpretsimpleState{\hla} \leq \interpretsimpleState{\hlb}}{\infty}{\interpretsimpleState{\hlb}}$ \\[1em]
				\bottomrule
			\end{tabular}}
		\end{center}
	\end{minipage}%
	\begin{minipage}{0.5\textwidth}
		\begin{center}
			\vspace*{-0.09\baselineskip}
			\adjustbox{width=\textwidth}{%
			\begin{tabular}{@{\qquad}l@{\qquad}l@{\qquad}l}
				\toprule
				$\hlc$ & $\interpretsimpleState{\hlc}$ \\
				\midrule
				$\embed{\bexpr}$ & $\ifThenElse{\interpretsimpleState{\bexpr} = \true}{\infty}{0}$  \\[1em]
				$\hla \cdot \hlb$ & $\interpretsimple{\hla}(\sigma) \cdot \interpretsimple{\hlb}(\sigma)$ \\
				$\hla \sqcup \hlb$ & $\max \Set{\interpretsimple{\hla}(\sigma),~ \interpretsimple{\hlb}(\sigma)}$ \\
				$\squant{\typeof{x}{\typevar}}{\hla}$ & $\sup \Set{ \interpretsimpleStateSubstBy{\hla}{x}{v} | v \in \typevar }$ & \\
				$\hla \coimpl \hlb$ & $\ifThenElse{\interpretsimpleState{\hla} \geq \interpretsimpleState{\hlb}}{0}{\interpretsimpleState{\hlb}}$ \\
				\bottomrule
			\end{tabular}}
		\end{center}
	\end{minipage}
	\caption{Semantics of \HeyLo. $\inf$ and $\sup$ are taken over $\PosRealsInf$. Here 
$\sigma\substBy{x}{v}(y) = \begin{cases}
	v, & \text{if $x=y$} \\
	\sigma(y), & \text{otherwise}~.
\end{cases}$}
	\label{fig:heylo-semantics}
\end{figure}
A term $\typeof{\termvar}{\typevar}$ evaluates to value $\interpretsimple{\termvar}(\sigma) \in \typevar$ on state $\sigma$. 
We assume the standard semantics for constants and built-in functions and that  
$\interpretsimple{\funcsymb}$ is given for all user-defined functions.

The \emph{semantics} of a \HeyLo formula $\hla$ is an expectation $\interpretsimple{\hla}\colon \States \to \PosRealsInf$ defined by induction on the structure of $\hla$ in \Cref{fig:heylo-semantics}, where we define $0 \cdot \infty = \infty \cdot 0 = 0$ as is common in measure theory.
Two \HeyLo formulae $\hla$ and $\hlb$ are \emph{equivalent}, denoted $\hla \equiv \hlb$, iff $\interpretsimple{\hla} = \interpretsimple{\hlb}$.
A \HeyLo formula%
\begin{align*}
	\hla \textnormal{ is \emph{valid}}  \qiff \hla \morespace{\equiv} \infty \qqand
	\hla \textnormal{ is \emph{covalid}}  \qiff \hla \morespace{\equiv} 0~.
\end{align*}%
For $\hla,\hlb \in \HeyLo$, we define
\[
    \underbrace{\hla \hheyloleq \hlb}_{\mathclap{\text{read: $\hla$ lower-bounds $\hlb$}}} \qqiff \underbrace{\interpretsimple{\hla} \eexpleq \interpretsimple{\hlb}}_{\mathclap{\text{pointwise inequality}}}
    \quad\qqand\quad
    \underbrace{ \hla \hheylogeq \hlb}_{\mathclap{\text{read: $\hla$ upper-bounds $\hlb$}}} \qqiff\interpretsimple{\hla} \eexpgeq \interpretsimple{\hlb}.
\]
These notions are central since we will encode verification problems as inequalities between \HeyLo formulae. In contrast to classical IVLs, \HeyLo contains constructs for \emph{both} reasoning about lower-bounds 
and for reasoning about upper bounds.
We briefly go over each construct in \Cref{fig:heylo-semantics}.

\paragraph{Arithmetic- and Boolean Expressions.}
These expressions form the atoms of \HeyLo. Consider, e.g.\ the arithmetic expressions $x+1$ for some numeric variable $x$ and $2\cdot \listlen(y)$ for a variable $\typeof{y}{\listtype}$. On state $\sigma$, $x+1$ evaluates to $\sigma(x) + 1$, and $2\cdot \listlen(y)$ evaluates to $2$ times the length of list $\sigma(y)$.

Boolean expressions $\bexpr$ are embedded in \HeyLo using the \emph{embedding operator} $\embed{\cdot}$: On state $\sigma$, $\embed{\bexpr}$ evaluates to $\infty$ (think: true, since $\infty$ is the top element in the lattice of expectations)  if $\sigma$ satisfies $\bexpr$, and to $0$ otherwise.
For instance, $\embed{x+1 = 2\cdot \listlen(y)}$ %
evaluates to $\infty$ if $\sigma(x)+1$ is equal to two times the length of the list $\sigma(y)$, and to $0$ otherwise. 

\paragraph{Addition, Multiplication, Minimum, and Maximum.}
\HeyLo formulae can be composed by standard binary arithmetic operations for sums ($+$), products ($\cdot$), minimum ($\sqcap$), and maximum ($\sqcup$).
Each of these operations is understood pointwise (with the assumption that $\infty \cdot 0 = 0$).
For instance, $\interpretsimple{\listlen(y_1) \sqcap \listlen(y_2)}(\sigma)$ is the minimum length of lists $\sigma(y_1)$ and $\sigma(y_2)$.

\paragraph{Quantifiers.}
The \emph{infimum quantifier} $\Inf $ and the \emph{supremum quantifier} $\Sup$ from \cite{batzRelativelyCompleteVerification2021} are the quantitative analogues of the universal $\forall$ and the existential $\exists$ quantifier from predicate logic. 
Intuitively, the $\Inf$ quantifier minimizes a quantity, just like the $\forall$ quantifier minimizes a predicate's truth value.
Dually, the $\Sup$ quantifier maximizes a quantity just like $\exists$ maximizes a predicate's truth value. %
The quantitative quantifiers embed %
$\forall$ and $\exists$ in \HeyLo, i.e.\ %
for $\typeof{\bexpr}{\Bools}$ and $\sigma \in \States$,
\[
    \interpretsimple{\iquant{\typeof{x}{\typevar}}{\embed{\bexpr}}}(\sigma) = 
    \begin{cases}
    	\infty, &\text{if}~\sigma \models \forall \quant{\typeof{x}{\typevar}}{\bexpr} \\
    	0, &\text{otherwise}
    \end{cases}
    \qquad
    \text{and}
    \qquad
    \interpretsimple{\squant{\typeof{x}{\typevar}}{\embed{\bexpr}}}(\sigma) =  
    \begin{cases}
    	\infty, &\text{if}~\sigma \models \exists \quant{\typeof{x}{\typevar}}{\bexpr} \\
    	0, &\text{otherwise}
    \end{cases}
\]
Here, $\models$ denotes the standard satisfaction relation of first-order logic. The above construction extends canonically to nested quantifiers, e.g.\ $\exists \quant{\typeof{x}{\typevar}}{\forall \quant{\typeof{y}{\typevar'}}{\bexpr}}$ corresponds to $\squant{\typeof{x}{\typevar}}{\iquant{\typeof{y}{\typevar'}}{\embed{\bexpr}}}$.

For a quantitative example, consider the formula $\hla = \squant{\typeof{x}{\PosRats}}{\embed{x\cdot x <2} \sqcap x}$. On state $\sigma$, the subformula $\embed{x\cdot x <2} \sqcap x$ evaluates to $\sigma(x)$ if $\sigma(x) \cdot \sigma(x) <2$, and to $0$ otherwise. Consequently,
\[
    \interpretsimple{\hla}(\sigma) \eeq \sup \set{r \in \PosRats ~\mid~ r\cdot r < 2} \eeq \sqrt{2}~.
\]
Notice that $\interpretsimple{\hla}(\sigma)$ is \emph{irrational} even though all constituents of $\hla$ are rational-valued. 
It has been shown in \cite{batzRelativelyCompleteVerification2021} that --- similar to our above construction of $\sqrt{2}$ --- the quantitative quantifiers combined with arithmetic- and (embedded) Boolean expressions over $\PosRats$ enable the construction of \emph{all} expected values emerging from discrete probabilistic programs.

\paragraph{(Co)implication}
\begin{figure}[t]
	\begin{minipage}{.45\linewidth}
		\centering
		\begin{tikzpicture}[scale=0.8]
			\begin{axis}[
				small,
				xtick={0,5,10},
				xlabel=$\sigma(x)$,
				ytick={0,5,10,12},
				xmin=0,
				xmax=10,
				yticklabels={$0$,$5$,$10$,$\infty$},
				ylabel=$\interpretsimple{5 \impl x}(\sigma)$,
				ymin=0,
				ymax=12,
				domain=0:10,
				width=5cm,
				axis lines=left,
				axis line style={-}, %
				axis line shift=10pt,
				scatter,
				scatter src=explicit symbolic,
				scatter/classes={
					o={mark=*,fill=white},
					c={mark=*}
				},
				]
				\addplot[mark=none,thick] table[meta index=2] {
					A B i
					0 0 c
					5 5 o
					
					5 12 c
					10 12 c
				};
				
				\draw[dashed,black] (0,5) -- (10,5);
				
				\draw[dashed,white] (axis cs: 0,10)--(axis cs: 0,11);
			\end{axis}
		\end{tikzpicture}
		\caption{$\interpretsimple{5 \impl x}(\sigma)$ for $\sigma(x) \in [0,10]$.}
		\label{fig:assume}
	\end{minipage}\hfill%
	\begin{minipage}{.45\linewidth}
		\centering
		\begin{tikzpicture}[scale=0.8]
			\begin{axis}[
				small,
				xtick={0,5,10},
				xlabel=$\sigma(x)$,
				ytick={0,5,10,12},
				xmin=0,
				xmax=10,
				yticklabels={$0$,$5$,$10$,$\infty$},
				ylabel=$\interpretsimple{5 \coimpl x}(\sigma)$,
				ymin=0,
				ymax=12,
				domain=0:10,
				width=5cm,
				axis lines=left,
				axis line style={-}, %
				axis line shift=10pt,
				scatter,
				scatter src=explicit symbolic,
				scatter/classes={
					o={mark=*,fill=white},
					c={mark=*}
				},
				]
				\addplot[mark=none,thick] table[meta index=2] {
					A B i
					0 0 c
					5 0 c
					
					5 5 o 
					10 10 c
				};
				
			\draw[dashed,black] (0,5) -- (10,5);
			\end{axis}
		\end{tikzpicture}
		\caption{$\interpretsimple{5 \coimpl x}(\sigma)$ for $\sigma(x) \in [0,10]$.}
		\label{fig:coassume}
	\end{minipage}
\end{figure}
$\impl$ and $\coimpl$ generalize Boolean implication and converse nonimplication.\footnote{The converse nonimplication of propositions $P$ and $Q$ is defined as $\neg(P \leftarrow Q)$ and is to be read as \enquote{$Q$ does \emph{not imply} $P$}.}
For state $\sigma$, the \emph{implication} $\hla \impl \hlb$ evaluates to $\infty$ if $\interpretsimple{\hla}(\sigma) \leq \interpretsimple{\hlb}(\sigma)$, and to $\interpretsimple{\hlb}(\sigma)$ otherwise. Dually, the \emph{coimplication} $\hla \coimpl \hlb$ evaluates to $0$ if $\interpretsimple{\hla}(\sigma) \geq \interpretsimple{\hlb}(\sigma)$, and to $\interpretsimple{\hlb}(\sigma)$ otherwise. 
    
	To gain some intuition,
	we first note that the top element $\infty$ of our quantitative domain $\PosRealsInf$ can be viewed as \enquote{entirely true} (i.e.\ as true as it can possibly get) and $0$ can be viewed as \enquote{entirely false} (i.e.\ as false as it can possibly get).
	The implication $\hla \impl \hlb$ makes $\hlb$ \emph{more true} by \emph{\underline{lowerin}g the threshold above which $\hlb$ is considered \underline{entirel}y\underline{ true}} -- and thus $\infty$ -- to $\hla$.
	In other words: Anything that is at least as true as $\hla$ is considered entirely true. Anything less true than $\hla$ remains as true as $\hlb$.
	\Cref{fig:assume} illustrates this for the formula $5 \impl x$.

	As another example, $x^2 \impl x$ evaluates to $\infty$ for states $\sigma$ with $\sigma(x) \in [0,1]$; otherwise, $x$ is below the threshold $x^2$ at which $x$ is considered entirely true and thus the implication evaluates to $x$.
	
	The intuition underlying the coimplication is dual: 
	$\hla \coimpl \hlb$ makes $\hlb$ \emph{less true} by \emph{\underline{raisin}g the threshold below which $\hlb$ is considered \underline{entirel}y\underline{ }f\underline{alse}}
	-- and thus $0$ -- to $\hla$.
	In other words: Anything that is not more true than $\hla$ is considered entirely false. Anything that is more true than $\hla$ remains as true as $\hlb$.
	\Cref{fig:coassume} illustrates this for the formula $5\coimpl x$. %

Chained implications can also be understood in terms of lowering thresholds: $\hla \impl (\hlb \impl \hlc)$ lowers the threshold at which $\hlc$ is considered entirely true to $\hla$ \emph{and} $\hlb$, whichever is lower. 
Formally, $\hla \impl (\hlb \impl \hlc)$ is equivalent to $(\hla \sqcap \hlb) \impl \hlc$.
More generally, (co)implications are the adjoints of the minimum $\sqcap$ and maximum $\sqcup$:%
\begin{restatable}[Adjointness Properties]{theorem}{thmHeyloAdjoint}
	\label{thm:adjoint}
	For all \HeyLo formulae $\hla$, $\hlb$, and $\hlc$, we have%
	\begin{align*}
		\hla \sqcap \hlb \hheyloleq \hlc&\qiff \hla \hheyloleq \hlb \impl \hlc \qquad\text{and}\qquad %
		\hlb \sqcup \hlc\ \hheylogeq \hla \qiff \hlc \hheylogeq \hlb \coimpl \hla ~.
	\end{align*}%
\end{restatable}%
\noindent%
Both $\impl$ and $\coimpl$ are backward compatible with Boolean implication and converse nonimplication:
\[
\interpretsimple{\embed{\bexpr_1} \impl \embed{\bexpr_2}}(\sigma) = 
\begin{cases}
	\infty, &\text{if}~\sigma \models \bexpr_1 \impl \bexpr_2 \\
	0, &\text{otherwise}
\end{cases}
\qquad
\interpretsimple{\embed{\bexpr_1} \coimpl \embed{\bexpr_2}}(\sigma) = 
\begin{cases}
	\infty, &\text{if}~\sigma \models \neg(\bexpr_1 \leftarrow \bexpr_2) \\
	0, &\text{otherwise}
\end{cases}
\]
We will primarily use (co)implications to (1) incorporate the capability of \emph{comparing} expectations syntactically in \HeyLo and to (2) express \emph{assumptions}.
Application (1) is justified by the following quantitative version of the well-known deduction theorem\footnote{We mean the deduction theorem that relates semantical entailment $\models$ with the material conditional $\rightarrow$. Another theorem also known as \emph{deduction theorem} relates syntactical entailment (i.e.\ provability) $\vdash$ with the material conditional $\rightarrow$.} from first-order logic \cite{kleeneIntroductionMetamathematics1952}:
\begin{restatable}[\HeyLo Deduction Theorem]{theorem}{thmHeyloDeduction}
	\label{thm:deduction-theorem}
	For all \HeyLo formulae $\hla$ and $\hlb$, we have%
	\begin{align*}
		\hla \hheyloleq \hlb \qiff \hla \impl \hlb ~\textnormal{is valid} 
		\qquad\textnormal{and}\qquad 
		\hla \hheylogeq \hlb \qiff \hla \coimpl \hlb ~\textnormal{is covalid}~.
	\end{align*}%
\end{restatable}%
\noindent%
\ifdefined\WITHAPPENDIX
The proof is in \cref{sec:app-proofs-heylo}.
\fi
For application (2), consider the implication $\embed{\bexpr} \impl \hlb$; it evaluates to $\hlb$ whenever $\bexpr$ holds, and to $\infty$ otherwise. As in predicate logic, the implication can be read as \emph{assuming} $\bexpr$ holds before evaluating $\hlb$. 
Formally, 
\[
  \interpretsimple{\embed{\bexpr} \impl \hlb}(\sigma) \eeq
  \begin{cases}
  	   \interpretsimple{\hlb}(\sigma), & \text{if}~\sigma \models \bexpr \\

  	   \infty, & \text{otherwise}~. \\
  \end{cases}
\]
Now, consider the inequality $\hla \heyloleq \embed{\bexpr} \impl \hlb$. 
For all states $\sigma$ \emph{not} satisfying $\bexpr$ (i.e.\ the set of states that we do \emph{not} assume), the inequality vacuously holds. 
For all other states (i.e.\ those states that we actually assume), $\hla$ must lower-bound $\hlb$ in order for the inequality to hold.
\begin{example}
	\label{ex:heylo_ifthenelse}
	Let $\hla,\,\hlb \in \HeyLo$ and $\typeof{\bexpr}{\Bools}$. We construct a \HeyLo formula $\hlc$ that, on state $\sigma$, evaluates to $\hla$ if $\sigma \models \bexpr$, and to $\hlb$ otherwise. For that, we use the Boolean embedding \mbox{and the implication:}
	\begin{align*}
		\hlc 
		\qqeq 
		\underbrace{
			(\embed{\bexpr} \impl \hla)
		}_{
			\mathclap{
				\text{if $\bexpr$ holds, evaluate to $\hla$}
			}
		} 
		\qquad
		\underbrace{\vphantom{(\embed{\bexpr} \impl \hla)}{\sqcap}}_{\text{and}}
		\qquad
		\underbrace{
			(\embed{\neg\bexpr} \impl \hlb)
		}_{
			\mathclap{
				\text{if $\neg\bexpr$ holds, evaluate to $\hlb$}
			}
		}
	\end{align*}%
\end{example}

\noindent%
To encode assumptions using the coimplication $\coimpl$, we first introduce Boolean \emph{co}-embeddings
\[
    \interpretsimple{\coembed{\bexpr}} \eeq \interpretsimple{\embed{\neg\bexpr}}
    \eeq \lambda \sigma. \ifThenElseDot{\interpretsimple{\bexpr}(\sigma) = \true}{0}{\infty}
\]
We then obtain a dual construction using $\coimpl$ for encoding assumptions: By \Cref{thm:adjoint}, we have %
\[
   \hla \heylogeq \coembed{\bexpr} \coimpl \hlb 
   \qquad\text{iff}\qquad
   \text{for all $\sigma \in \States$}\colon \interpretsimple{\bexpr}(\sigma) = \true~\text{implies}~\interpretsimple{\hla}(\sigma) \geq \interpretsimple{\hlb}(\sigma)~,
\]
i.e.\ the coimplication $\coembed{\bexpr} \coimpl \hlb $ ensures that it suffices to reason about states satisfying $\bexpr$.

\subsection{Qualitative Reasoning in \HeyLo}
\label{subsec:syntactic_sugar}
The verification of probabilistic programs comprises both quantitative \emph{and} qualitative reasoning. Whereas questions like \enquote{what is the expected value of program variable $x$ upon termination} are inherently quantitative, questions like \enquote{does $x$ increase in expectation after one loop iteration?} 
are qualitative. \HeyLo marries quantitative and qualitative reasoning.
 To shift to a qualitative statement, we first consider the %
 \emph{negation} $\neg \hla$ and \emph{conegation} $\coneg \hla$ of $\hla$
 obtained from our (co)implications:
\[
\interpretsimple{\neg \hla} \eeq \interpretsimple{\hla \impl 0} \eeq \lambda \sigma. \ifThenElse{\interpretsimpleState{\hla} = 0}{\infty}{0}%
\qquad
\interpretsimple{\coneg \hla} \eeq \interpretsimple{\hla \coimpl \infty} \eeq \lambda \sigma. \ifThenElseDot{\interpretsimpleState{\hla} = \infty}{0}{\infty}
\]
The (co)negation always evaluates to either $\infty$, the top element of $\PosRealsInf$ (entirely true), or $0$, the bottom element of $\PosRealsInf$ (entirely false).
By applying a (co)negation twice, we turn an arbitrary expectation into a  qualitative statement. 
Formally,
we define the \emph{(pointwise) validation} $\heylovalidate{\hla}$ and \emph{(pointwise) covalidation} $\heylocovalidate{\hla}$ by\footnote{In G\"odel logics, these are also called \emph{projection modalities} \cite{baazInfinitevaluedGodelLogics1996}.}
\[
   \interpretsimple{\heylovalidate{\hla}} \eeq \interpretsimple{\coneg\coneg \hla} \eeq \lambda \sigma. \ifThenElse{\interpretsimpleState{\hla} = \infty}{\infty}{0}
   \quad\text{and}\quad
   \interpretsimple{\heylocovalidate{\hla}} \eeq \interpretsimple{\neg\neg \hla} \eeq \lambda \sigma. \ifThenElseDot{\interpretsimpleState{\hla} = 0}{0}{\infty}
\]
In words, the validation $\heylovalidate{\hla}$ is (pointwise) entirely true whenever $\hla$ is entirely true, and entirely false otherwise. Dually, $\heylocovalidate{\hla}$ is entirely false whenever $\hla$ is entirely false, and entirely true otherwise. Thus, both validations and covalidations \enquote{boolify} \HeyLo formulae.
The difference is that validations pull intermediate truth values down to entire falsehood whereas covalidations lift intermediate truth values up to entire truth.

Turning expectations into qualitative statements 
has an important application, which often arises when encoding verification problems: Suppose we are given two formulae $\hla,\hlb$ with free variables $y_1,\ldots,y_n$. Moreover, our goal is to construct a \HeyLo formula $\hlc$ that evaluates to $x$ of type $\PosRats$ %
if $\hla \heyloleq \hlb$, and to $0$ otherwise. For that, we first construct the formula $\iquant{y_1,\ldots,y_n}{\heylovalidate{\hla \impl \hlb}}$. Due to the infimum quantifier over all free variables, this formula is \emph{equivalent} to $\infty$ if $\hla \heyloleq \hlb$, and \emph{equivalent} to $0$ otherwise. Hence, we construct $\hlc$ as 
\[
   \underbrace{\big(\iquant{\varandtype{y_1}{\typevar_1},\ldots,\varandtype{y_n}{\typevar_n}}{\heylovalidate{\hla \impl \hlb}}\big)}_{\text{evaluate to $0$ if $\hla\not\heyloleq\hlb$}} 
   \hspace{0.8cm}
   \underbrace{\vphantom{(\iquant{y_1,\ldots,y_n}{\heylovalidate{\hla \impl \hlb}})}\sqcap}_{\text{and}} 
   \underbrace{\vphantom{(\iquant{y_1,\ldots,y_n}{\heylovalidate{\hla \impl \hlb}})}x}_{\text{evaluate to $x$ otherwise}}~.
\]
Moreover, we obtain a dual construction using $\coimpl$ and the supremum quantifier:
\[
\underbrace{\big(\squant{\varandtype{y_1}{\typevar_1},\ldots,\varandtype{y_n}{\typevar_n}}{\heylocovalidate{\hla \coimpl \hlb}}\big)}_{\text{evaluate to $\infty$ if $\hla\not\heylogeq\hlb$}} 
\hspace{0.8cm}
\underbrace{\vphantom{(\squant{y_1,\ldots,y_n}{\heylocovalidate{\hla \coimpl \hlb}})}\sqcup}_{\text{and}} 
\underbrace{\vphantom{(\squant{y_1,\ldots,y_n}{\heylocovalidate{\hla \coimpl \hlb}})}x}_{\text{evaluate to $x$ otherwise}}
\]

%% file: sections/4_heyvl.tex
\section{\textnormal{\textbf{\protect\HeyVL}}: A Quantitative Intermediate Verification Language}
\label{sec:heyvl}

Many verification problems for probabilistic programs reduce naturally to checking inequalities between \HeyLo formulae.\footnote{Or equivalently by \Cref{thm:deduction-theorem}: Checking (co)validity, i.e.\ whether a \HeyLo formula is equivalent to $\infty$ (resp. $0$).}
Consider, for instance, the program%
\begin{align*}
	\stmtRasgn{y}{
		\color{heyvlColor}\pexp{\sfrac{1}{2}}{x} \pexpand \pexp{\sfrac{1}{2}}{x+1}
	}~,
\end{align*}%
which sets $y$ either to $x$ or to $x + 1$, depending on the outcome of a fair coin flip.
Suppose we want to verify that $x + \tfrac{1}{2}$ is a \emph{lower} bound on the expected value of $y$ after executing above program.
According to \citet{mciverAbstractionRefinementProof2005},
verifying this %
bound amounts to proving the inequality%
\begin{align*}
	\underbrace{
		\colheylo{x + \tfrac{1}{2}}
	}_{
		\mathclap{\text{proposed lower bound}}
	}
	\qquad
	{\heyloleq}
	\qquad
	\tfrac{1}{2} \cdot x + \tfrac{1}{2} \cdot (x + 1)
	\qquad{\triangleq}\qquad
	\underbrace{
		\wp{\stmtRasgn{y}{\color{heyvlColor}\pexp{\sfrac{1}{2}}{x} \pexpand \pexp{\sfrac{1}{2}}{x+1}}}(\colheylo{y})
	}_{
		\mathclap{\text{expected outcome of {\color{heyvlColor}$x$ + fair coin flip} stored in $\colheylo{y}$}}
	}~,
	\tag{ex}
\end{align*}
where the weakest preexpectation $\wp{C}(f)$ is a function (which we can represent as a \HeyLo formula) that maps every initial state $\sigma$ to the expected value of $f$ after executing the program $C$ on input $\sigma$.
Our goal is to simplify writing, composing, and reasoning \emph{modularly} about such expected values and similar quantities.
To this end, we propose \HeyVL, a novel intermediate verification language for modeling quantitative verification problems.

\HeyVL \emph{programs} are organized as a \emph{collection of procedures}.
Each procedure~$P$ is equipped with a body $S$ and a specification.
The body $S$ is a \HeyVL \emph{statement} and can for now be thought of as a more or less ordinary probabilistic program.\footnote{There are verification-specific statements which can be part of the procedure body which we will describe later.}
The specification of a procedure comprises a \emph{pre}~$\hla$ and a \emph{post}~$\hlb$, both \HeyLo formulae.
Intuitively, a procedure $P$ \emph{verifies} if its body $S$ adheres to $P$'s specification, meaning essentially that the inequality $\hla \heyloleq \wp{S}(\hlb)$ holds, i.e.\ the expected value of $\hlb$ after executing $S$ is lower-bounded by $\hla$.
This inequality will be called the \emph{verification condition}~of~$P$.
An entire \HeyVL program \emph{verifies} if all of its procedures verify.

How do we describe the verification problem (ex) in \HeyVL?
As shown in \Cref{fig:heyvl:ex-procedure}, we write a single procedure $P$ with body $\stmtRasgn{y}{\color{heyvlColor}\pexp{\sfrac{1}{2}}{x} \pexpand \pexp{\sfrac{1}{2}}{x+1}}$, pre $\colheylo{x + \tfrac{1}{2}}$, and post~$\colheylo{y}$.
This gives rise to the verification condition $\colheylo{x + \tfrac{1}{2}} \heyloleq \wp{\stmtRasgn{y}{\color{heyvlColor}\pexp{\sfrac{1}{2}}{x} \pexpand \pexp{\sfrac{1}{2}}{x+1}}}(\colheylo{y})$, which is precisely the inequality (ex) we aim to verify.
The \HeyVL program (i.e.~the single procedure $P$) verifies if and only if we have positively answered the verification problem (ex).

To encode more complex verification problems or proof rules, one may need to write more than one \HeyVL procedure. For example, in \Cref{sec:encodings:loop-free}, we will encode a proof rule for conditional expected values that requires establishing a lower \emph{and} a different upper bound. The latter can be described using a second \HeyVL procedure, see \Cref{sec:heyvl:structure}.
Furthermore, it is natural to break down large programs and/or complex proof rules into smaller (possibly mutually recursive) procedures, which can be verified modularly based on the truth of their verification conditions.

\begin{figure}[t]
		\begin{align*}
			&\proc{ex}{\varandtype{x}{UInt}}{\varandtype{y}{UInt}} 
			~~\intersem{\text{procedure that takes $x$ as input and returns the value of $y$}} 
			\\
			&\qquad\Requires{x + \sfrac{1}{2}}
			~~\intersem{\text{lower bound on the expected value of ${\colheylo{y}}$ after termination of the body}} \\
			&\qquad\Ensures{y} 
			~~\intersem{\text{quantity of interest evaluated in final states}} \\
			&\blockStart \\
			&\qquad {
				\color{heyvlColor}
				\stmtRasgn{y}{
					\pexp{\sfrac{1}{2}}{x} \pexpand \pexp{\sfrac{1}{2}}{x+1} 
				}
			} 
			~~\intersem{\text{returns the sum of $x$ plus outcome of a fair coin flip}}
			\\
			& \blockEnd
		\end{align*}%
		\caption{A \HeyVL procedure whose verification condition is equation (ex).}%
		\label{fig:heyvl:ex-procedure}%
\end{figure}%

\subsection{\HeyVL Procedures}\label{sec:heyvl:structure} 
A \HeyVL procedure consists of a name, a list of (typed) input and output variables, a body, and a quantitative specification.
Syntactically, a \HeyVL procedure is of the form
\begin{align*}
  &\proc{\procname}{\overline{\varandtype{\varin}{\typevar}}}{\overline{\varandtype{\varout}{\typevar}}}
  && \intersem{\text{procedure name $\procname$ with read-only inputs $\overline{\varin}$ and outputs $\overline{\varout}$}}
   \\
  &\qquad\Requires{\hla} && \intersem{\text{pre: \HeyLo formula over inputs }} \\
  &\qquad\Ensures{\hlb} && \intersem{\text{post: \HeyLo formula over inputs or outputs}} \\
  &\blockStart~\sstmt~\blockEnd && \intersem{\text{procedure body}}
\end{align*}%
where $\procname$ is the procedure's name, 
$\overline{\varin}$ and $\overline{\varout}$ are (possibly empty and pairwise distinct) lists of typed program variables called the \emph{inputs} and \emph{outputs} of $\procname$.
The specification is given by a \emph{pre} $\hla$ which is a \HeyLo formula over variables in $\overline{\varin}$ and a \emph{post} $\hlb$ which is also a \HeyLo formula but ranging over variables in $\overline{\varin}$ or $\overline{\varout}$.
The \emph{procedure body} $\sstmt$ is a \HeyVL statement, whose syntax and semantics will be formalized in \Cref{sec:heyvl-blocks-syntax,sec:heyvl-blocks-semantics}.

As mentioned above, the procedure $P$ gives rise to a verification condition, namely $\hla \heyloleq \wp{S}(\hlb)$.
However, this is only accurate if $S$ is an ordinary %
probabilistic program.
As our statements~$S$ may also contain non-executable\footnote{But expected value changing.} verification-specific \texttt{assume} and \texttt{assert} commands, the \emph{verification condition generated by $P$} is actually%
\begin{align*}
	\colheylo{\hla} \quad{\heyloleq}\quad \vc{\colheyvl{S}}(\colheylo{\hlb})~,
\end{align*}%
where $\symVc$ is the \emph{verification preexpectation transformer} that
 extends the aforementioned weakest preexpectation $\symWp$ by semantics for the verification-specific statements, see \Cref{sec:heyvl-blocks-semantics}.
For procedure calls, we approximate the weakest preexpectation based on the callee's specification to enable modular verification, see \Cref{sec:heyvl_proc_calls}.

Readers familiar with classical Boolean deductive verification may think of the verification condition~\mbox{$\colheylo{\hla} \heyloleq \vc{\colheyvl{S}}(\colheylo{\hlb})$} as a \emph{quantitative Hoare triple} $\triple{\colheylo{\hla}}{\colheyvl{S}}{\colheylo{\hlb}}$, where $\heyloleq$ takes the quantitative role of the Boolean ${\Longrightarrow}$, i.e.\ we have%
\[
   \triple{\hhla}{\sstmt}{\hhlb} ~\text{is valid} \quad\qiff\quad \hhla \quad\heyloleq\quad \vc{\sstmt}(\hhlb).
\]
Indeed, if $\hhla$ and $\hhlb$ are ordinary Boolean predicates and $\sstmt$ is a non-recursive non-probabilistic program, then $\triple{\hhla}{\sstmt}{\hhlb}$ is a standard Hoare triple: whenever state $\State$ satisfies precondition $\hhla$, then procedure body $\sstmt$ must successfully terminate on $\State$ in a state satisfying postcondition $\hhlb$.

Phrased differently: for every initial state~$\State$, the truth value $\hhla(\State)$ lower-bounds the \emph{anticipated} truth value (evaluated in $\State$) of postcondition $\hhlb$ after termination of $\sstmt$ on $\State$.
For arbitrary \HeyLo formulae $\hhla, \hhlb$ and probabilistic procedure bodies $\sstmt$, the second view generalizes to quantitative reasoning à la~\citet{mciverAbstractionRefinementProof2005}:
The quantitative triple $\triple{\hhla}{\sstmt}{\hhlb}$ is valid iff
the pre $\hhla$ lower-bounds the \emph{expected value} (evaluated in initial states) of the post $\hhlb$ after termination of $\sstmt$.
In \Cref{sec:heyvl_proc_calls}, we will describe how \emph{calling} a (verified) procedure $P$ can be thought of as \enquote{invoking} the validity of the quantitative Hoare triple that is given by $P$'s specification.

Notice that the above inequality is our definition of validity of a quantitative Hoare triple and we do not provide an operational definition of validity. 
This is due to a lack of an intuitive operational semantics for quantitative \texttt{assume} and \texttt{assert} statements (cf.~also \Cref{sec:conclusion}).

\paragraph{Examples.}
Besides \Cref{fig:heyvl:ex-procedure},
\Cref{fig:heyvl:n-dice,fig:heyvl:rabin} further illustrate how \HeyVL procedures specify quantitative program properties; we omit concrete procedure bodies $\sstmt$ to focus on the specification.
The procedure in \Cref{fig:heyvl:n-dice} specifies that the expected value of output $\colheylo{r}$ must be at least $\colheylo{3.5 \cdot n}$ -- a property satisfied by any statement $\sstmt$ that rolls $n$ fair dice.
The procedure in \Cref{fig:heyvl:rabin} specifies that the expected value of output $ok$ being true after termination of $\sstmt$, i.e.\ the probability that the returned value $ok$ will be true, is at least $\sfrac{2}{3}$ whenever input $i$ is greater than one -- a key property of Rabin's randomized mutual exclusion algorithm~\cite{kushilevitzRandomizedMutualExclusion1992} from \Cref{fig:intro:rabin} and discussed in the introduction.
Since we aim to reason about probabilities, we ensure that the post is one-bounded by considering $1 \sqcap \embed{ok}$ instead of $ \embed{ok}$.

\begin{figure}
	\begin{minipage}{0.5\textwidth}
		\begin{align*}
   		  &\proc{n\_dice}{\varandtype{n}{UInt}}{\varandtype{r}{UReal}} \\
   		  &\qquad\Requires{3.5 \cdot n} \\
  		  &\qquad\Ensures{r} \\
  		  &\blockStart~\sstmt~\blockEnd
		\end{align*}%
		\vspace{-1.5em}
	    \caption{Expected sum of rolling $n$ fair dice.}
	    \label{fig:heyvl:n-dice}
	\end{minipage}%
	\begin{minipage}{0.5\textwidth}
		\begin{align*}
   		  &\proc{rabin}{\varandtype{i}{UInt}}{\varandtype{ok}{Bool}} \\
   		  &\qquad\Requires{\nicefrac{2}{3} \sqcap \embed{1 < i}} \\
  		  &\qquad\Ensures{1 \sqcap \embed{ok}} \\
  		  &\blockStart~\sstmt~\blockEnd
		\end{align*}%
		\vspace{-1.5em}
	    \caption{Rabin's mutual exclusion property.}
	    \label{fig:heyvl:rabin}
	\end{minipage}
\end{figure}
\paragraph{Coprocedures -- Duals to Procedures.}
Proving \emph{upper} bounds is often relevant for quantitative verification, e.g.\ when analyzing expected runtimes of randomized algorithms (cf. ~\cite{kaminskiWeakestPreconditionReasoning2018}).
\HeyVL also supports \emph{co}procedures which give rise to the dual verification condition $\hhla \heylogeq \vc{\sstmt}(\hhlb)$.\footnote{Notice $\heylogeq$ for coprocedures as opposed to $\heyloleq$ for procedures.}
The syntax of coprocedures is analogous to \HeyVL procedures; the only difference is the keyword $\symcoProc$ instead of $\symProc$.
For example, a coprocedure which was defined as in %
\Cref{fig:heyvl:n-dice} (except for replacing $\symProc$ by $\symcoProc$)
would specify that the expected value of output $\colheylo{r}$ must be \emph{at most} $\colheylo{3.5 \cdot n}$.
We demonstrate in \Cref{sec:encodings} that intricate verification techniques for probabilistic programs may require lower \emph{and} upper bound reasoning, i.e.\ \HeyVL programs that are collections of both procedures and coprocedures.
\paragraph{\HeyVL Programs.}
To summarize, a \HeyVL \emph{program} is a list of procedures and coprocedures that each give rise to a verification condition, i.e.\ a \HeyLo inequality. 
We say that a \HeyVL program \emph{verifies} iff all verification conditions of its (co)procedures hold.
\paragraph{Design Decisions.}
Since \HeyVL is an intermediate language, we favor simplicity over convenience.
In particular, we require procedure inputs to be read-only, i.e.\ evaluate to the same values in initial and final states.
Moreover, \HeyVL has no loops and no global variables.
All variables that can possibly be modified by a procedure call are given by its outputs.
All of the above restrictions can be lifted by high-level languages that encode to \HeyVL.

\subsection{Syntax of \HeyVL Statements}\label{sec:heyvl-blocks-syntax}
\HeyVL statements, which appear in procedure bodies,
provide a programming-language-style to express and approximate \emph{expected values} arising in the verification of %
probabilistic programs, including expected outcomes of program variables, reachability probabilities such as the probability of termination, and expected rewards. %
\HeyVL statements consist of %
(a) \emph{standard constructs} such as assignments, sampling from discrete probability distributions, sequencing, and nondeterministic branching, and (b) \emph{verification-specific constructs} for modeling rewards %
such as runtime, quantitative assertions and assumptions, and for forgetting values of program variables in the current state.

The syntax of \HeyVL statements $\stmt$ is given by the grammar%
\begin{center}%
\vspace*{-.5\baselineskip}%
\begin{minipage}[t]{0.3\textwidth}%
	\abovedisplayskip=0pt%
	\begin{align*}
		\stmt \morespace{\Coloneqq}&  \stmtDeclInit{x}{\typevar}{\mu} \\
		\vert~& \stmtAsgn{x_1,\ldots,x_n}{P(e_1,\ldots,e_m)} \\
		\vert~& \stmtTick{\aexpr} \\
		\vert~& \stmtSeq{\stmt}{\stmt} \\
	\end{align*}%
\end{minipage}\hfil%
\begin{minipage}[t]{0.3\textwidth}%
	\abovedisplayskip=0pt%
	\begin{align*}
		\vert~& \stmtDemonic{\stmt}{\stmt} \\
		\vert~& \Assert{\hlb} \\
		\vert~& \Assume{\hlb} \\
		\vert~& \Havoc{x} \\
		\vert~& \Validate
	\end{align*}%
\end{minipage}\hfil%
\begin{minipage}[t]{0.3\textwidth}%
	\abovedisplayskip=0pt%
	\begin{align*}
		\vert~& \stmtAngelic{\stmt}{\stmt} \\
		\vert~& \coAssert{\hlb} \\
		\vert~& \coAssume{\hlb} \\
		\vert~& \coHavoc{x} \\
		\vert~& \coValidate~,
	\end{align*}%
\end{minipage}%
\end{center}%
\medskip%
where $x \in \Vars$ is of type $\typevar$,  $a$ is an arithmetic expression, %
and $\hlb$ is a \HeyLo formula. %
Moreover, $\mu$ is a \emph{distribution expression} of type $\typevar$\footnote{$\mu$ can be instantiated with more general distribution expressions as long as the $\symVc$ semantics (cf.\ \Cref{sec:heyvl-blocks-semantics}) is computable.}
\[
\mu \eeq \pexp{p_1}{\termvar_1} \pexpand \ldots \pexpand \pexp{p_n}{\termvar_n} 
\]
with $n \geq 1$, where each $p_i$ is a term of type $[0,1]$, each $\termvar_i$ is a term of type $\typevar$, and $\sum_{i=1}^n\interpretsimpleState{p_i} = 1$ for every state $\State$. A distribution expression $\mu$ represents finite-support probability distributions, which assign probability $p_i$ to each $\termvar_i$. We often write $\exprFlip{p}$ instead of $\pexp{p}{\true} \pexpand \pexp{(1-p)}{\false}$.

We briefly go over the above constructs.
$\stmtDeclInit{x}{\typevar}{\mu}$ is a \emph{probabilistic assignment} which assigns to variable $x$ a value \emph{sampled} from the probability distribution described by $\mu$. 
The statement $\stmtAsgn{x_1,\ldots,x_n}{P(e_1,\ldots,e_m)}$ is a (co)procedure call.
We can think of it as passing the parameters~\mbox{$e_1,\ldots,e_m$} to (co)procedure $P$, executing $P$'s body, and assigning the return values to variables~\mbox{$x_1,\ldots,x_n$}.
The statement $\stmtTick{a}$ collects/accumulates/adds a reward of $a$, modeling e.g.~progression in (run)time or resource consumption. 
$\stmtSeq{\stmt_1}{\stmt_2}$ puts $\HeyVL$ statements in sequence. 
$\stmtCdot{\stmt_1}{\stmt_2}$ is a \emph{nondeterministic} choice between $\stmt_1$ and $\stmt_2$, where $\cdot$ determines whether the nondeterminism is resolved in a minimizing ($\sqcap$) or maximizing ($\sqcup$) manner.
$\Assert{\hlb}$ and $\Assume{\hlb}$ are quantitative generalizations of assertions and assumptions from classical IVLs.
$\coAssert{\hlb}$ and $\coAssume{\hlb}$ are novel statements that enable reasoning about upper bounds; there is yet no analogue in classical verification infrastructures.

 $\Havoc{x}$ and $\coHavoc{x}$ forget the current value of $x$ by branching nondeterministically over all possible values of $x$ either in a minimizing ($\Havoc{x}$) or maximizing ($\coHavoc{x}$) manner.  
 Finally, $\Validate$ and $\coValidate$ turn \emph{quantitative} expectations into \emph{qualitative} expressions, much in the flavor of validation and covalidation described earlier (see \Cref{subsec:syntactic_sugar}).

\paragraph{Declarations and Types.}
We assume that all local variables (those that are neither inputs nor outputs) are initialized by an assignment before they are used; those assignments also declare the variables' types.
If we assign to an already initialized variable, we often write $\stmtRasgn{x}{\mu}$ instead of $\stmtDeclInit{x}{\typevar}{\mu}$. Moreover, if $\mu$ is a \emph{Dirac} distribution, i.e.\ if $p_1 = 1$, we often write $\stmtAsgn{x}{\termvar_1}$ instead of $\stmtRasgn{x}{\mu}$. Finally, we assume that all programs and associated \HeyLo formulae are well-typed.

\subsection{Semantics of \HeyVL Statements}
\label{sec:heyvl-blocks-semantics}

\begin{figure}
	\begin{minipage}{0.5\textwidth}
		\begin{center}
			\adjustbox{width=\textwidth}{%
			\renewcommand*{\arraystretch}{1.5}%
			\begin{tabular}{l@{\quad}l@{\quad}}
				\toprule
				$\stmt$ & $\vc{\stmt}(\hla)$ \\
				\midrule
				\multirow{2}{*}{$\stmtDeclInit{x}{\typevar}{\mu}$} & $p_1 \cdot \hla\substBy{x}{\termvar_1}$ \\
					& \quad$+ \ldots + p_n \cdot \hla\substBy{x}{\termvar_n}$ \\
				\stmtDemonic{\stmt_1}{\stmt_2} & $\vc{\stmt_1}(\hla) \morespace{\sqcap} \vc{\stmt_2}(\hla)$\hspace*{1em} \\
				\Assert{\hlb} & $\hlb \sqcap \hla$ \\
				\Assume{\hlb} & $\hlb \rightarrow \hla$ \\
				\Havoc{x} & $\iquant{x}{\hla}$ \\
				\Validate & $\heylovalidate{\hla}$ \\
				\bottomrule
			\end{tabular}}
		\end{center}
	\end{minipage}%
	\begin{minipage}{0.5\textwidth}
		\begin{center}
			\adjustbox{width=\textwidth}{%
			\renewcommand*{\arraystretch}{1.5}%
			\begin{tabular}{@{\quad}l@{\quad}l}
				\toprule
				$\stmt$ & $\vc{\stmt}(\hla)$ \\
				\midrule
				\stmtTick{\aexpr} & $\hla + \aexpr$ \\
				\stmtSeq{\stmt_1}{\stmt_2} & $\vc{\stmt_1}\bigl(\vc{\stmt_2}(\hla)\bigr)$ \\
				\stmtAngelic{\stmt_1}{\stmt_2} & $\vc{\stmt_1}(\hla) \morespace{\sqcup} \vc{\stmt_2}(\hla)$\hspace*{1em} \\
				\coAssert{\hlb} & $\hlb \sqcup \hla$ \\
				\coAssume{\hlb} & $\hlb \coimpl \hla$ \\
				\coHavoc{x} & $\squant{x}{\hla}$ \\
				\coValidate & $\heylocovalidate{\hla}$ \\
				\bottomrule
			\end{tabular}}
		\end{center}
	\end{minipage}%
	\caption{Semantics of \HeyVL statements. Here $\mu = \pexp{p_1}{\termvar_1} \pexpand \ldots \pexpand \pexp{p_n}{\termvar_n}$ and $\hla\substBy{x}{\termvar_i}$ is the formula obtained from substituting every occurrence of $x$ in $\hla$ by $\termvar_i$ in a capture-avoiding manner. For procedure calls, see \Cref{sec:heyvl_proc_calls}.}
	\label{fig:heyvl-semantics}
\end{figure}
Inspired by weakest preexpectations~\cite{mciverAbstractionRefinementProof2005,kaminskiAdvancedWeakestPrecondition2019},
we give semantics to \HeyVL statements as a backward-moving continuation-passing style \HeyLo transformer%
	\[
	      \vc{\stmt} \colon \HeyLo \to \HeyLo~
	\]
by induction on $\stmt$ in \Cref{fig:heyvl-semantics}.
(Co)procedure calls are treated separately in \Cref{sec:heyvl_proc_calls}.
We call $\vc{\stmt}(\hla)$ the \emph{verification preexpectation} of $\stmt$ with respect to post $\hla$.
Intuitively, $\interpretsimple{\vc{\stmt}(\hla)}(\State)$ is the expected value of $\hla$ w.r.t.\ the distribution of final states obtained from \enquote{executing}\footnote{Some verification-specific statements are not really \emph{executable} but serve the purpose of manipulating expected values.} $\stmt$ on $\State$. 
The post $\hla$ is either given by the surrounding procedure declaration or can be thought of as the verification preexpectation described by the \emph{remaining} \HeyVL statement: for $\stmt = \stmtSeq{\stmt_1}{\stmt_2}$, we first obtain the intermediate verification preexpectation $\vc{\stmt_2}(\hla)$ --- the expected value of what remains after executing $\stmt_1$ --- and pass this into $\vc{\stmt_1}$.

\paragraph{Random Assignments}
The expected value of $\hla$ after executing $\stmtDeclInit{x}{\typevar}{\mu}$ is the weighted sum $p_1 \cdot \hla\substBy{x}{\termvar_1} + \ldots + p_n \cdot \hla\substBy{x}{\termvar_n}$, where each $p_i$ is the probability that $x$ is assigned $\termvar_i$.

\paragraph{Rewards} 
Suppose that the post $\hla$ captures the expected reward collected in an execution that follows \emph{after} executing $\stmtTick{\aexpr}$.
Then the entire expected reward is given by $\hla + \aexpr$.

\paragraph{Nondeterministic Choices}
$\vc{\stmtCdot{\stmt_1}{\stmt_2}}(\hla)$ is the pointwise minimum ($\cdot = \sqcap$) or maximum ($\cdot = \sqcup$) of the expected values obtained from $\stmt_1$ and $\stmt_2$, respectively.

\paragraph{(Co)assertions}
In \emph{classical} intermediate verification languages, the statement $\stmtAssert{A}$ for some predicate $A$ models a proof obligation: All states reaching $\stmtAssert{A}$ on some execution must satisfy~$A$. In terms of classical weakest preconditions, $\stmtAssert{A}$ transforms a postcondition $B$ to
\[
\wp{\stmtAssert{A}}(B) \eeq A \wedge B~.
\]
In words, $\stmtAssert{A}$ \emph{caps} the truth of postcondition $B$ at $A$: all lower-bounds on the above weakest precondition (in terms of the Boolean lattice $(\States \to \Bools,\, {\Rightarrow})$) must not exceed $A$.

This perspective generalizes well to our quantitative assertions: Given a \HeyLo formula $\hlb$, the statement $\stmtAssert{\hlb}$ \emph{caps} the post at $\hlb$. Thus, analogously to classical assertions, all \emph{lower} bounds on the verification preexpectation $\vc{\stmtAssert{\hlb}}(\hla)$ (in terms of $\heyloleq$) must not exceed $\hlb$.

Coassertions are dual to assertions: $\coAssert{\hlb}$ \emph{raises} the post $\hla$ to at least $\hlb$. 
Hence, all \emph{upper} bounds on $\vc{\coAssert{\hlb}}(\hla)$ must not \emph{sub}ceed $\hlb$.

\paragraph{(Co)assumptions}
In the classical setting, the statement $\stmtAssume{A}$ for some predicate $A$ \emph{weakens} the verification condition: verification succeeds vacuously for all states not satisfying $A$. In terms of classical weakest preconditions, $\stmtAssume{A}$ transforms a postcondition $B$ to
\[
\wp{\stmtAssume{A}}(B) \eeq A \impl B~
\]
i.e.\ $\stmtAssume{A}$ \emph{lowers} the threshold at which the post $B$ is considered $\true$ (the top element of the Boolean lattice) to $A$. Indeed, if we identify $\true =1$ and $\false =0$, then
\[
   \interpretsimple{\wp{\stmtAssume{A}}(B)}(\sigma) \eeq 
   \begin{cases}
   	 1, & \text{if $\interpretsimple{A}(\sigma) \leq \interpretsimple{B}(\sigma)$} \\
   	 \interpretsimple{B}(\sigma), &\text{otherwise}~.
   \end{cases}
\]

The above perspective on classical assumptions generalizes to our quantitative assumptions. Given a $\HeyLo$ formula $\hlb$, $\stmtAssume{\hlb}$ lowers the threshold above which the post $\hla$ is considered entirely true (i.e.\ $\infty$ -- the top element of the lattice of expectations) to $\hlb$. 
Formally, %
\[
   \interpretsimple{\vc{\stmtAssume{\hlb}}(\hla)}(\sigma)
   \eeq 
   \begin{cases}
   	\infty, & \text{if $\interpretsimple{\hlb}(\sigma) \leq \interpretsimple{\hla}(\sigma)$} \\
   	\interpretsimple{\hla}(\sigma), & \text{otherwise}~.
   \end{cases}
\]
Reconsider \Cref{fig:assume} on \cpageref{fig:assume}, which illustrates $\vc{\stmtAssume{5}}(x)$: $\stmtAssume{5}$ lowers the threshold at which the post $x$ is considered entirely true to $5$, i.e.\ whenever the post-expectation $x$ evaluates at least to $5$, then $\vc{\stmtAssume{5}}(x)$ evaluates to $\infty$.
Notice furthermore that our quantitative $\symAssume$ is backward compatible with the classical one in the sense that $\vc{\stmtAssume{\embed{\bexpr}}}(\hla)$ evaluates to $\hla$ for every state satisfying $\bexpr$, and to $\infty$ otherwise.

Coassumptions are dual to assumptions. $\coAssume{\hlb}$ raises the threshold at which the post $\hla$ is considered entirely false (i.e.\ $0$ -- the bottom element of the lattice of expectations) to $\hlb$. 
Reconsider \Cref{fig:coassume} on \cpageref{fig:coassume} illustrating $\vc{\coAssume{5}}(x)$: 
$\coAssume{5}$ raises the threshold below which the post $x$ is considered entirely false to $5$, i.e.\ if the post $x$ evaluates at most to $5$, then $\vc{\coAssume{5}}(x)$ \mbox{evaluates to $0$.}

\begin{example}[Modeling Conditionals]\label{ex:ite}
	We did not include $\stmtIf{\bexpr}{\stmt_1}{\stmt_2}$ for conditional branching in \HeyVL's grammar.
	 We can encode it as follows (and will use it from now on):
	\[
	    \stmtDemonic{\stmtSeq{\stmtAssume{\embed{\bexpr}}}{\stmt_1} }
	    {\stmtSeq{\stmtAssume{\embed{\neg\bexpr}}}{\stmt_2}}
	\]
	The $\symVc$ semantics of this statement is analogous to the formula described in \Cref{ex:heylo_ifthenelse} and complies with our above description of assumptions: Depending on the satisfaction of $\bexpr$ by the current state $\sigma$, the $\symVc$ of $\stmt$ either evaluates to the $\symVc$ of $\stmt_1$ or $\stmt_2$, respectively.
\end{example}

\paragraph{(Co)havocs} 
In the classical setting, $\stmtHavoc{x}$ %
forgets the current value of $x$ by universally quantifying over all possible initial values of $x$.
In terms of classical weakest preconditions, we have
\[
   \wp{\stmtHavoc{x}}(B) \eeq \quant{\forall \typeof{x}{\typevar}}{B}~,
\]
i.e.\ $\stmtHavoc{x}$ \emph{minimizes} the post $B$ under all possible values for $x$, thus requiring $B$ to hold for all $x$. This perspective generalizes to our quantitative setting: In terms of $\symVc$, $\stmtHavoc{x}$ forgets the current value of $x$ by minimizing the post-expectation under all possible values of $x$.
Dually, $\coHavoc{x}$ forgets the value of $x$ but this time \emph{maximizes} the post-expectation under all possible values for $x$.%

\paragraph{(Co)validations} 
These statements convert quantitative statements into qualitative ones by casting expectations into the $\{0,\infty\}$-valued realm, thus eradicating intermediate truth values strictly between 0 and $\infty$. 
Their classical analogues would be effectless, as the Boolean setting features no intermediate truth values.
We briefly explained in \Cref{subsec:syntactic_sugar} how such a conversion to a qualitative statement works in \HeyLo. An example will be discussed in \Cref{sec:encodings:park}.

\subsection{Properties of \HeyVL Statements}
\label{sec:heyvl-properties}
We study two properties of \HeyVL. First, our $\symVc$ semantics is \emph{monotonic} --- a crucial property for encoding proof rules (cf.\ \Cref{sec:heyvl_proc_calls}).%
\begin{restatable}[Monotonicity of $\symVc$]{theorem}{thmHeyvlMonotonicity}
	\label{thm:monotonicity}
	For all \HeyVL statements $\stmt$ and \HeyLo formulae $\hla,\hla'$,
	\[
	\hla \hheyloleq \hla' \quad\text{implies}\quad \vc{\stmt}(\hla) \hheyloleq \vc{\stmt}(\hla')~.
	\]
\end{restatable}%
\noindent%
Furthermore, \HeyVL conservatively extends an existing IVL for non-probabilistic programs due to \mbox{\citet{mullerBuildingDeductiveProgram2019}} in the following sense:%
\begin{restatable}[Conservativity of \HeyVL]{theorem}{thmHeyvlConservativity}
	\label{thm:heyvl-conservativity}
	Let $C$ be a program in the programming language of \textnormal{\citet{mullerBuildingDeductiveProgram2019}} and let $B$ be a postcondition. 
	Moreover, let $\overline{C}$ be obtained by replacing every $\stmtAssert{A}$ and every $\stmtAssume{A}$ occurring in~$C$ by $\stmtAssert{\embed{A}}$ and $\stmtAssume{\embed{A}}$, respectively (cf.~Boolean embeddings, \textnormal{\Cref{sec:heylo-sem}}). 
	Then
	\abovedisplayskip=0pt%
	\begin{align*}
		\embed{
			\:\underbrace{
				\wp{C}(B)
			}_{
				\mathclap{
					\text{
						verification condition obtained from \textnormal{\cite{mullerBuildingDeductiveProgram2019}}
					}
				}
			}\:
		}
		\qquad\equiv\qquad 
		\overbrace{\vc{\overline{C}}(\embed{B})}^{\smash{\HeyVL}}
	  ~.
	\end{align*}%
	\normalsize%
\end{restatable}

\subsection{Procedure Calls}
\label{sec:heyvl_proc_calls}
We conclude this section with a treatment of (co)procedure calls.
Consider a callee \emph{procedure} $\procname$ as shown in \Cref{fig:calls:proc}. 
\begin{figure}
\begin{align*}
  &\proc{\procname}{\varandtype{x_1}{\typevar_1}, \ldots, \varandtype{x_n}{\typevar_n}}{\varandtype{y_1}{\typevar_1'}, \ldots, \varandtype{y_m}{\typevar_m'}} \\
  &\qquad\Requires{\hlc}  \\
  &\qquad\Ensures{\hlb} \\
  &\blockStart~\sstmt~\blockEnd
\end{align*}
\caption{A procedure $\procname$.
We encode calls $\stmtAsgn{z_1,\ldots,z_m}{P(t_1,\ldots,t_n)}$ for arbitrary probabilistic statements $\sstmt$.}
\label{fig:calls:proc}	
\end{figure}
Intuitively, the effect of a call $\stmtAsgn{z_1,\ldots,z_m}{P(t_1,\ldots,t_n)}$
corresponds to
(1) initializing $\procname$'s formal input parameters $x_1,\ldots,x_n$
with the arguments $t_1,\ldots,t_n$,
(2) inlining $\procname$'s body $\sstmt$, and 
(3) assigning to $z_1,\ldots,z_m$ the values of outputs $y_1,\ldots,y_m$.
The semantics of $\stmtAsgn{z_1,\ldots,z_m}{P(t_1,\ldots,t_n)}$ can be thought of as the statement\footnote{For the sake of simplicity, we ignore potential scoping issues arising if $\sstmt$ uses variables that are declared in the calling context; these issues can be resolved by a straightforward yet tedious variable renaming.}
\begin{align*}
	\underbrace{
		\stmtAsgn{x_1}{t_1}\symSemi
		\ldots\symSemi
		\stmtAsgn{x_n}{t_n}
	}_{
		\mathclap{
			\qquad\qquad\qquad\qquad\qquad{}~{}\quad
			{}\eqqcolon{}~ 
			\mathit{init}\quad \text{(initialize procedure inputs)}
		}
	}
	\symSemi
	\quad
	\overbrace{\sstmt\symSemi}^{\mathclap{\text{inlining of the procedure body}}}
	\quad
	\underbrace{
		\stmtAsgn{z_1}{y_1}\symSemi
		\ldots\symSemi
		\stmtAsgn{z_m}{y_m}
	}_{
		\mathclap{
			\qquad\qquad\qquad\qquad\qquad{}~{}~{}\quad
			{}\eqqcolon{}~ 
			\mathit{return}\quad \text{(assign procedure outputs)}
		}
	}~.
\end{align*}%
There are two main issues that arise when we would actually inline $S$ at every call-site:
(1) For recursive procedure calls~\cite{olmedoReasoningRecursiveProbabilistic2016}, we would need to define a (non-computable) fixed point semantics for the $\symVc$ transformer.
Our goal, however, is to render verification feasible in practice, so we would like to avoid fixed point computations. %
(2) Even without recursive calls, we would have to re-verify $\sstmt$ at every call-site, which would not scale.

We thus do not inline the procedure body but use an \emph{encoding} $\sstmt_\mathit{encoding}$ which \emph{underapproximates} the effect of $\sstmt$ in the sense that $\vc{\sstmt_\mathit{encoding}}(\hla) \heyloleq \vc{\sstmt}(\hla)$ for all \HeyLo formulae~$\hla$.
By monotonicity of $\symVc$, we can then verify lower bounds for calls:
for all $\hla,\hld \in \HeyLo $, 
\[
  \hld \hheyloleq
  \vc{\underbrace{\mathit{init}\symSemi\sstmt_\mathit{encoding}\symSemi\mathit{return}}_{\mathclap{\text{modular encoding of calls}}}}(\hla)
  \qquad\text{implies}\qquad
  \hld \hheyloleq  
  \vc{\underbrace{\mathit{init}\symSemi\sstmt\symSemi\mathit{return}}_{\mathclap{\text{actual inlining of calls}}}}(\hla)~,
\]
so whenever we can verify a \HeyVL program using the modular encoding, we could have also verified it using inlining.
The advantage of the modular encoding is that $\sstmt_\mathit{encoding}$ does not contain the procedure body -- it could be changed without requiring re-verification of call sites, so long as the updated procedure body still adheres to the procedure's specification.
To construct $\sstmt_\mathit{encoding}$, we leverage only $\procname$'s specification pre $\hhlc$ and post $\hhlb$, cf.~\Cref{fig:calls:proc}:
Assuming that $P$ verifies, we can safely assume that $P$'s verification condition -- namely $\hhlc \heyloleq \vc{\sstmt}(\hhlb)$ -- holds.\footnote{Otherwise, procedure $\procname$ in \Cref{fig:calls:proc} does not verify and verification of the whole \HeyVL program fails anyway.}
By monotonicity of~$\symVc$, we have $\hhlc \heyloleq \vc{\sstmt}(\hhlb) \heyloleq \vc{\sstmt}(\hla)$ whenever $\hhlb \heyloleq \hla$ holds.
To underapproximate $\vc{\sstmt}(\hla)$, we construct $\sstmt_\mathit{encoding}$ such that $\vc{\sstmt_\mathit{encoding}}(\hla)$ is the known lower bound $\hhlc$ if $\hhlb \heyloleq \hla$; otherwise, it is the trivial lower bound $0$.
So how do we construct $\sstmt_\mathit{encoding}$ concretely?

In classical verification infrastructures (cf.~\cite{mullerBuildingDeductiveProgram2019}), $\sstmt_\mathit{encoding}$ corresponds to the statement
\[
  \stmtAssert{\hhlc}\symSemi
  \stmtHavoc{z_1}\symSemi \ldots \symSemi\stmtHavoc{z_m}\symSemi
  \stmtAssume{\hhlb}.
\]
That is, we assert the procedure's pre $\hhlc$ before the call, forget the values of all outputs, i.e.\ variables that are potentially modified by the call, and assume the procedure's post $\hhlb$ after the call.
Phrased in terms of underapproximations:
We assert that we have at most $\hhlc$ before the call and, while minimising over all possible outputs (using the havoc statements), lower the threshold at which the post is considered entirely true (i.e.\ $\infty$) to $\hhlb$, i.e.\ whenever $\hhlb$ lower-bounds the post.

The intuition underlying the above \HeyVL statement works for encoding procedure calls of non-probabilistic programs. However, there is a subtle \emph{unsoundness} that arises when reasoning about \emph{expected} behaviors.
\Cref{fig:calls:cex} shows two procedures, $\mathit{foo}$ and $\mathit{bar}$.
\begin{figure}
\begin{minipage}{0.5\textwidth}
\begin{align*}
  & \proc{\mathit{foo}}{\varandtype{x}{\uint}}{} \\
  &\qquad\Requires{x}  \\
  &\qquad\Ensures{2\cdot x} \\
  &\blockStart ~~ \intersem{\text{verifies: } x \heyloleq 0.5 \cdot (2\cdot x)} \\
  &\qquad \stmtDeclInit{b}{\bool}{\pexp{0.5}{\true} \pexpand \pexp{0.5}{\false}}\symSemi \\
  &\qquad \stmtAssert{\embed{b}}~\blockEnd
\end{align*}
\end{minipage}%
\begin{minipage}{0.5\textwidth}
\begin{align*}
  & \proc{\mathit{bar}}{\varandtype{x}{\uint}}{} \\
  &\qquad\Requires{x}  \\
  &\qquad\Ensures{x} \\
  &\blockStart ~~ \intersem{\text{verifies: } x \heyloleq x \sqcap (2\cdot x \impl x) } \\
  &\qquad \intersem{\text{encoding of }\mathit{foo}(x)} \\
  &\qquad \stmtAssert{\colheylo{x}}\symSemi \stmtAssume{\colheylo{2\cdot x}}~\blockEnd
\end{align*}
\end{minipage}
\caption{Unsound encoding of a procedure call $\mathit{foo}(x)$ in $\mathit{bar}$.
Both procedures verify but inlining the body of $\mathit{foo}$ in $\mathit{bar}$ does not as it produces the (wrong) inequality $\colheylo{x} \heyloleq 0.5 \cdot \colheylo{x}$.}
\label{fig:calls:cex}	
\end{figure}
Intuitively, $\mathit{foo}$ flips a fair coin and aborts execution if the result is heads ($\false$).
Read backwards, the expected value of the post will be at most $\colheylo{x}$ after executing $\mathit{foo}$ -- exactly as stated in $\mathit{foo}$'s specification.
 Procedure $\mathit{bar}$ encodes the call $\mathit{foo}(x)$ in its body\footnote{There are no havoc statements because $\mathit{foo}$ has no outputs; we also omitted $\mathit{init}$ and $\mathit{return}$ for simplicity.} and requires in its specification that the expected value of $x$ does not decrease, i.e.\ is at least $\colheylo{x}$.
 Both procedures verify. However, when inlining $\mathit{foo}$, i.e.\ using its body instead of the encoding $\stmtAssert{\colheylo{x}}\symSemi\stmtAssume{\colheylo{2\cdot x}}$, $\mathit{bar}$ does \emph{not} verify.
 Hence, the above encoding does, in general, not model a sound underapproximation of a procedure's inlining.
 
Taking a closer look, recall from above that $\stmtAssume{\colheylo{2\cdot x}}$ is used to encode a monotonicity check,\footnote{More precisely: a check whether monotonicity of $\symVc$ can be applied, namely whether $\hhlb \heyloleq \hla$ holds where $\hhlb$ is the callee's \emph{specified} post and $\hla$ is the \emph{actual} post at the call-site.} which is an inherently \emph{quali}tative property.
 However, verifying $\mathit{bar}$ involves proving $\colheylo{x} \heyloleq \colheylo{x} \sqcap (\colheylo{2\cdot x} \impl \colheylo{x})$, where the quantitative implication $\colheylo{2\cdot x} \impl \colheylo{x}$ evaluates to $\colheylo{x}$ for $x > 0$; the expectation $\colheylo{x}$ does not reflect the inherently qualitative nature of the monotonicity check.
 To fix this issue, we add a $\stmtValidate$ statement that turns \emph{quanti}tative results into \emph{quali}tative ones: it reduces any value less than $\infty$, which indicates a failed monotonicity check, to $0$.
An encoding underapproximating the inlining of $\mathit{foo}(x)$ -- and thus correctly failing verification of $\mathit{bar}$ -- is $\stmtAssert{\colheylo{x}}\symSemi\stmtValidate\symSemi\stmtAssume{\colheylo{2\cdot x}}$.
 Similarly to~\Cref{subsec:syntactic_sugar},
 verifying $\mathit{bar}$ for the fixed encoding involves proving
 $\colheylo{x} \heyloleq \colheylo{x} \sqcap \heylovalidate{\colheylo{2\cdot x} \impl \colheylo{x}}$, which does not hold for $x > 0$.
 
More generally, a sound construction of $\sstmt_\mathit{encoding}$ (wrt. underapproximating procedure body $\sstmt$)  is 
\begin{align*}
    \sstmt_\mathit{encoding}\colon \qquad  &
  \stmtAssert{\hhlc}\symSemi
  \stmtHavoc{z_1}\symSemi \ldots \symSemi\stmtHavoc{z_m}\symSemi
  \stmtValidate\symSemi
  \stmtAssume{\hhlb}.
\end{align*}
Formally, we obtain an underapproximating \HeyVL encoding of procedure calls of the form $\stmtAsgn{z_1,\ldots,z_m}{P(t_1,\ldots,t_n)}$ for arbitrary probabilistic procedures as in \Cref{fig:calls:proc}:%
\begin{restatable}{theorem}{thmProcCalls}
 Let $\sstmt$ be the body of the procedure $\procname$ in \Cref{fig:calls:proc}.
 Then, for every \HeyLo formula $\hla$, 
 \[ 
   \vc{\sstmt_\mathit{encoding}}(\hla) \hheyloleq \vc{\sstmt}(\hla)
   \qand
   \vc{\mathit{init}\symSemi\sstmt_\mathit{encoding}\symSemi\mathit{return}}(\hla)
   \hheyloleq  
   \vc{\mathit{init}\symSemi\sstmt\symSemi\mathit{return}}(\hla).
 \]
\end{restatable}
\noindent%
\ifdefined\WITHAPPENDIX
A proof is found in \cref{sec:app-proofs-heyvl}.
\fi
A \HeyVL encoding that \emph{over}approximates calls of \emph{co}procedures is analogous -- it suffices to use the dual \emph{co}statements in $\sstmt_\mathit{encoding}$.
The presented under- and overapproximations are useful when encoding proof rules in \HeyVL. 
Whether they are meaningful does, however, depend on the verification technique at hand that should be encoded.

%% file: sections/5_case_studies.tex
\section{Encoding Case Studies}
\label{sec:encodings}
To evaluate the expressiveness of our verification language, we encoded various existing calculi and proof rules targeting verification problems for probabilistic programs in \HeyVL.
We will first focus on programs without $\symWhile$ loops (\Cref{sec:encodings:loop-free}) and then consider loops (\Cref{sec:encodings:park}). 
The practicality of our automated verification infrastructure will be evaluated separately in \Cref{sec:implementation}.
A summary of all encodings is given at the end of this section.
\ifdefined\WITHAPPENDIX
Further details are found in \Cref{app:encodings}.
\fi

\subsection{Reasoning about While-Loop-Free \pGCL Dialects}\label{sec:encodings:loop-free}
Pioneered by \citet{kozenProbabilisticPDL1983}, expectation-based techniques have been successfully applied to analyze various probabilistic program properties.
\citet{mciverAbstractionRefinementProof2005} incorporated nondeterminism and introduced the probabilistic Guarded Command Language (\pGCL), which is convenient for modelling probabilistic systems.
The syntax of $\symWhile$-loop-free \pGCL programs $\pcc$ is\footnote{\pGCL usually supports only one type, e.g. integers, rationals, or reals. We are more liberal and admit arbitrary terms $t$ but assume a sufficiently strong type inference system and consider only well-typed programs.}
\begin{align*}
  \pcc \morespace{\Coloneqq} 
  \pSkip
  ~|~ \pDiverge
  ~|~ \pAssign{x}{t} 
  ~|~ \pcc_1;\pcc_2 
  ~|~ \pIte{b}{\pcc_1}{\pcc_2}
  ~|~ \pChoice{\pcc_1}{p}{\pcc_2}
  ~|~ \pNd{\pcc_1}{\pcc_2}~,
\end{align*}
where $\pSkip$ has no effect,
$\pDiverge$ never terminates,
$\pAssign{x}{t}$ assigns the value of term $t$ to $x$,
$\pcc_1;\pcc_2$ executes $\pcc_2$ after $\pcc_1$, 
$\pIte{b}{\pcc_1}{\pcc_2}$ executes $\pcc_1$ if Boolean expression $b$ holds and $\pcc_2$ otherwise,
$\pChoice{\pcc_1}{p}{\pcc_2}$ executes $\pcc_1$ with probability $p \in [0,1]$ and $\pcc_2$ with probability $(1-p)$, and 
$\pNd{\pcc_1}{\pcc_2}$ nondeterministically executes either $\pcc_1$ or $\pcc_2$.

We now outline encodings of several reasoning techniques targeting \pGCL and extensions thereof.
We will only consider expectations that can be expressed as \HeyLo formulae.
To improve readability, we identify every \HeyLo formula $\hla$ with its expectation $\interpretsimple{\hla} \in \Expectations$.
\paragraph{Weakest Preexpectations ($\textit{wp}$)}
The \emph{weakest preexpectation calculus} of \citet{mciverAbstractionRefinementProof2005} maps every \pGCL command $\pcc$ and postexpectation $\hla$ %
to the \emph{minimal} (to resolve nondeterminism) \emph{expected value 
$\foreignWp{wp}{\pcc}{\hla}$ of $\hla$ after termination of $\pcc$} -- the same intuition underlying \HeyVL's $\symVc$ transformer.
\Cref{fig:encodings:wp} shows a sound and complete \HeyVL encoding $\downTransWp{\pcc}$ of the weakest preexpectation calculus, i.e. 
  $\vc{\downTransWp{\pcc}}(\hla) =\foreignWp{wp}{\pcc}{\hla}$.
  Most \pGCL commands have \HeyVL equivalents; conditionals are encoded as in~\Cref{ex:ite}. 
  $\pDiverge$ is encoded as $\stmtAssert{0}$ as it never terminates, i.e. $\foreignWp{wp}{\pDiverge}{\hla} = 0$.
The program in \Cref{fig:encodings:wp-proc} then verifies iff $\hlb$ lower bounds $\foreignWp{wp}{\pcc}{\hla}$, i.e. $\hlb \heyloleq \foreignWp{wp}{\pcc}{\hla}$.
To reason about \emph{upper} bounds, it suffices to use a \emph{co}procedure instead.
\begin{figure}
	\begin{minipage}{0.55\textwidth}
		\begin{tabular}{ll}
			\toprule
			$\pcc$ & $\downTransWp{\pcc}$ \\
			\midrule
  			$\pSkip$ & $\stmtTick{0}$ \\
  			$\pDiverge$ & $\Assert{0}$ \\
  			$\pAssign{x}{t}$ & $\stmtRasgn{x}{t}$ \\
  			$\pcc_1;\pcc_2$ & $\downTransWp{\pcc_1};\downTransWp{\pcc_2}$ \\
  			$\stmtIfStart{b}~\pcc_1~\} $ & $\stmtDemonicStart~\Assume{\embed{b}};\downTransWp{\pcc_1}~\}$ \\
  			$\quad\stmtElseStart~\pcc_2 \}$ & $\stmtElseStart~\Assume{\embed{\neg b}};\downTransWp{\pcc_2} \}$ \\
  			$\pChoice{\pcc_1}{p}{\pcc_2}$ & $\stmtDeclInit{\mathit{tmp}}{\bool}{\exprFlip{p}}\symSemi$ \\ %
  			& $\downTransWp{\pIte{\mathit{tmp}}{\pcc_1}{\pcc_2}}$ \\
			$\pNd{\pcc_1}{\pcc_2}$ & $\stmtDemonicStart~\downTransWp{\pcc_1}~\}$ \\
			& $\stmtElseStart~\downTransWp{\pcc_2}~\}$ \\
  			\bottomrule
		\end{tabular}
	    \caption{Encoding of weakest preexpectation for \pGCL, where $\mathit{tmp}$ is a fresh variable.}
	    \label{fig:encodings:wp}
	\end{minipage}%
	\hfill%
	\begin{minipage}{0.44\textwidth}
		\begin{align*}
   		  &\proc{lower}{\overline{\varin}}{\overline{\varout}} \\
   		  &\quad\Requires{\hlb}~\intersem{\overline{\varin}\text{: variables in $\hlb$}} \\
  		  &\quad\Ensures{\hla~}~\blockStart~\intersem{\overline{\varout}\text{: var. in $\hla$ but not $\hlb$}} \\
  		  & \qquad\intersem{\text{declare local variables, i.e.}} \\
  		  & \qquad\intersem{\text{those not in $\hla$ or $\hlb$, using}} \\
  		  & \qquad\intersem{\stmtDeclInit{x}{\typevar}{\mathit{default}}\symSemi\stmtHavoc{x}} \\
  		  & \qquad \downTransWp{\pcc} \\
  		  &\blockEnd
		\end{align*}%
		\vspace{-1.5em}
	    \caption{Encoding of $\hhlb \heyloleq \foreignWp{wp}{\pcc}{\hhla}$.}
	    \label{fig:encodings:wp-proc}
	\end{minipage}
\end{figure}

\paragraph{Weakest Liberal Preexpectations ($\textit{wlp}$).}
\citet{mciverAbstractionRefinementProof2005} also proposed a \emph{liberal} weakest preexpectation calculus, a partial correctness variant of weakest preexpectations.
	More precisely, if $\hla \heyloleq 1$, then the weakest liberal preexpectation
	$\foreignWp{wlp}{\pcc}{\hla}$ is the expected value of $\hla$ after termination of $\pcc$ \emph{plus} the probability of non-termination of $\pcc$ (on a given initial state).
	We denote by $\downTransWlp{\pcc}$ the \HeyVL encoding of the weakest liberal preexpectation calculus; it is defined analogously to 
	\Cref{fig:encodings:wp} except for $\pDiverge$.
	Since $\pDiverge$ never terminates, the probability of non-termination is one, i.e. $\foreignWp{wlp}{\pDiverge}{\ldots} = 1$.
	The updated encoding of $\pDiverge$ is 
	\[
	  \downTransWlp{\pDiverge} \qeq \Assert{1}\symSemi \Assume{0}~,
	\]
	where $\Assert{1}$ ensures one-boundedness and $\Assume{0}$ lowers the threshold at which the post is considered entirely true to $0$.
	Put together, we have
	$\vc{\downTransWlp{\pDiverge}}(\hla) = 1 \sqcap \infty = 1 = \foreignWp{wlp}{\pDiverge}{\hla}$.

\paragraph{Conditional Preexpectations ($\textit{cwp}$).}
\emph{Conditioning} on observed events (in the sense of conditional probabilities) is a key feature of modern probabilistic programming languages \cite{gordonProbabilisticProgramming2014}.
Intuitively, the statement $\stmtObserve{b}$ discards an execution whenever Boolean expression $b$ does not hold. Moreover, it re-normalizes such that the accumulated probability of all executions violating no observation equals one.
\citet{olmedoConditioningProbabilisticProgramming2018} showed that reasoning about \stmtObserve{b} requires a combination of $\symForeign{wp}$ and $\symForeign{wlp}$ reasoning.
They extended both calculi such that violating an observation is interpreted as a failure resulting in pre-expectation zero; we can encode it with an assertion: 
\[
\foreignWp{w(l)p}{\stmtObserve{b}}{\hla} 
\eeq 
\embed{b} \sqcap \hla
\eeq
\vc{\stmtAssert{\embed{b}}}(\hla).
\]
For every \pGCL program $\pcc$ with observe statements, initial state $\State$ and expectation $\hla$, the \emph{conditional} expected value $\foreignWp{cwp}{\pcc}{\hla}(\State)$ of $\hla$ after termination of $\pcc$ is then given by the expected value $\foreignWp{wp}{\pcc}{\hla}(\State)$ normalized by the probability $\foreignWp{wlp}{\pcc}{1}(\State)$ of violating no observation:
\[
   \foreignWp{cwp}{\pcc}{\hla}(\State) \qeq \frac{\foreignWp{wp}{\pcc}{\hla}(\State)}{\foreignWp{wlp}{\pcc}{1}(\State)}
   \qquad (\text{undefined if }\foreignWp{wlp}{\pcc}{1}(\State) \eeq 0)
\]
We can re-use our existing \HeyVL encodings to reason about conditional expected values.
Notice that proving bounds on $\symForeign{cwp}$ requires establishing both lower and upper bounds.
For example, the \pGCL program $\pcc_\mathit{die}$ in \Cref{fig:die-pgcl} assigns to $r$ the result of a six-sided die roll, which is simulated using three fair coin flips and an observation.
To show that the expected value of $\colheylo{r}$ is at most $\colheylo{3.5}$ -- the expected value of a six-sided die roll -- we prove the
upper bound $\foreignWp{wp}{\pcc_\mathit{die}}{\colheylo{r}} \heyloleq \colheylo{2.625}$ and the lower bound $\colheylo{0.75} \heyloleq \foreignWp{wlp}{\pcc_\mathit{die}}{\colheylo{1}}$. Then, $\foreignWp{cwp}{\pcc_\mathit{die}}{\colheylo{r}} \heyloleq \frac{\colheylo{2.625}}{\colheylo{0.75}} = \colheylo{3.5}$.
\Cref{fig:die-heyvl} shows the \HeyVL encoding of $\pcc_\mathit{die}$ (cleaned up for readability).
As shown in \Cref{fig:cwp-die}, the proof obligations $\foreignWp{wp}{\pcc_\mathit{die}}{\colheylo{r}} \heyloleq \colheylo{2.625}$ and $\colheylo{0.75} \heyloleq \foreignWp{wlp}{\pcc_\mathit{die}}{\colheylo{1}}$ are then encoded using a coprocedure for the upper bound and a procedure for the lower bound, respectively. 

There exist alternative interpretations of conditioning. For instance, \citet{r2} use $\foreignWp{wp}{\pcc}{1}(\State)$ in the denominator in the above fraction. A benefit of \HeyVL is that such alternative interpretations can be realized by a straightforward adaptation of our encoding. 
\begin{figure}[t]
	\begin{halfboxl}
		\begin{align*}
			& \pChoice{\pAssign{a}{0}}{0.5}{\pAssign{a}{1}}; \\
			& \pChoice{\pAssign{b}{0}}{0.5}{\pAssign{b}{1}}; \\
			& \pChoice{\pAssign{c}{0}}{0.5}{\pAssign{c}{1}}; \\
			& \pAssign{r}{4 \cdot a + 2 \cdot b + c + 1}; \\
			& {\color{heyvlColor}\stmtObserve{r \leq 6}}
		\end{align*}
		\vspace{-2em}
		\caption{\pGCL program $\pcc_\mathit{die}$.}
		\label{fig:die-pgcl}
		\begin{align*}
			&\symUp\proc{die\_wp}{}{\varandtype{r}{UInt}} \\
			&\quad\Requires{2.625} \\
			&\quad\Ensures{r} \\
			&\quad\blockStart~\stmt_\mathit{die}~\blockEnd
		\end{align*}
	\end{halfboxl}%
	\begin{halfboxr}
		\begin{align*}
			&\stmtDeclInit{a}{\uint}{\pexp{0.5}{1} \pexpand \pexp{0.5}{0}}\symSemi \\
			&\stmtDeclInit{b}{\uint}{\pexp{0.5}{1} \pexpand \pexp{0.5}{0}}\symSemi \\
			&\stmtDeclInit{c}{\uint}{\pexp{0.5}{1} \pexpand \pexp{0.5}{0}}\symSemi \\
			&\stmtRasgn{r}{4 \cdot a + 2 \cdot b + c + 1}\symSemi \\
			&{\color{heyvlColor}\downAssert{\embed{r \leq 6}}}
		\end{align*}
		\vspace{-2em}
		\caption{\HeyVL encoding $\stmt_\mathit{die}$ of $\pcc_\mathit{die}$.}
		\label{fig:die-heyvl}
		\begin{align*}
			&\proc{die\_wlp}{}{\varandtype{r}{UInt}} \\
			&\quad\Requires{\nicefrac{6}{8}} \\
			&\quad\Ensures{1} \\
			&\quad\blockStart~\stmt_\mathit{die}~\blockEnd
		\end{align*}
	\end{halfboxr}
	\vspace{-1.5em}
	\caption{\HeyVL encoding of the proof obligations $\wp{\pcc_\mathit{die}}(\colheylo{r}) \heyloleq \colheylo{2.625}$ and $\colheylo{0.75} \heyloleq \wlp{\pcc_\mathit{die}}(\colheylo{1})$.}
	\label{fig:cwp-die}
\end{figure}

\subsection{Reasoning about Expected Values of Loops}
\label{sec:encodings:park}
We encoded various proof rules for loops $\stmtWhile{b}{\pcc}$ in \HeyVL. %
As an example, we consider the Park induction rule~\cite{parkFixpointInductionProofs1969,kaminskiAdvancedWeakestPrecondition2019} for lower bounds on weakest liberal preexpectations: for all $\hla,\loopinv \heyloleq 1$,
\begin{align*}
    \underbrace{
      \loopinv \hheyloleq \left( \embed{b} \impl \foreignWp{wlp}{\pcc}{\loopinv} \right) 
      \sqcap 
      \left(\embed{\neg b} \impl \hla\right)
    }_{\loopinv\text{ is an inductive invariant}}
  \quad \text{implies} 
  \quad 
  \underbrace{
  \loopinv \hheyloleq \foreignWp{wlp}{\stmtWhile{b}{\pcc}}{\hla}
  }_{
  	\loopinv\text{ underapproximates the loop's \symForeign{wlp}}
  }
    .
\end{align*}
The rule can be viewed as a quantitative version of the loop rule from~\citet{hoareAxiomaticBasisComputer1969} logic, where $\loopinv$ is an \emph{inductive invariant} underapproximating the expected value of any loop iteration.
\Cref{fig:park-wlp-heyvl} depicts an %
encoding $\downTransWlp{\stmtWhile{b}{\pcc}}$ that underapproximates $\foreignWp{wlp}{\stmtWhile{b}{\pcc}}{\hla}$, i.e.
\[
   \vc{\downTransWlp{\stmtWhile{b}{\pcc}}}(\hla) \heyloleq 
  \begin{cases}
  	\loopinv, & \text{if } \loopinv \,\heyloleq\, \left( \embed{b} \impl \foreignWp{wlp}{\pcc}{\loopinv} \right) 
  	                        \sqcap \left(\embed{\neg b} \impl \hla\right) \\
  	0, & \text{otherwise}
  \end{cases}
  \hheyloleq
  \foreignWp{wlp}{\ldots}{\hla}.
\] 
Before we go into details, we remark for readers familiar with classical deductive verification that our encoding is almost identical to standard loop encodings (cf.~\cite{mullerBuildingDeductiveProgram2019}).
Apart from the quantitative interpretation of statements, the only exception is the $\stmtValidate$ in line 3. %
 \begin{figure}[t]
 	\begin{minipage}{0.4\textwidth}
	\begin{align*}
		& \textcolor{col1}{\downAssert{\loopinv}}\symSemi \\
		& \textcolor{col2}{\downHavoc{\textit{variables}}}\symSemi \\
		& \textcolor{col3}{\stmtValidate}\symSemi \\
		& \textcolor{col4}{\stmtAssume{\loopinv}}\symSemi \\
		& \textcolor{col5}{\symIf~(b)~\blockStart} \\
		& \quad \textcolor{col6}{\downTransWlp{\pcc}\symSemi} \\
		& \quad \textcolor{col6}{\downAssert{\loopinv}\symSemi} \\
		& \quad \textcolor{col6}{\downAssume{\embed{\false}}} \\
		& \textcolor{col5}{\blockEnd ~\stmtElseStart} 
		~ \textcolor{col5}{\blockEnd} \quad \intersem{\hla}
	\end{align*}
	\caption{Encoding of Park Induction rule for underapproximating $\foreignWp{wlp}{\stmtWhile{b}{\pcc}}{\hla}$.}
	\label{fig:park-wlp-heyvl}
	\end{minipage}%
	\hfill%
	\begin{minipage}{0.5\textwidth}
\newcommand{\probChoice}{\mathit{tmp}}
	\begin{align*}
		& \textcolor{col1}{ \upAssert{\ohfiveExp{\mathit{len}(l)}} }\symSemi \\
		& \textcolor{col2}{ \upHavoc{l}\symSemi\upHavoc{\probChoice} }\symSemi \\
		& \textcolor{col3}{ \coValidate }\symSemi \\
		& \textcolor{col4}{ \coAssume{\ohfiveExp{\mathit{len}(l)}} }\symSemi \\
		& \textcolor{col5}{\symIf~(\mathit{len}(l) > 0)~\blockStart} \\
		&\quad \textcolor{col6}{\stmtDeclInit{\probChoice}{\bool}{\exprFlip{0.5}}} \\
		&\quad \textcolor{col6}{\symIf~(\probChoice)~\blockStart~\stmtAsgn{l}{\mathit{tail}(l)}~\blockEnd~\symElse~\blockStart~\Assert{0}~\blockEnd} \\
		&\quad \textcolor{col6}{\upAssert{\ohfiveExp{\mathit{len}(l)}}\symSemi \upAssume{\coembed{\false}}} \\
		& \textcolor{col5}{\}~\stmtElseStart ~ \blockEnd} \quad \intersem{1}
	\end{align*}
	\caption{Exemplary \HeyVL encoding overapproximating the $\symWp$ of a loop.}
	\label{fig:park-wp-heyvl}
	\end{minipage}
\end{figure}

It is instructive to go over the encoding in \Cref{fig:park-wlp-heyvl} step by step for a given initial state $\State$.
The following expanded version of the above equation's right-hand side serves as a roadmap:
\begin{align*}
  \textcolor{col1}{
  \loopinv(\State) ~\sqcap~ 
  }
  \textcolor{col2}{
  \inf_{\State' \in \States}
  }
  \textcolor{col3}{
  \begin{cases}
  	\textcolor{col4}{\infty,} & 
  	    \textcolor{col4}{\text{ if } \loopinv(\State') \lleq} 
  	    \textcolor{col5}{(\embed{b}(\State') \impl~}
  	    \textcolor{col6}{\foreignWp{wlp}{\pcc}{\loopinv}(\State')} 
  	    \textcolor{col5}{)~\sqcap~( \embed{\neg b}(\State') \impl~} 
  	    \textcolor{hintgray}{\hla(\State')}
  	    \textcolor{col5}{)}
  	    \\
  	0, & \text{ otherwise},
  \end{cases} 
  }
\end{align*}
Reading the \HeyVL code in \Cref{fig:park-wlp-heyvl} top-down then corresponds to reading the equation from left to right as indicated by the colors.
We first \textcolor{col1}{assert} that our underapproximation of the loop's $\symForeign{wlp}$ is at most $\textcolor{col1}{\loopinv}$. The remaining code will ensure that said underapproximation is exactly $\textcolor{col1}{\loopinv}$ whenever $\textcolor{col1}{\loopinv}$ is an inductive loop invariant; it will be $0$ otherwise.
Proving that $\textcolor{col1}{\loopinv}$ is an inductive loop invariant requires checking an inequality $\heyloleq$, where $\hlb \heyloleq \hlc$ holds iff $\hlb(\State') \leq \hlc(\State')$ for all states $\State'$.
We \textcolor{col2}{havoc} the values of all program variables such that the invariant check encoded afterward is performed for every evaluation of the program variables, i.e. for every state $\State'$.\footnote{An optimized encoding may only havoc those variables that are modified in the loop body. However, we opted to encode the rule as it is typically presented in the literature.} Moreover, \textcolor{col2}{havoc} picks the \emph{minimal} result of all those invariant checks.
The statement \enquote{$\textcolor{col1}{\loopinv}$ is an inductive loop invariant} is inherently qualitative. We thus \textcolor{col3}{validate} that the invariant check encoded next is a qualitative statement that can only have two results: $\infty$ if $\textcolor{col1}{\loopinv}$ is an inductive invariant and $\textcolor{col3}{0}$ if it is not.
To check if $\textcolor{col1}{\loopinv}$ is an inductive invariant for a fixed state $\State'$, we need to prove an inequality, namely that $\textcolor{col4}{\loopinv}(\State')$ lower bounds $\textcolor{col6}{\foreignWp{wlp}{\pcc}{\loopinv}(\State')}$ if loop guard $\textcolor{col5}{b}$ holds
and $\textcolor{hintgray}{\hla(\State')}$ if $\textcolor{col5}{b}$ does not hold.
We first use \textcolor{col4}{\stmtAssume{\loopinv}} to lower the threshold for the expected value of the remaining code to be considered $\infty$ to $\textcolor{col4}{\loopinv}(\State')$.
Hence, we obtain $\textcolor{col4}{\infty}$ if the invariant check succeeds for $\State'$.
The \textcolor{col5}{conditional choice} is the invariant check's right-hand side.
If state $\State'$ satisfies $\textcolor{col5}{b}$, we use our existing $\symForeign{wlp}$ encoding to compute $\textcolor{col6}{\foreignWp{wlp}{\pcc}{\loopinv}(\State')}$, where $\textcolor{col6}{\downAssert{\loopinv}\symSemi \downAssume{\embed{\false}}}$ ensures that $\symForeign{wlp}$ is computed with respect to postexpectation $\textcolor{col6}{\loopinv}$.
If state $\State'$ satisfies $\textcolor{col5}{\neg b}$, we do nothing and just take the postexpectation $\textcolor{hintgray}{\hla}$.

\paragraph{Upper bounds.} 
Consider an iterative version of the lossy list traversal from \Cref{fig:intro:lossy} on page \pageref{fig:intro:lossy}:
\begin{align*} %
    &\headerWhile{\mathit{len}(l) > 0}~\{ %
    ~\stmtProb{0.5}{~\stmtAsgn{l}{\mathit{tail}(l)~}}{~\pDiverge~} %
    ~\}
\end{align*}
\noindent
The Park induction rule can also be used to \emph{over}approximate weakest preexpectations.
The encoding is dual, i.e. it suffices to use the \emph{co}-versions of the involved statements.
For example, \Cref{fig:park-wp-heyvl} encodes the above loop
with $\ohfiveExp{\mathit{len}(l)}$ as inductive invariant overapproximating the loop's termination probability. 
The list type and the exponential function $\ohfiveExp{\mathit{len}(l)}$ are represented in \HeyLo by custom domain declarations (cf.\ \Cref{sec:domain-decl}). 

\paragraph{Recursion.}
We can encode verification of $\symWlp$-lower bounds for recursive procedure calls of \pGCL programs as discussed in \Cref{sec:heyvl_proc_calls} and justified by \citet{olmedoReasoningRecursiveProbabilistic2016} and \citet{mathejadiss} -- it is another application of Park induction. For $\symWp$-upper bounds, the encoding is dual.
Hence, \Cref{fig:lossy-heyvl} on page \pageref{fig:lossy-heyvl} encodes that the termination probability of the program in \Cref{fig:intro:lossy} is at most $0.5^{\mathit{len}(l)}$.

\subsection{Overview of Encodings}
\Cref{tbl:case-studies-overview} summarizes all verification techniques -- program logics and proof rules -- that have been encoded in \HeyVL.
While a detailed discussion is beyond the scope of this paper, we briefly go over \Cref{tbl:case-studies-overview}.
The main takeaway is that
\HeyVL enables the encoding -- and thus automation -- of advanced verification methods based on diverse theoretical foundations and targeting different verification problems.
The practicality of our encodings will be evaluated in \Cref{sec:implementation}.
\begin{table}[t]
\caption{%
Verification techniques encoded in \HeyVL sorted by %
verification problem:
lower- and upper bounds on probability of events (LPROB and UPROB), 
upper- and lower bounds on expected values (UEXP and LEXP), 
conditional expected values (CEXP),
almost-sure termination (AST),
positive almost-sure termination (PAST), 
upper bounds on expected runtimes (UERT),
and lower bounds on expected runtimes (LERT).
}\label{tbl:case-studies-overview}
\begin{tabular}{lllll}
 	\toprule
	\textbf{Problem} & \textbf{Verification Technique} & \textbf{Source} & \ifdefined\WITHAPPENDIX\textbf{Encoding}\fi \\ 
	\midrule
	\multirow{2}{*}{LPROB} 
	& $\symWlp$ + Park induction & \citet{mciverAbstractionRefinementProof2005} 
	& \ifdefined\WITHAPPENDIX\cref{sec:encodings:park}\fi \\
	& $\symWlp$ + latticed $k$-induction & (new?) 
	& \ifdefined\WITHAPPENDIX\cref{sec:app-encodings-kind-wlp}\fi \\[1ex]
	UPROB & $\symWlp$ + $\omega$-invariants & \citet{kaminskiAdvancedWeakestPrecondition2019} 
	& \ifdefined\WITHAPPENDIX\cref{sec:app-enc-omega-wlp}\fi \\[1ex]
	\multirow{2}{*}{UEXP} 
	& $\symWp$ + Park induction & \citet{mciverAbstractionRefinementProof2005} 
	& \ifdefined\WITHAPPENDIX\cref{sec:app-enc-kind-wp}\fi \\
	& $\symWp$ + latticed $k$-induction & \citet{batzLatticedKInductionApplication2021} 
	& \ifdefined\WITHAPPENDIX\cref{sec:app-enc-kind-wp}\fi \\[1ex]
	\multirow{3}{*}{LEXP} 
	& $\symWp$ + $\omega$-invariants & \citet{kaminskiAdvancedWeakestPrecondition2019} 
	& \ifdefined\WITHAPPENDIX\cref{sec:app-enc-omega-wp}\fi \\
	& $\symWp$ + Optional Stopping Theorem & \citet{harkAimingLowHarder2019} 
	& \ifdefined\WITHAPPENDIX\cref{sec:app-enc-ost-wp}\fi \\[1ex]
	\multirow{1}{*}{CEXP} 
	& conditional $\symWp$ & \citet{olmedoConditioningProbabilisticProgramming2018} 
	& \ifdefined\WITHAPPENDIX\Cref{sec:encodings:loop-free}\fi \\
	\multirow{1}{*}{UERT}  
	& ert calculus + UEXP rules & \citet{kaminskiWeakestPreconditionReasoning2016} 
	& \ifdefined\WITHAPPENDIX\cref{sec:app-encodings-ert}\fi \\
	\multirow{1}{*}{LERT} 
	& ert calculus + $\omega$-invariants & \citet{kaminskiWeakestPreconditionReasoning2016} 
	& \ifdefined\WITHAPPENDIX\cref{sec:app-encodings-ert}\fi \\
	\multirow{1}{*}{AST} 
	& parametric super-martingale rule & \citet{mciverNewProofRule2018} 
	& \ifdefined\WITHAPPENDIX\Cref{sec:app-enc-new-proof-rule-ast}\fi \\
	\multirow{1}{*}{PAST} 
	& program analysis with martingales 
	& \ifdefined\WITHAPPENDIX
	    	\parbox{10em}{\citet{chakarovProbabilisticProgramAnalysis2013}}
	   \else
     	   	\citet{chakarovProbabilisticProgramAnalysis2013}
	   \fi
	& \ifdefined\WITHAPPENDIX\cref{sec:app-enc-past}\fi \\
	\bottomrule 
\end{tabular}
\end{table}

\paragraph{Expected Values}
We encoded \citet{mciverAbstractionRefinementProof2005}'s weakest (liberal) preexpectation calculus for analyzing expected values of probabilistic programs (cf. \Cref{sec:encodings:loop-free}). 
To analyze \emph{conditional} expected values, we combined the two calculi as suggested by \citet{olmedoConditioningProbabilisticProgramming2018}.
For loops, we encoded three proof rules based on domain theory:

First, \emph{Park Induction} generalizes the standard loop rule from Hoare logic~\cite{hoareAxiomaticBasisComputer1969} to a quantitative setting; it can be applied to lower bound weakest liberal preexpectations and upper bound weakest preexpectations (cf. \Cref{sec:encodings:park}).
However, it is unsound for the converse directions. 

Second, \emph{$\omega$-Invariants} are sound and complete for proving lower and upper bounds. However, they are arguably more complex because users must provide a family of invariants and compute limits. We modeled families of invariants as \HeyLo formulas with additional free variables and used $\stmtHavoc{x}$ and $\coHavoc{x}$ to represent limits.

Third, we encoded a quantitative version of %
\emph{$k$-induction} (for proving upper bounds) -- an established verification technique (cf. \cite{sheeranCheckingSafetyProperties2000}). 
The encodings are based on latticed $k$-induction~\cite{batzLatticedKInductionApplication2021}, a generalization of $k$-induction to arbitrary complete lattices.
After encoding $k$-induction for upper bounds on $\symWp$, we benefited from the duality of \HeyVL statements: 
we obtained a dual encoding for lower bounds on $\symWlp$ that has, to our knowledge, not been implemented before.
Furthermore, we encoded an advanced proof rule for lower bounds on expected values by \citet{harkAimingLowHarder2019}.
In contrast to the above rules, this rule is based on stochastic processes, particularly the Optional Stopping Theorem.
Using our encoding, we automated the main examples in \cite{harkAimingLowHarder2019}.

\paragraph{Expected Runtimes}
To analyze the performance of randomized algorithms, we encoded the expected runtime calculus by~\citet{kaminskiWeakestPreconditionReasoning2016,kaminskiWeakestPreconditionReasoning2018}
and its recent extension to amortized analysis~\cite{batzCalculusAmortizedExpected2023}.
Although reasoning about expected runtimes of loops involves some subtleties, we could adapt our \HeyVL encodings for expected values 
by inserting $\symTick$ statements.
We encoded and automated examples from \cite{kaminskiWeakestPreconditionReasoning2016,kaminskiWeakestPreconditionReasoning2018} and \cite{ngoBoundedExpectationsResource2018}.

\paragraph{Almost-Sure Termination (AST)}
\citet{mciverNewProofRule2018} proposed a proof rule for almost-sure termination -- does a probabilistic program terminate with probability one? 
The rule is based on a parametric martingale that must satisfy four conditions, which we encoded in separate \HeyVL (co)procedures.
We automated the verification of their examples, including the one in  \Cref{fig:ast-rule}.

\emph{Positive Almost-Sure Termination (PAST).}
PAST is a stronger notion than almost-sure termination, which requires a program's expected runtime to be finite.
We can apply our \HeyVL encodings for upper bounding expected runtimes to prove PAST.
Moreover, we encoded a dedicated proof rule for PAST by \citet{chakarovProbabilisticProgramAnalysis2013} based on martingales and concentration bounds.

%% file: sections/6_implementation.tex
\section{Implementation}\label{sec:implementation}
We first describe user-defined types and functions by means of \emph{domain declarations} in \Cref{sec:domain-decl}.
We then describe our tool \tool{Caesar} alongside empirical results validating the feasibility of our deductive verification infrastructure for the automated verification of probabilistic programs.

\subsection{Domain Declarations}
\label{sec:domain-decl}

Recall from \Cref{sec:heylo} that we assume all type- and function symbols to be interpreted. In practice, we support custom first-order theories via \emph{domain declarations} as is standard in classical deductive verification infrastructures \cite{viper}.
A domain declaration introduces a new type symbol alongside a set of typed function symbols and first-order formulae (called \emph{axioms}) characterizing feasible interpretations of the type- and function symbols.

Consider the harmonic numbers --- often required for, e.g., expected runtime analysis --- as an example. The $n$-th harmonic number is given by $H_n = \sum_{k=1}^n \frac{1}{k}$. To enable reasoning about verification problems involving the harmonic numbers, we introduce the following domain declaration:
\begin{align*}
    \texttt{domain}~HarmonicNums~\{\qquad &\quad \texttt{func}~H(\varandtype{n}{\Nats}){:}~ \PosReals \\
    &\quad \texttt{axiom}~h_0~H(0) = 0 \\
    &\quad \texttt{axiom}~h_n~\forall \typeof{n}{\Nats}.~H(n+1) = H(n) + \nicefrac{1}{n+1}
    \qquad \}
\end{align*}
$HarmonicNums$ introduces a new function symbol $\typeof{H}{\Nats \to \PosReals}$ and two axioms $h_0$ and $h_n$ characterizing feasible interpretations of $H$ recursively. Other non-linear functions such as exponential functions (e.g., $\ohfiveExp{n}$ from \Cref{sec:encodings:park}) as well as algebraic data types can be defined in a similar way (see, e.g., \cite{viperappendix}). In our implementation, validity of verification conditions --- inequalities between \HeyLo formulae --- is defined \emph{modulo} validity of all user-provided axioms.

\subsection{The Verifier \tool{\textbf{Caesar}}}
\label{sec:caesar}

We have implemented \HeyVL in our tool \tool{Caesar}\footnote{All tools and benchmarks are available as open-source software at \url{https://github.com/moves-rwth/caesar}.} which consists of approximately\ 10k lines of Rust code.
\tool{Caesar} takes as input a \HeyVL program $C$ and a set of domain declarations (cf. \Cref{sec:domain-decl}).
It then generates all verification conditions described by $C$, i.e., inequalities between \HeyLo formulae of the form $\hla \heyloleq \vc{S}(\hlb)$ or $\hla \heylogeq \vc{S}(\hlb)$, and translates these verification conditions to a Satisfiability Modulo Theories (SMT) query.
Our SMT back end is \tool{z3}~\cite{demouraZ3EfficientSMT2008}.
Since the translation to SMT can involve undecidable theories, \tool{Caesar} might return \emph{unknown}.
Otherwise, \tool{Caesar} either returns \emph{verified} or \emph{not verified}.
In the latter case, \tool{z3} often reports a counterexample state witnessing the violation of one of the verification conditions, which helps, e.g., debugging loop invariants.

Moreover, we have implemented a \emph{prototypical front-end} that translates (numeric) $\pGCL$ programs and their specifications to \HeyVL, and invokes \tool{Caesar} for automated verification.
Currently, it supports all techniques from \Cref{tbl:case-studies-overview} targeting loops.

\emph{SMT Encodings and Optimizations.} 
We translate validity of inequalities between \HeyLo formulae to SMT following the semantics of formulae from \Cref{fig:heylo-semantics}.

To encode the sort $\PosRealsInf$, we evaluated two options, which are both supported by our implementation.
The first option represents every number of sort $\PosRealsInf$ as a pair $(r, \textit{isInfty})$, where $r$ is a real number and $\textit{isInfty}$ is a Boolean flag that is true if and only if the represented number is equal to $\infty$.
We add constraints $r \geq 0$ to ensure that $r$ is non-negative.
All operations on $\PosRealsInf$ are then defined over such pairs.
For example, the addition $(r_1, \textit{isInfty}_1) + (r_2, \textit{isInfty}_2)$ is defined as $(r_1 + r_2, \textit{isInfty}_1 \vee \textit{isInfty}_2)$.
For multiplication, we ensure that $0 \cdot \infty = 0$ -- a common assumption in probability theory.
The second option leverages Z3-specific data type declarations to specify values that are either infinite or non-negative reals.
We observed that the first option performs better overall and thus use it by default.

The $\Inf$- and $\Sup$ quantifiers are translated using the textbook definition of infima and suprema over $\PosRealsInf$, but are eliminated whenever possible using that for $A \subseteq \PosRealsInf$ and $r \in \PosRealsInf$, we have
\[
\sup A \leq r \quad\text{iff}\quad \forall a \in A \colon a\leq r
\qquad\quad\text{and dually}\qquad\quad
r\leq \inf A \quad\text{iff}\quad \forall a \in A \colon r \leq a~.
\]
Finally, we simplify sub-formulae by, e.g., rewriting $\embed{\bexpr} \sqcap \hlb$ to $0$ if $\bexpr$ is unsatisfiable.

\paragraph{Benchmarks} 
To validate whether our implementation is capable of verifying interesting quantitative properties of probabilistic programs, we have considered various verification problems taken from the literature. These benchmarks involve unbounded probabilistic loops or recursion and include quantitative correctness properties of communication protocols \cite{DBLP:conf/tacas/DArgenioKRT97,DBLP:conf/types/HelminkSV93} and randomised algorithms \cite{DBLP:journals/corr/abs-1304-1916,ProbabilisticGuardedCommands2005,kushilevitzRandomizedMutualExclusion1992}, bounds on expected runtimes of stochastic processes \cite{ngoBoundedExpectationsResource2018,kaminskiExpectedRuntimeAnalysis2020,kaminskiWeakestPreconditionReasoning2018}, proofs of \emph{positive} almost-sure termination \cite{chakarovProbabilisticProgramAnalysis2013} and proofs of almost-sure termination for the case studies provided in \cite{mciverNewProofRule2018}. For each of these benchmarks, 
\ifdefined\WITHAPPENDIX
we apply the \HeyVL encodings provided in \Cref{sec:encodings} and \cref{app:encodings},
\else
we wrote \HeyVL encodings, including the ones in \Cref{sec:encodings},
\fi 
and cover all verification techniques from \Cref{tbl:case-studies-overview}. 
    
\Cref{fig:benchmarks} summarizes the results of our benchmarks. For each benchmark, it provides the benchmark name, the verification problem, the encoded techniques (cf.~\cref{tbl:case-studies-overview}), the lines of HeyVL code (without comments), notable features, and running time.
For the running time, we also provide the shares of pruning, i.e.\ simplification of sub-formulae, and the final SAT check.
\ifdefined\WITHAPPENDIX
\Cref{tbl:case-studies-overview} together with the column ``Problem'' provides pointers to each benchmark's source and encoding.
\else 
Details about each benchmark's source and encoding are found in the technical report.
\fi
For latticed $k$-induction, we indicate the value of $k$ that was used for the encoding.
Benchmarks that use exponential functions (e.g.\ rabin, zeroconf) or harmonic numbers (e.g.\ ast) are marked with F1.
Benchmarks that use multiple possibly mixed (co)procedures are marked with F2.
Two examples encode verification of nested loops (feature F3).

The size of our benchmarks ranges from 19-270 lines of \HeyVL code.
85\% 
of our benchmarks (those shaded in gray) have been verified with %
our front-end; the remaining encodings are handcrafted.
All benchmark files are available as part of our artifact.

\paragraph{Evaluation}

 On average, \tool{Caesar} needs 0.2 seconds to verify a \HeyVL program, with a maximum of 2.3 seconds.
Most benchmarks verify within less than a second.
The brp3 benchmark times out because of the large nested branching resulting from the exponential size of the $k$-induction encoding with $k=23$.
 
We conclude that \tool{Caesar} is capable of verifying interesting quantitative verification problems of probabilistic programs taken from the literature.
Moreover, we conclude that modern SMT solvers are a suitable back-end despite the fact that our benchmarks often require reasoning about highly non-linear functions.
This is due to the fact that it often suffices to (un)fold recursive definitions of, e.g., the harmonic numbers, finitely many times.
 Finally, our benchmarks demonstrate that our verification infrastructure provides a unifying interface for \emph{encoding and solving} various kinds of probabilistic verification problems in an automated manner.

\begin{table}
    \centering

    \caption{Benchmarks. 
        Rows shaded in gray indicate \HeyVL examples automatically generated from pGCL code with annotations using our frontend.
        Timeout (TO) was set to 10 seconds. 
        Verification techniques correspond to those presented in \cref{tbl:case-studies-overview}. 
        Lines of HeyVL code (LOC) are counted without comments. 
        Features: user-defined uninterpreted functions (F1), multiple (co)procedures (F2), nested loops (F3).}
    \label{fig:benchmarks}

    \adjustbox{max width=\textwidth}{%
        \begin{tabular}{l|l|l|r|l|r|r|r}
            Name                                             & Problem & Verification Technique                     & LOC & Features & Total (s) & Pruning & SAT    \\ \hline\hline
            \rowcolor{gray!30} rabin                       & LPROB   & $\symWlp$ + Park induction                 & 43  & F1, F3   & $0.33$    & $3\%$   & $96\%$ \\
            \rowcolor{gray!30} unif\_gen1                  & LPROB   & $\symWlp$ + Latticed $k$-induction ($k=2$) & 61  &          & $0.02$    & $52\%$  & $35\%$ \\
            \rowcolor{gray!30} unif\_gen2                  & LPROB   & $\symWlp$ + Latticed $k$-induction ($k=3$) & 82  &          & $0.05$    & $68\%$  & $25\%$ \\
            \rowcolor{gray!30} unif\_gen3                  & LPROB   & $\symWlp$ + Latticed $k$-induction ($k=3$) & 82  &          & $0.05$    & $71\%$  & $22\%$ \\
            \rowcolor{gray!30} unif\_gen4                  & LPROB   & $\symWlp$ + Latticed $k$-induction ($k=5$) & 124 &          & $0.86$    & $90\%$  & $7\%$  \\
            \rowcolor{gray!30} rabin1                      & LPROB   & $\symWlp$ + Park induction                 & 36  &          & $0.01$    & $45\%$  & $40\%$ \\
            \rowcolor{gray!30} rabin2                      & LPROB   & $\symWlp$ + Latticed $k$-induction ($k=5$) & 116 &          & $0.08$    & $27\%$  & $67\%$ \\
            chain                                            & UEXP    & $\symWp$ + Park induction                  & 28  & F1       & $0.03$    & $24\%$  & $66\%$ \\
            ohfive                                           & UEXP    & $\symWp$ + Park induction                  & 34  & F1, F3   & $0.02$    & $33\%$  & $56\%$ \\
            \rowcolor{gray!30} brp1                        & UEXP    & $\symWp$ + Latticed $k$-induction ($k=5$)  & 72  &          & $0.03$    & $45\%$  & $42\%$ \\
            \rowcolor{gray!30} brp2                        & UEXP    & $\symWp$ + Latticed $k$-induction ($k=11$) & 138 &          & $0.46$    & $70\%$  & $16\%$ \\
            \rowcolor{gray!30} brp3                        & UEXP    & $\symWp$ + Latticed $k$-induction ($k=23$) & 270 &          & TO        &         &        \\
            \rowcolor{gray!30} geo1                        & UEXP    & $\symWp$ + Latticed $k$-induction ($k=2$)  & 32  &          & $0.02$    & $44\%$  & $41\%$ \\
            geo (recursive)                                  & UEXP    & $\symWp$ + Park induction                  & 19  &          & $0.02$    & $43\%$  & $42\%$ \\
            \rowcolor{gray!30} rabin1                      & UEXP    & $\symWp$ + Park induction                  & 36  &          & $0.02$    & $44\%$  & $73\%$ \\
            \rowcolor{gray!30} rabin2                      & UEXP    & $\symWp$ + Latticed $k$-induction ($k=5$)  & 116 &          & $0.12$    & $22\%$  & $46\%$ \\
            \rowcolor{gray!30} unif\_gen1                  & UEXP    & $\symWp$ + Latticed $k$-induction ($k=2$)  & 61  &          & $0.03$    & $44\%$  & $46\%$ \\
            \rowcolor{gray!30} unif\_gen2                  & UEXP    & $\symWp$ + Latticed $k$-induction ($k=3$)  & 82  &          & $0.11$    & $41\%$  & $53\%$ \\
            \rowcolor{gray!30} unif\_gen3                  & UEXP    & $\symWp$ + Latticed $k$-induction ($k=3$)  & 82  &          & $0.10$    & $41\%$  & $53\%$ \\
            \rowcolor{gray!30} unif\_gen4                  & UEXP    & $\symWp$ + Latticed $k$-induction ($k=5$)  & 124 &          & $2.26$    & $47\%$  & $49\%$ \\
            zeroconf                                         & UEXP    & $\symWp$ + Park induction                  & 43  & F1, F2   & $0.03$    & $36\%$  & $49\%$ \\
            \rowcolor{gray!30} ost                         & LEXP    & $\symWp$ + Optional Stopping Theorem       & 93  & F2       & $0.07$    & $33\%$  & $51\%$ \\
            die                                              & CEXP    & conditional $\symWp$                       & 22  & F2       & $0.02$    & $17\%$  & $63\%$ \\
            \rowcolor{gray!30} 2drwalk                     & UERT    & $\symErt$ + Park induction                 & 224 &          & $0.02$    & $41\%$  & $44\%$ \\
            \rowcolor{gray!30} {\small{}bayesian\_network} & UERT    & $\symErt$ + Park induction                 & 107 &          & $0.02$    & $45\%$  & $40\%$ \\
            \rowcolor{gray!30} C4b\_t303                   & UERT    & $\symErt$ + Latticed $k$-induction ($k=3$) & 73  &          & $0.03$    & $29\%$  & $58\%$ \\
            \rowcolor{gray!30} condand                     & UERT    & $\symErt$ + Park induction                 & 24  &          & $0.02$    & $42\%$  & $42\%$ \\
            \rowcolor{gray!30} fcall                       & UERT    & $\symErt$ + Park induction                 & 26  &          & $0.02$    & $52\%$  & $44\%$ \\
            \rowcolor{gray!30} hyper                       & UERT    & $\symErt$ + Park induction                 & 31  &          & $0.02$    & $41\%$  & $44\%$ \\
            \rowcolor{gray!30} linear01                    & UERT    & $\symErt$ + Park induction                 & 23  &          & $0.02$    & $42\%$  & $43\%$ \\
            \rowcolor{gray!30} prdwalk                     & UERT    & $\symErt$ + Park induction                 & 62  &          & $0.02$    & $56\%$  & $31\%$ \\
            \rowcolor{gray!30} prspeed                     & UERT    & $\symErt$ + Park induction                 & 45  &          & $0.02$    & $41\%$  & $45\%$ \\
            \rowcolor{gray!30} rdspeed                     & UERT    & $\symErt$ + Park induction                 & 48  &          & $0.02$    & $38\%$  & $47\%$ \\
            \rowcolor{gray!30} rdwalk                      & UERT    & $\symErt$ + Park induction                 & 24  &          & $0.02$    & $42\%$  & $43\%$ \\
            \rowcolor{gray!30} sprdwalk                    & UERT    & $\symErt$ + Park induction                 & 26  &          & $0.02$    & $42\%$  & $43\%$ \\
            \rowcolor{gray!30} omega                       & LERT    & $\symErt$ + $\omega$-invariants            & 33  & F2       & $0.02$    & $42\%$  & $47\%$ \\
            \rowcolor{gray!30} ast1                        & AST     & parametric super-martingale rule           & 67  & F2       & $0.06$    & $33\%$  & $49\%$ \\
            \rowcolor{gray!30} ast2                        & AST     & parametric super-martingale rule           & 79  & F2       & $0.05$    & $38\%$  & $50\%$ \\
            ast3                                             & AST     & parametric super-martingale rule           & 65  & F1, F2   & $1.94$    & $1\%$   & $99\%$ \\
            \rowcolor{gray!30} ast4                        & AST     & parametric super-martingale rule           & 55  & F2       & $0.05$    & $33\%$  & $52\%$ \\
            \rowcolor{gray!30} past                        & PAST    & program analysis with martingales          & 26  & F2       & $0.04$    & $40\%$  & $46\%$ \\
        \end{tabular}}
\end{table}

%% file: sections/7_related_work.tex
\section{Related Work}\label{sec:related-work}
We focus on automated verification techniques for probabilistic programs and deductive verification infrastructures for non-probabilistic programs; encoded proof rules have been discussed in \Cref{sec:encodings}.

\emph{Probabilistic Program Verification.}
Expectation-based probabilistic program verification has been pioneered by \citet{kozenProbabilisticPDL1983,kozenProbabilisticPDL1985} and McIver \& Morgan \cite{mciverAbstractionRefinementProof2005}. 
 \citet{ProbabilisticGuardedCommands2005} formalised the $\textsf{w(l)p}$ calculus in \emph{Isabelle/HOL} \cite{isbabelle}. They focus on the calculus' meta theory and provide a verification-condition generator for proving partial correctness. %
 \citet{hoelzlert} implemented the meta theory of \citet{kaminskiWeakestPreconditionReasoning2016}'s $\symErt$ calculus in \emph{Isabelle/HOL} and verified bounds on expected runtimes of randomised algorithms. %
 We focus on unifying verification techniques in a single infrastructure.

\emph{Easycrypt} \cite{easycrypt1,easycrypt2} is a theorem prover for verifying cryptographic protocols, featuring libraries for data structures and algebraic reasoning.
\emph{Ellora} \cite{bartheAssertionBasedProgramLogic2018} is an assertion-based program logic for probabilistic programs implemented in \emph{Easycrypt}, benefiting from \emph{Easycrypt}'s features.
Their specifications are predicates over (sub)distributions instead of expectations. While \emph{Ellora} employs \emph{specialised} proof rules for loops and does not support non-determinism or recursion, thus being more restrictive than \HeyVL in this regard, \emph{Ellora} embeds, e.g., logics for reasoning about probabilistic independence. As stated in \cite{bartheAssertionBasedProgramLogic2018}, an in-depth comparison of assertion- and expectation-based approaches is difficult. %
  \citet{pardoSpecificationLogicPrograms2022} propose a propositional dynamic logic for $\pGCL$ featuring reasoning about convergence of estimators. Their logic is not automated yet.

\emph{Fully automatic} analyses of probabilistic programs are limited to specific properties, e.g. bounding expected runtimes or proving (positive) almost-sure termination 
\cite{koat,leutgebAutomatedExpectedAmortised2022,ngoBoundedExpectationsResource2018,batzHowLongBayesian2018,batzProbabilisticProgramVerification2022,amber,automatedterm,termpos,stochterm,termnondet,termcompletesound,cegisupast,modularcost}.
We might also benefit from invariant synthesis approaches~\cite{batzProbabilisticProgramVerification2022,batzPrIC3PropertyDirected2020,symbolicexec,invsys1,invsys2,invsys3,invsys4,chakarovProbabilisticProgramAnalysis2013,rankingsuper,invsys7,moments1,moment3}. %

\emph{Deductive Verification Infrastructures.}
\tool{Boogie}~\cite{leinoThisBoogie2008} and \tool{Why3}~\cite{DBLP:conf/esop/FilliatreP13} are prominent examples of IVLs for non-probabilistic programs that lie at the foundation of various modern verifiers, such as \tool{Dafny}~\cite{leinoDafnyAutomaticProgram2010} and \tool{Frama-C}~\cite{DBLP:journals/fac/KirchnerKPSY15}.
Neither of these IVLs targets reasoning about expectations or upper bounds (aka necessary preconditions~\cite{cousotPreconditionInferenceIntermittent2011}).
For example, \tool{Boogie}'s statements are specific to verifying lower bounds on Boolean predicates. 
Evaluating whether our implementation could benefit from encoding \HeyLo formulae into \tool{Why3} is interesting future work.

%% file: sections/8_conclusion.tex
\section{Conclusion and Future Work}\label{sec:conclusion}

We have presented a verification infrastructure for probabilistic programs based on a novel quantitative intermediate verification language that aids researchers with prototyping and automating their proof rules.
As future work, we plan to automate more rules and explore the relationship between our language, particularly its dual operators, and (partial) incorrectness logic~\cite{ohearnIncorrectnessLogic2020,zhangQuantitativeStrongestPost2022}. A further promising direction is to generalize our infrastructure for the verification of probabilistic pointer programs~\cite{batzFoundationsEntailmentChecking2022,batzQuantitativeSeparationLogic2019} and weighted programs~\cite{batzWeightedProgrammingProgramming2022}.

Furthermore, establishing a formal ``ground truth'' for our intermediate language \HeyVL in terms of an operational semantics that assigns precise meaning to quantitative Hoare triples, which we admittedly introduced ad-hoc, is important future work.
However, defining an operational semantics that yields a \emph{pleasant forward-reading} intuition for all statements in our intermediate language \HeyVL appears non-trivial. In particular, we are unaware of a semantics for (co)assume statements that is independent of the semantics of the remaining program. We believe that stochastic games might be an adequate formalism but the details have not been worked out yet.

%% file: sections/9_data_availability.tex
\section*{Data-Availability Statement}

The tool \tool{Caesar}, our prototypical front-end for \pGCL programs, as well as our benchmarks that we submitted for the artifact evaluation are available~\cite{philipp_schroer_2023_8146987}.
We also develop our tools as open-source software at \url{https://github.com/moves-rwth/caesar}.

%% file: appendices/proofs_heylo.tex
\section{Omitted Proofs: \HeyLo}
\label{sec:app-proofs-heylo}
\begingroup\allowdisplaybreaks[1]

\thmHeyloAdjoint*

\begin{proof}
    Let $\hla,\hlb,\hlc \in \HeyLo$. For the adjointness of $\impl$ and $\sqcap$, consider the following:
    \begin{align*}
        &\hla \sqcap \hlb \heyloleq \hlc \\
        &\qiff \forall \State \in \States.~ \min \Set{ \evalState{\hla},~ \evalState{\hlb} } \leq \evalState{\hlc} \tag{\Cref{fig:heylo-semantics}} \\
        &\qiff \forall \State \in \States.~ \evalState{\hla} \leq \evalState{\hlc} \lor \evalState{\hlb} \leq \evalState{\hlc} \\
        &\qiff \forall \State \in \States.~ \ifThenElse{\evalState{\hlb} \leq \evalState{\hlc}}{\exprTrue}{\evalState{\hla} \leq \evalState{\hlc}} \\
        &\qiff \forall \State \in \States.~ \evalState{\hla} \leq \ifThenElse{\evalState{\hlb} \leq \evalState{\hlc}}{\infty}{\evalState{\hlc}} \\
        &\qiff \hla \heyloleq \hlb \impl \hlc. \tag{\Cref{fig:heylo-semantics}}
    \end{align*}
    The proof of adjointness of the coimplication $\coimpl$ and $\sqcup$ is analogous.
\end{proof}

\thmHeyloDeduction*

\begin{proof}
    Let $\hla,\hlb \in \HeyLo$. Then,
    \begin{align*}
        &\hla \heyloleq \hlb \\
            &\qiff \forall \State \in \States.~ \evalState{\hla} \leq \evalState{\hlb} \tag{definition of $\heyloleq$} \\
            &\qimplies \forall \State \in \States.~ \evalState{\hla \impl \hlb} = \infty \tag{\Cref{fig:heylo-semantics}} \\
            &\qiff \hla \impl \hlb \text{ is valid} \tag{definition of validity} \\[1em]
        & \text{and} \\[1em]
        &\hla \impl \hlb \text{ is valid} \\
            &\qiff \forall \State \in \States.~ \evalState{\hla \impl \hlb} = \infty \tag{definition of validity} \\
            &\qimplies \forall \State \in \States.~ \ifThenElse{\evalState{\hla} \leq \evalState{\hlb}}{\infty = \infty}{\evalState{\hlb} = \infty} \tag{\Cref{fig:heylo-semantics}} \\
            &\qimplies \hla \heyloleq \hlb.
    \end{align*}
    The proof of the second equivalence is analogous.
\end{proof}

\endgroup

%% file: appendices/proofs_heyvl.tex
\section{Omitted Proofs: \HeyVL}
\label{sec:app-proofs-heyvl}
\begingroup\allowdisplaybreaks[1]

\subsection{Properties of HeyVL}

Recall \cref{thm:monotonicity} (\cpageref{thm:monotonicity}):

\thmHeyvlMonotonicity*

\begin{proof}
    Let $S \in \HeyVL$.
    We do a proof by induction over the structure of $\stmt$ to show
    \[
        \forall \hla,\hla' \in \HeyLo.\quad \hla \heyloleq \hla' \qimplies \vc{\stmt}(\hla) \heyloleq \vc{\stmt}(\hla')~.
    \]

    For the base cases, let $\hla,\hla' \in \HeyLo$ such that $\hla \heyloleq \hla'$.
    \begin{itemize}[nosep]
        \proofcase{$S = \stmtDeclInit{x}{\typevar}{\mu}}$
            \begin{align*}
                &\vc{\stmtDeclInit{x}{\typevar}{\mu}}(\hla) \\
                &= p_1 \cdot \hla\substBy{x}{\termvar_1} + \ldots + p_n \cdot \hla\substBy{x}{\termvar_n} \tag{definition} \\
                &\heyloleq p_1 \cdot \hla'\substBy{x}{\termvar_1} + \ldots + p_n \cdot \hla'\substBy{x}{\termvar_n} \tag{$\hla \heyloleq \hla'$} \\
                &= \vc{\stmtDeclInit{x}{\typevar}{\mu}}(\hla') \tag{definition}
            \end{align*}
        \proofcase{$S = \stmtTick{a}$}
            \begin{align*}
                &\vc{\stmtTick{a}}(\hla) \\
                &= \hla + a \tag{definition of $\symTick$} \\
                &\heyloleq \hla' + a \tag{$\hla \heyloleq \hla'$} \\
                &= \vc{\stmtTick{a}}(\hla') \tag{definition of $\symTick$}
            \end{align*}
        \proofcase{$S = \downAssert{\hlb}$}
            \begin{align*}
                &\vc{\downAssert{\hlb}}(\hla) \\
                &= \hlb \sqcap \hla \tag{definition of $\symAssert$}  \\
                &\heyloleq \hlb \sqcap \hla' \tag{$\hla \heyloleq \hla'$} \\
                &= \vc{\downAssert{\hlb}}(\hla') \tag{definition of $\symAssert$}
            \end{align*}
        \proofcase{$S = \downAssume{\hlb}$} For all $\State \in \States$,
            \begin{align*}
                &\vc{\downAssume{\hlb}}(\hla)(\State) \\
                &= \evalState{\hlb \impl \hla} \tag{definition of $\symAssume$} \\
                &= \ifThenElse{\evalState{\hlb} \leq \evalState{\hla}}{\infty}{\evalState{\hla}} \tag{definition of $\impl$}\\
                &\leq \ifThenElse{\evalState{\hlb} \leq \evalState{\hla}}{\infty}{\evalState{\hla'}} \tag{$\evalState{\hla} \leq \evalState{\hla'}$} \\
                &\leq \begin{cases}
                        \infty, & \text{if } \evalState{\hlb} \leq \evalState{\hla}\\
                        \evalState{\hla'}, & \text{if } \evalState{\hla} < \evalState{\hlb} \leq \evalState{\hla'}\\
                        \evalState{\hla'}, & \text{otherwise} 
                    \end{cases} \tag{case distinction} \\
                &\leq \ifThenElse{\evalState{\hlb} \leq \evalState{\hla'}}{\infty}{\evalState{\hla'}} \tag{$\evalState{\hla'} \leq \infty$} \\
                &= \evalState{\hlb \impl \hla'} \tag{definition of $\impl$}\\
                &= \vc{\downAssume{\hlb}}(\hla')(\State) \tag{definition of $\symAssume$}
            \end{align*}
        \proofcase{$S = \downHavoc{x}$}
            \begin{align*}
                &\vc{\downHavoc{x}}(\hla) \\
                &= \inf \Set{ \hla\substBy{x}{v} | v \in \Vals } \tag{definition of $\symHavoc$} \\
                &\heyloleq \inf \Set{ \hla'\substBy{x}{v} | v \in \Vals } \tag{$\hla \heyloleq \hla'$} \\
                &= \vc{\downHavoc{x}}(\hla')
            \end{align*}
        \proofcase{$S = \stmtValidate$} For all $\State \in \States$,
            \begin{align*}
                &\vc{\stmtValidate}(\hla)(\State) \\
                &= \evalState{\heylovalidate{\hla}} \tag{definition of $\symValidate$} \\
                &= \ifThenElse{\evalState{\hla} = \infty}{\infty}{0} \tag{definition of $\triangle$} \\
                &\leq \ifThenElse{\evalState{\hla'} = \infty}{\infty}{0} \tag{$\evalState{\hla} \leq \evalState{\hla'}$} \\
                &= \evalState{\heylovalidate{\hla'}} \tag{definition of $\triangle$} \\
                &= \vc{\stmtValidate}(\hla')(\State) \tag{definition of $\symValidate$}
            \end{align*}
    \end{itemize}

    The $\symUp$ cases are dual, but we show the $\symUp\symAssume$ case for illustration:
    \begin{itemize}[nosep]
        \proofcase{$S = \upAssume{\hlb}$} For all $\State \in \States$,
        \begin{align*}
            &\vc{\upAssume{\hlb}}(\hla)(\State) \\
            &= \evalState{\hlb \coimpl \hla} \tag{definition of $\symUp\symAssume$} \\
            &= \ifThenElse{\evalState{\hlb} \geq \evalState{\hla}}{0}{\evalState{\hla}} \tag{definition of $\coimpl$} \\
            &\leq \ifThenElse{\evalState{\hlb} \geq \evalState{\hla}}{0}{\evalState{\hla'}} \tag{$\evalState{\hla} \leq \evalState{\hla'}$} \\
            &\leq \begin{cases}
                0, & \text{if } \evalState{\hlb} \geq \evalState{\hla'}\\
                0, & \text{if } \evalState{\hla'} > \evalState{\hlb} \geq \evalState{\hla}\\
                \evalState{\hla'}, & \text{otherwise} 
            \end{cases} \tag{case distinction} \\
            &\leq \ifThenElse{\evalState{\hlb} \geq \evalState{\hla'}}{0}{\evalState{\hla'}} \tag{$\evalState{\hla'} \geq 0$} \\
            &= \evalState{\hlb \coimpl \hla'} \tag{definition of $\coimpl$ }\\
            &= \vc{\upAssume{\hlb}}(\hla')(\State) \tag{definition of $\symUp\symAssume$}
        \end{align*}
    \end{itemize}

    Now assume that the induction hypothesis holds for arbitrary but fixed $\stmt_1,\stmt_2 \in \HeyVL$.

    Induction step:
    \begin{itemize}[nosep]
        \proofcase{$S = \stmtAsgn{x_1,\ldots,x_n}{P(e_1,\ldots,e_m)}$} According to \Cref{sec:heyvl_proc_calls}, (co)procedure calls are encoded as a sequential composition of the atomic $\symAssert$, $\symHavoc$, $\symValidate$, and $\symAssume$ (co)statements and are thus covered by the following case $S = \stmtSeq{\stmt_1}{\stmt_2}$.
        \proofcase{$S = \stmtSeq{\stmt_1}{\stmt_2}$} \\
            Let $\hla,\hla' \in \HeyLo$ such that $\hla \heyloleq \hla'$.
            We use the induction hypothesis for $\stmt_2$:
            \[
                \vc{\stmt_2}(\hla) \heyloleq \vc{\stmt_2}(\hla')~.
            \]
            By the induction hypothesis for $S_1$:
            \[
                \vc{\stmt_1}(\vc{\stmt_2}(\hla)) \heyloleq \vc{\stmt_1}(\vc{\stmt_2}(\hla'))~.
            \]
            Applying definitions, we get:
            \begin{align*}
                \vc{\stmtSeq{\stmt_1}{\stmt_2}}(\hla)
                &= \vc{\stmt_1}(\vc{\stmt_2}(\hla)) \tag{definition of $\symSemi$} \\
                &\heyloleq \vc{\stmt_1}(\vc{\stmt_2}(\hla')) \tag{I.H. on $S_1$ and $S_2$} \\
                &= \vc{\stmtSeq{\stmt_1}{\stmt_2}}(\hla') \tag{definition of $\symSemi$}
            \end{align*}
        \proofcase{$S = \stmtDemonic{\stmt_1}{\stmt_2}$}
            \begin{align*}
                &\vc{\stmtDemonic{\stmt_1}{\stmt_2}}(\hla) \\
                &= \vc{\stmt_1}(\hla) \sqcap \vc{\stmt_2}(\hla) \tag{definition of $\symDemonic$} \\
                &\heyloleq \vc{\stmt_1}(\hla') \sqcap \vc{\stmt_2}(\hla) \tag{induction hypothesis} \\
                &\heyloleq \vc{\stmt_1}(\hla') \sqcap \vc{\stmt_2}(\hla') \tag{induction hypothesis} \\
                &= \vc{\stmtDemonic{\stmt_1}{\stmt_2}}(\hla') \tag{definition of $\symDemonic$}
            \end{align*}
        \proofcase{$S = \stmtAngelic{\stmt_1}{\stmt_2}$: Analogous to the $\symDemonic$ case}
    \end{itemize}

    By the principle of structural induction, \cref{thm:monotonicity} holds.
\end{proof}

Recall \cref{thm:heyvl-conservativity} (\cpageref{thm:heyvl-conservativity}):

\thmHeyvlConservativity*

\begin{proof}
    Let $C$ be a program in the Boolean IVL of~\cite{mullerBuildingDeductiveProgram2019}.
    Let $B \in \mathbb{P}$ be a predicate.
    We prove 
    \[
        \vc{\overline{C}}(\embed{B}) = \embed{\vcB{C}(B)}
    \]
    by induction over the structure of $C$.

    Base cases:
    \begin{itemize}[nosep]
        \proofcase{$C = \stmtDeclInit{x}{\typevar}{e}$}
            \begin{align*}
                &\embed{\vcB{\stmtDeclInit{x}{\typevar}{e}}(B)} \\
                &= \embed{B\substBy{x}{e}} \\
                &= 1 \cdot \embed{B}\substBy{x}{e} \\
                &= \vc{\stmtDeclInit{x}{\typevar}{\mu}}(\embed{B})
            \end{align*}
        \proofcase{$C = \stmtHavoc{x}$ where $\typeof{x}{\typevar}$}
            \begin{align*}
                &\embed{\vcB{\stmtHavoc{x}}(B)} \\
                &= \embed{\forall x \in \tau.~ B} \\
                &= \iquant{\typeof{x}{\typevar}}{\embed{B}} \\
                &= \vc{\stmtHavoc{x}}(\embed{B})
            \end{align*}
        \proofcase{$C = \stmtAssert{A}$}
            \begin{align*}
                &\embed{\vcB{\stmtAssert{A}}(B)} \\
                &= \embed{A \land B} \\
                &= \embed{A} \sqcap \embed{B} \\
                &= \vc{\stmtAssert{\embed{A}}}(\embed{B})
            \end{align*}
        \proofcase{$C = \stmtAssume{A}$}
            \begin{align*}
                &\embed{\vcB{\stmtAssume{A}}(B)} \\
                &= \embed{A \Rightarrow B} \\
                &= \embed{A} \impl \embed{B} \\
                &= \vc{\stmtAssume{\embed{A}}}(\embed{B})
            \end{align*}
    \end{itemize}

    Now assume that the induction hypothesis holds for arbitrary, but fixed $C_1,C_2$ in the Boolean IVL.
    Let $\overline{C_1}, \overline{C_2} \in \HeyVL$ be obtained from $C_1$ and $C_2$ by replacement of $\stmtAssert{A}$ and $\stmtAssume{A}$ by $\stmtAssert{\embed{A}}$ and $\stmtAssume{\embed{A}}$, respectively.

    Induction step:
    \begin{itemize}[nosep]
        \proofcase{$C = \stmtSeq{C_1}{C_2}$}
            \begin{align*}
                &\embed{\vcB{\stmtSeq{C_1}{C_2}}(B)} \\
                &= \embed{\vcB{C_1}(\vcB{C_2}(B))} \\
                &= \vc{\overline{C_1}}(\embed{\vcB{C_2}(B)}) \tag{induction hypothesis} \\
                &= \vc{\overline{C_1}}(\vc{\overline{C_2}}(\embed{B})) \tag{induction hypothesis} \\
                &= \vc{\overline{\stmtSeq{C_1}{C_2}}}(\embed{B})
            \end{align*}
        \proofcase{$C = \stmtDemonic{C_1}{C_2}$}
            \begin{align*}
                &\embed{\vcB{\stmtDemonic{C_1}{C_2}}(B)} \\
                &= \embed{\vcB{C_1}(B) \land \vcB{C_2}(B)} \\
                &= \embed{\vcB{C_1}(B)} \sqcap \embed{\vcB{C_2}(B)} \\
                &= \vc{\overline{C_1}}(\embed{B}) \sqcap \vc{\overline{C_2}}(\embed{B}) \tag{induction hypothesis} \\
                &= \vc{\stmtDemonic{\overline{C_1}}{\overline{C_2}}}(\embed{B}) \\
                &= \vc{\overline{\stmtDemonic{C_1}{C_2}}}(\embed{B})
            \end{align*}
    \end{itemize}

    By structural induction on $C$, \cref{thm:heyvl-conservativity} holds.
\end{proof}

\subsection{Soundness and Semantics of Procedure Calls}
\label{sec:app-proc-calls}

We want to encode a procedure call $\stmtAsgn{z_1,\ldots,z_m}{P(t_1,\ldots,t_n)}$ for a procedure $P$:
\begin{appendixprogram}
    &\proc{\procname}{\varandtype{x_1}{\typevar_1}, \ldots \varandtype{x_n}{\typevar_n}}{\varandtype{y_1}{\typevar_1}, \ldots \varandtype{y_m}{\typevar_m}} \\
    &\quad\Requires{\hlc}  \\
    &\quad\Ensures{\hlb} \\
    &\blockStart \\
    &\quad \sstmt \\
    &\}
\end{appendixprogram}

Recall the definition of $\sstmt_\mathit{encoding}$ for the above call and procedure from \cref{sec:heyvl_proc_calls}:
\begin{align*}
    \sstmt_\mathit{encoding}\colon \qquad  &
    \stmtAssert{\hhlc}\symSemi
    \stmtHavoc{z_1}\symSemi \ldots \symSemi\stmtHavoc{z_m}\symSemi
    \stmtValidate\symSemi
    \stmtAssume{\hhlb}.
\end{align*}

\thmProcCalls*

\begin{proof}
    For simplicity, we assume that $\procname$ has no local variables and that its output variables $y_1,\ldots,y_m$ have been renamed to $z_1,\ldots,z_m$.
    We fix arbitrary values for all remaining variables.
    In the following, $\States$ denotes the states obtained by modifying only $z_1,\ldots,z_m$, and inequalities between expectations below apply only on these states.
    First, we show that for all $\hhla \in \HeyLo$ and $\State \in \States$, we have
    \[
        \vc{\sstmt_\mathit{encoding}}(\hhla)(\State) = \ifThenElseDot{\hhlb \heyloleq \hhla}{\evalState{\hhlc}}{0}
    \]

    From the definition of $\symVc$ (\cref{fig:heyvl-semantics}), it follows that
    \begin{align*}
        &\vc{\sstmt_\mathit{encoding}}(\hhla)(\State) \\ 
        &= \evalState{\hhlc \sqcap \iquant{z_1}{\ldots\iquant{z_m}{\heylovalidate{\hhlb \impl \hhla}}}} \\
        &= \evalState{\hhlc} \sqcap \evalState{\iquant{z_1}{\ldots\iquant{z_m}{{\heylovalidate{\hhlb \impl \hhla}}}}} \tag{definition of $\sqcap$} \\
        &= \evalState{\hhlc} \sqcap \inf \{ \eval{\heylovalidate{\hhlb \impl \hhla}}(\State') \mid \State' \in \States \} \tag{definition of $\Inf$} \\
    \shortintertext{Let $\exprIte{a}{b}{c}$ denote a conditional choice that evaluates to $b$ if $a$ is true and to $c$ otherwise.}
        &= \evalState{\hhlc} \sqcap \inf \{ ~\exprIte{\eval{\hhlb \impl \hhla}(\State') = \infty}{\infty}{0}~ \mid \State' \in \States \} \tag{definition of $\triangle$} \\
        &= \evalState{\hhlc} \sqcap \inf \{ ~\exprIte{\eval{\hhlb}(\State') \leq \eval{\hhla}(\State')}{\infty}{0}~ \mid \State' \in \States \} \tag{cf.~\cref{thm:adjoint}} \\
    \shortintertext{The infimum evaluates to $\infty$ iff $\eval{\hhlb}(\State') \heyloleq \eval{\hhla}(\State')$ for all $\State' \in \States$. Thus,}
        &= \evalState{\hhlc} \sqcap \exprIte{\hhlb \heyloleq \hhla}{\infty}{0} \\
        &= \ifThenElse{\hhlb \heyloleq \hhla}{\evalState{\hhlc}}{0} \tag{$\evalState{\hhlc} \sqcap \infty = \evalState{\hhlc}$}
    \end{align*}

    \noindent
    Now we show for all $\hhla \in \HeyLo$ that $\vc{\sstmt_\mathit{encoding}}(\hhla) \heyloleq \vc{\sstmt}(\hhla)$.
    Let $\hhla \in \HeyLo$.

    \noindent
    In case that $\hhlb \heyloleq \hhla$ holds, we have by monotonicity of $\symVc$ (\cref{thm:monotonicity}):
    \[
        \vc{\sstmt}(\hhlb) \heyloleq \vc{\sstmt}(\hhla)~.
    \]
    From the assumption $\hhlc \heyloleq \vc{\sstmt}(\hhlb)$ it follows that
    \[
        \hhlc \heyloleq \vc{\sstmt}(\hhla)~.
    \]
    Thus, 
    \[
        \vc{\sstmt_\mathit{encoding}}(\hhla) = \hhlc \heyloleq \vc{\sstmt}(\hhla)~.
    \]
    If $\hhlb \heyloleq \hhla$ does not hold, we have
    \[
        \vc{\sstmt_\mathit{encoding}}(\hhla) = 0 \heyloleq \vc{\sstmt}(\hhla)~.  
    \]
    In conclusion,
    \[
        \vc{\sstmt_\mathit{encoding}}(\hhla) \heyloleq \vc{\sstmt}(\hhla)~.
    \]
    The other claim,
    \[
        \vc{\mathit{init}\symSemi\sstmt_\mathit{encoding}\symSemi\mathit{return}}(\hla)
        \hheyloleq  
        \vc{\mathit{init}\symSemi\sstmt\symSemi\mathit{return}}(\hla)~,
    \]
    follows by the above and the definition of $\symVc$.
\end{proof}

\endgroup

%% file: appendices/encodings.tex
\section{Proof Rule Encodings Into \HeyVL}\label{app:encodings}

This appendix section details the \HeyVL encodings mentioned in \cref{sec:encodings}.
These encodings are all implemented in our frontend that translates annotated pGCL programs to HeyVL.
We follow \cref{tbl:case-studies-overview} and present encodings for the various verification problems.
For each encoding, we first state the formal \emph{proof rule} on expectations.
Then, we specify the \emph{encoding inputs} that our frontend requires, as well as a schematic description of the \emph{encoding output}.
All encodings of loops require \HeyVL encodings of their loop bodies.
For loop-free programs, the encoding from \cref{sec:encodings:loop-free} can be used.
Furthermore, proof rule encodings from this section may be used to encode nested loops.

\subsection{Loop Rule: $k$-Induction for $\symWlp$}
\label{sec:app-encodings-kind-wlp}

The \emph{$k$-induction} encoding for $\symWlp$ encodes a while loop and under-approximates its $\symWlp$ semantics.
The proof rule is a generalization of Park induction (cf.~\cref{sec:encodings:park}).
For the $k$-induction encoding, the user needs to provide a potential subinvariant $I \in \HeyLo$ and a number $k \in \mathbb{N}_{\geq 1}$ of how many times to unfold the loop.
For the loop body, we assume another under-approximating encoding is given.
If it contains loops, $k$-induction can be encoded recursively, but other encodings can be used as well.
\medskip

\noindent
\textbf{Proof Rule:}
\emph{Latticed $k$-induction}~\cite{batzLatticedKInductionApplication2021} for $\symWlp$.\footnote{In~\cite{batzLatticedKInductionApplication2021}, latticed $k$-induction is only defined for upper bounds on least fixed points. These occur e.g.\ in $\symWp$ and $\symErt$ semantics. However, the dual principle can be applied to the greatest fixed point that underlies the $\symWlp$ semantics.} \\
Let $\CC = \WHILEDO{\bexpr}{\loopbody}$ be a pGCL loop and let $\expa \in \OneBoundedExpectations$.
The \emph{$k$-induction operator} for $\symWlp$ is given by
\[
	\symDownPsi_I \colon \OneBoundedExpectations \to \OneBoundedExpectations,\quad \expb \mapsto \wlpPhi_\expa(\expb) \sqcup I~,
\]
where the \emph{loop-characteristic functional} $\wlpPhi_\expa$ with respect to post $\expa$ is defined as
\[
	\wlpPhi_\expa(\expb) = \iverson{\bexpr} \cdot \wlp{\loopbody}(\expb) + \iverson{\neg b} \cdot \expa~.
\]
Then, for $k \in \mathbb{N}_{\geq 1}$,
\[
	I \expleq \wlpPhi_\expa(\symDownPsi_I^{k-1}(I)) \quad\text{implies}\quad I \expleq \wlp{\CC}(\expa)~.
\]
\medskip

\noindent
\textbf{Encoding Input:} 
\begin{itemize}
	\item pGCL loop $\CC = \WHILEDO{\bexpr}{\loopbody}$.
	\item HeyVL statement $\downTransWlp{\loopbody}$ that satisfies $\vc{\downTransWlp{\loopbody}} \heyloleq \wlp{\loopbody}$.
	\item $k \in \mathbb{N}_{\geq 1}$.
	\item Potential $k$-inductive $\symWlp$-subinvariant $I \in \HeyLo$ with $I \heyloleq 1$.
\end{itemize}
\medskip

\noindent
\textbf{Encoding Output:}
\begin{itemize}
	\item HeyVL statement $\stmt$ that satisfies $\vc{\stmt} \expleq \wlp{\CC}$.
		\begin{itemize}
			\item If $\loopbody$ is encoded exactly, i.e. $\vc{\downTransWlp{\loopbody}} = \wlp{\loopbody}$ holds, then
				\[
					\vc{\stmt} = \ifThenElseDot{I \expleq \wlpPhi_\expa(\symDownPsi_I^{k-1}(I))}{I}{0}
				\]
		\end{itemize}
\end{itemize}
\medskip

\noindent
The $k$-induction encoding is similar to Park induction, but the sequence $\Assert{I}\symSemi \Assume{\embed{\exprFalse}}$ in the Park induction encoding is replaced by recursive encodings of the $\symDownPsi_I$ operator.

\begin{samepage}
\noindent
Formally, the HeyVL statement $\stmt$ is given by:
\begin{align*}
	\stmt &=\quad \downSpec{I}{I}\symSemi \downIterWlp{\downExtendWlpK{k-1}{\const{I}}} \\
\shortintertext{where}
	\downSpec{\hlb}{\hla} &=\quad \Assert{\hlb}\symSemi \Havoc{\textit{variables}}\symSemi \Validate\symSemi \Assume{\hla} \\
	\downIterWlp{\stmt} &=\quad \stmtIf{b}{~\downTransWlp{\loopbody}\symSemi \stmt~}{~\stmtSkip~} \\
	\const{I} &=\quad \Assert{I}\symSemi \Assume{\embed{\exprFalse}} \\
    \downExtendWlp{S} &=\quad \upAssert{I}\symSemi \downIterWlp{S}
\end{align*}
\end{samepage}
\bigskip

\noindent
\textbf{Sketches} for $k=2$ and $k=3$:\par\nopagebreak
\noindent\begin{minipage}[t]{0.5\textwidth}
\vspace{0pt}
    \begin{appendixprogram}
		&\Assert{I} \\
		&\Havoc{\textit{variables}} \\
		&\Validate \\
		&\Assume{I} \\
		&\symIf~(\bexpr)~\blockStart \\
		&\quad \downTransWlp{\loopbody} \\
		&\quad \coAssert{I} \\
		&\quad \symIf~(\bexpr)~\blockStart \\
		&\quad\quad \downTransWlp{\loopbody} \\
		&\quad\quad \Assert{I} \\
		&\quad\quad \Assume{\embed{\exprFalse}} \\
		&\quad \blockEnd \\
		&\}
    \end{appendixprogram}
\end{minipage}%
\begin{minipage}[t]{0.5\textwidth}
\vspace{0pt}
    \begin{appendixprogram}
		&\Assert{I} \\
		&\Havoc{\textit{variables}} \\
		&\Validate \\
		&\Assume{I} \\
		&\symIf~(\bexpr)~\blockStart \\
		&\quad \downTransWlp{\loopbody} \\
		&\quad \coAssert{I} \\
		&\quad \symIf~(\bexpr)~\blockStart \\
		&\quad\quad \downTransWlp{\loopbody} \\
		&\quad\quad \coAssert{I} \\
		&\quad\quad \symIf~(\bexpr)~\blockStart \\
		&\quad\quad\quad \downTransWlp{\loopbody} \\
		&\quad\quad\quad \Assert{I} \\
		&\quad\quad\quad \Assume{\embed{\exprFalse}} \\
		&\quad\quad\blockEnd \\
		&\quad \blockEnd \\
		&\}
    \end{appendixprogram}
\end{minipage}

\subsection{Loop Rule: $k$-Induction for $\symWp$}
\label{sec:app-enc-kind-wp}

The $k$-induction encoding for $\symWp$ is dual to the $k$-induction encoding of $\symWlp$ (cf.~\cref{sec:app-encodings-kind-wlp}).
It encodes a while loop and over-approximates its $\symWp$ semantics.
The user needs to provide a potential superinvariant $I \in \HeyLo$ and a number $k \in \mathbb{N}_{\geq 1}$ of how many times to unfold the loop.
For the loop body, we assume another over-approximating encoding is given.
If it contains loops, $k$-induction can be encoded recursively, but other encodings can be used as well.
\medskip

\noindent
\textbf{Proof Rule:}
\emph{Latticed $k$-induction}~\cite{batzLatticedKInductionApplication2021} for $\symWp$. \\
Let $\CC = \WHILEDO{\bexpr}{\loopbody}$ be a pGCL loop and let $\expa \in \Expectations$.
The \emph{$k$-induction operator} for $\symWp$ is given by
\[
	\symUpPsi_I \colon \Expectations \to \Expectations,\quad \expb \mapsto \wpPhi_\expa(\expb) \sqcap I~,
\]
where the \emph{loop-characteristic functional} $\wpPhi_\expa$ with respect to post $\expa$ is defined as
\[
	\wpPhi_\expa(\expb) = \iverson{\bexpr} \cdot \wp{\loopbody}(\expb) + \iverson{\neg b} \cdot \expa~.
\]
Then, for $k \in \mathbb{N}_{\geq 1}$,
\[
	\wpPhi_\expa(\symUpPsi_I^{k-1}(I)) \expleq I \quad\text{implies}\quad \wp{\CC}(\expa) \expleq I~.
\]
\medskip

\noindent
\textbf{Encoding Input:} 
\begin{itemize}
	\item pGCL loop $\CC = \WHILEDO{\bexpr}{\loopbody}$.
	\item HeyVL encoding $\upTransWp{\loopbody}$ that satisfies $\wp{\loopbody} \expleq \vc{\upTransWp{\loopbody}}$.
	\item $k \in \mathbb{N}_{\geq 1}$.
	\item Potential $k$-inductive $\symWp$-superinvariant $I \in \HeyLo$.
\end{itemize}
\medskip

\noindent
\textbf{Encoding Output:}
\begin{itemize}
	\item HeyVL encoding $\stmt$ that satisfies $\wp{\CC} \expleq \vc{\stmt}$.
		\begin{itemize}
			\item If $\loopbody$ is encoded exactly, i.e. $\vc{\upTransWp{\loopbody}} = \wp{\loopbody}$ holds, then
				\[
					\vc{\stmt} = \ifThenElseDot{\wpPhi_\expa(\symUpPsi_I^{k-1}(I)) \expleq I}{I}{\infty}
				\]
		\end{itemize}
\end{itemize}
\medskip

\begin{samepage}
\noindent
Formally, the encoding $\stmt$ is given by:
\begin{align*}
	\stmt &=\quad \upSpec{I}{I}\symSemi \upIterWp{\upExtendWpK{k-1}{\const{I}}} \\
\shortintertext{where}
	\upSpec{\hlb}{\hla} &=\quad \coAssert{\hlb}\symSemi \coHavoc{\textit{variables}}\symSemi \coValidate\symSemi \coAssume{\hla} \\
	\upIterWp{\stmt} &=\quad \stmtIf{b}{~\upTransWp{\loopbody}\symSemi \stmt~}{~\stmtSkip~} \\
	\const{I} &=\quad \coAssert{I}\symSemi \coAssume{\embed{\exprTrue}} \\
    \upExtendWp{S} &=\quad \Assert{I}\symSemi \upIterWp{S}
\end{align*}
\end{samepage}
\bigskip

\noindent
\textbf{Sketches} for $k = 2$ and $k = 3$:\par\nopagebreak
\noindent\begin{minipage}[t]{0.5\textwidth}
\vspace{0pt}
\begin{appendixprogram}
	&\coAssert{I} \\
	&\coHavoc{\textit{variables}} \\
	&\coValidate \\
	&\coAssume{I} \\
	&\symIf~(\bexpr)~\blockStart \\
	&\quad \upTransWp{\loopbody} \\
	&\quad \Assert{I} \\
	&\quad \symIf~(\bexpr)~\blockStart \\
	&\quad\quad \upTransWp{\loopbody} \\
	&\quad\quad \coAssert{I} \\
	&\quad\quad \coAssume{\embed{\exprTrue}} \\
	&\quad \blockEnd \\
	&\} 
\end{appendixprogram}
\end{minipage}%
\begin{minipage}[t]{0.5\textwidth}
\vspace{0pt}
    \begin{appendixprogram}
		&\coAssert{I} \\
		&\coHavoc{\textit{variables}} \\
		&\coValidate \\
		&\coAssume{I} \\
		&\symIf~(\bexpr)~\blockStart \\
		&\quad \upTransWp{\loopbody} \\
		&\quad \Assert{I} \\
		&\quad \symIf~(\bexpr)~\blockStart \\
		&\quad\quad \upTransWp{\loopbody} \\
		&\quad\quad \Assert{I} \\
		&\quad\quad \symIf~(\bexpr)~\blockStart \\
		&\quad\quad\quad \upTransWp{\loopbody} \\
		&\quad\quad\quad \coAssert{I} \\
		&\quad\quad\quad \coAssume{\embed{\exprTrue}} \\
		&\quad\quad\blockEnd \\
		&\quad \blockEnd \\
		&\}
    \end{appendixprogram}
\end{minipage}

\subsection{Loop Rule: $\omega$-invariants for $\symWlp$}
\label{sec:app-enc-omega-wlp}

\textbf{Proof Rule:} \emph{$\omega$-invariants for $\symWlp$} (adapted from~\cite{kaminskiWeakestPreconditionReasoning2018})~\footnote{Different versions of this proof rule exist.
An overview is found in~\cite[page~108]{kaminskiAdvancedWeakestPrecondition2019}.}.
Let $\CC = \WHILEDO{\bexpr}{\loopbody}$ be a pGCL loop and let $\expa \in \OneBoundedExpectations$.
Let $(I_n)_{n \in \mathbb{N}} \subset \OneBoundedExpectations$ with $\wlpPhi_\expa(1) \expleq I_0$.
If $(I_n)_{n \in \mathbb{N}}$ is a \emph{$\symWlp$-$\omega$-superinvariant}, then $\inf_{n \in \mathbb{N}} I_n$ upper-bounds $\wlp{\CC}(\expa)$, i.e.
\[
	(\forall n \in \mathbb{N}.~ \wlpPhi_\expa(I_n) \expleq I_{n + 1}) \quad\text{implies}\quad	\wlp{\CC}(\expa) \expleq \inf_{n \in \mathbb{N}} I_n~,
\]
where the \emph{loop-characteristic functional} $\wlpPhi_\expa$ with respect to post $\expa$ is defined as
\[
	\wlpPhi_\expa(\expb) = \iverson{\bexpr} \cdot \wlp{\loopbody}(\expb) + \iverson{\neg b} \cdot \expa~.
\]

\medskip

\noindent
\textbf{Encoding Input:} 
\begin{itemize}
	\item pGCL loop $\CC = \WHILEDO{\bexpr}{\loopbody}$.
	\item HeyVL encoding $\upTransWlp{\loopbody}$ that satisfies $\wlp{\loopbody} \expleq \vc{\upTransWlp{\loopbody}}$.
	\item Potential $\symWlp$-$\omega$-superinvariant $I_n \in \HeyLo$ that represents $(I_n)_{n \in \mathbb{N}} \subset \OneBoundedExpectations$ by a free variable $n$.
	\item Post $\hla \in \HeyLo$ with $\hla \heyloleq 1$.
\end{itemize}
\medskip

\noindent
\textbf{Encoding Output:}
\begin{itemize}
	\item Two HeyVL procedures whose joint verification implies $\wlp{\CC}(\hla) \expleq \inf_{n \in \mathbb{N}} I_n$.
\end{itemize}
\medskip

\noindent
We generate two procedures to check the proof rule conditions.
\smallskip

\noindent
The first procedure checks that $\wlpPhi_{\eval{\hla}}(1) \expleq I_0$ holds:
\begin{appendixprogram}
	&\symUp\proc{condition\_1}{\typeof{x_1^0}{\typevar_1},\ldots,\typeof{x_m^0}{\typevar_m}}{\typeof{x_1}{\typevar_1},\ldots,\typeof{x_m}{\typevar_m}} \\
	&\quad\Requires{I_n\substBy{n}{0}\substBy{x_1}{x_1^0}\ldots\substBy{x_m}{x_m^0}} \\
	&\quad\Ensures{\hla} \\
	&\blockStart \\
	&\quad \stmtDeclInit{x_1}{\typevar_1}{x_1^0}\symSemi \ldots\symSemi \stmtDeclInit{x_m}{\typevar_m}{x_m^0} \\
	&\quad \symIf~(\bexpr)~\blockStart \\
	&\quad\quad \upTransWlp{\loopbody} \\
	&\quad\quad \coAssert{1} \\
	&\quad\quad \coAssume{\embed{\exprTrue}} \\
	&\quad \blockEnd \\
	&\blockEnd
\end{appendixprogram}

The second procedure checks that $\wlpPhi_{\eval{\hla}}(I_n) \expleq I_{n+1}$ holds for all $n \in \mathbb{N}$:
\begin{appendixprogram}
	&\symUp\proc{condition\_2}{\typeof{n}{\Nats},\typeof{x_1^0}{\typevar_1},\ldots,\typeof{x_m^0}{\typevar_m}}{\typeof{x_1}{\typevar_1},\ldots,\typeof{x_m}{\typevar_m}} \\
	&\quad\Requires{I_n\substBy{n}{n+1}\substBy{x_1}{x_1^0}\ldots\substBy{x_m}{x_m^0}} \\
	&\quad\Ensures{\hla} \\
	&\blockStart \\
	&\quad \stmtDeclInit{x_1}{\typevar_1}{x_1^0}\symSemi \ldots\symSemi \stmtDeclInit{x_m}{\typevar_m}{x_m^0} \\
	&\quad \symIf~(\bexpr)~\blockStart \\
	&\quad\quad \upTransWlp{\loopbody} \\
	&\quad\quad \coAssert{I_n} \\
	&\quad\quad \coAssume{\embed{\exprTrue}} \\
	&\quad \blockEnd \\
	&\blockEnd
\end{appendixprogram}

\subsection{Loop Rule: $\omega$-invariants for $\symWp$}
\label{sec:app-enc-omega-wp}

\textbf{Proof Rule:} \emph{$\omega$-invariants for $\symWp$}~ \cite{kaminskiAdvancedWeakestPrecondition2019}.
Let $\CC = \WHILEDO{\bexpr}{\loopbody}$ be a pGCL loop and let $\expa \in \Expectations$.
Let $(I_n)_{n \in \mathbb{N}} \subset \Expectations$ with $I_0 \expleq \wpPhi_\expa(0)$.
If $(I_n)_{n \in \mathbb{N}}$ is a \emph{$\symWp$-$\omega$-subinvariant}, then $\sup_{n \in \mathbb{N}} I_n$ lower-bounds $\wp{\CC}(\expa)$, i.e.
\[
	(\forall n \in \mathbb{N}.~ I_{n + 1} \expleq \wpPhi_\expa(I_n)) \quad\text{implies}\quad \sup_{n \in \mathbb{N}} I_n \expleq \wp{\CC}(\expa)~,
\]
where the \emph{loop-characteristic functional} $\wpPhi_\expa$ with respect to post $\expa$ is defined as
\[
	\wpPhi_\expa(\expb) = \iverson{\bexpr} \cdot \wp{\loopbody}(\expb) + \iverson{\neg b} \cdot \expa~.
\]

\medskip

\noindent
\textbf{Encoding Input:} 
\begin{itemize}
	\item pGCL loop $\CC = \WHILEDO{\bexpr}{\loopbody}$.
	\item HeyVL encoding $\downTransWp{\loopbody}$ that satisfies $\vc{\downTransWp{\loopbody}}\expleq \wp{\loopbody}$.
	\item Potential $\symWp$-$\omega$-subinvariant $I_n \in \HeyLo$ that represents $(I_n)_{n \in \mathbb{N}} \subset \Expectations$ by a free variable $n$.
	\item Post $\hla \in \HeyLo$.
\end{itemize}
\medskip

\noindent
\textbf{Encoding Output:}
\begin{itemize}
	\item Two HeyVL procedures whose joint verification implies $\sup_{n \in \mathbb{N}} I_n \expleq \wp{\CC}(\hla)$.
\end{itemize}
\medskip

\noindent
We generate two procedures to check the proof rule conditions.
\smallskip

\noindent
The first procedure checks that $I_0 \expleq \wpPhi_{\eval{\hla}}(0)$ holds:
\begin{appendixprogram}
	&\proc{condition\_1}{\typeof{x_1^0}{\typevar_1},\ldots,\typeof{x_m^0}{\typevar_m}}{\typeof{x_1}{\typevar_1},\ldots,\typeof{x_m}{\typevar_m}} \\
	&\quad\Requires{I_n\substBy{n}{0}\substBy{x_1}{x_1^0}\ldots\substBy{x_m}{x_m^0}} \\
	&\quad\Ensures{\hla} \\
	&\blockStart \\
	&\quad \stmtDeclInit{x_1}{\typevar_1}{x_1^0}\symSemi \ldots\symSemi \stmtDeclInit{x_m}{\typevar_m}{x_m^0} \\
	&\quad \symIf~(\bexpr)~\blockStart \\
	&\quad\quad \downTransWp{\loopbody} \\
	&\quad\quad \Assert{0} \\
	&\quad\quad \Assume{\embed{\exprFalse}} \\
	&\quad \blockEnd \\
	&\blockEnd
\end{appendixprogram}

The second procedure checks that $I_{n+1} \expleq \wpPhi_{\eval{\hla}}(I_n)$ holds for all $n \in \mathbb{N}$:
\begin{appendixprogram}
	&\proc{condition\_2}{\typeof{n}{\Nats},\typeof{x_1^0}{\typevar_1},\ldots,\typeof{x_m^0}{\typevar_m}}{\typeof{x_1}{\typevar_1},\ldots,\typeof{x_m}{\typevar_m}} \\
	&\quad\Requires{I_n\substBy{n}{n+1}\substBy{x_1}{x_1^0}\ldots\substBy{x_m}{x_m^0}} \\
	&\quad\Ensures{\hla} \\
	&\blockStart \\
	&\quad \stmtDeclInit{x_1}{\typevar_1}{x_1^0}\symSemi \ldots\symSemi \stmtDeclInit{x_m}{\typevar_m}{x_m^0} \\
	&\quad \symIf~(\bexpr)~\blockStart \\
	&\quad\quad \downTransWp{\loopbody} \\
	&\quad\quad \Assert{I_n} \\
	&\quad\quad \Assume{\embed{\exprFalse}} \\
	&\quad \blockEnd \\
	&\blockEnd
\end{appendixprogram}

\subsection{Encoding of the Optional Stopping Theorem for $\symWp$}
\label{sec:app-enc-ost-wp}

\textbf{Proof Rule:} \emph{Optional Stopping Theorem for $\symWp$ Reasoning}~\cite{harkAimingLowHarder2019}.
Let $\CC = \WHILEDO{\bexpr}{\loopbody}$ be a pGCL loop and let $\expa \in \Expectations$.
If all of the following conditions hold:
\begin{itemize}
	\item $I$ is a $\symWp$-subinvariant: $I \expleq \wpPhi_\expa(I)$,
	\item $\CC$ is positively almost-surely terminating (PAST),
	\item $I$ harmonizes with $f$: $\neg b \Rightarrow (I = f)$,
	\item $\wpPhi_\expa(I)$ is finite: $\wpPhi_\expa(I) < \infty$,
	\item $I$ is \emph{conditionally difference bounded} for some $c \in \mathbb{R}_{\geq 0}$: $$\wp{\loopbody}(|I - I(\State)|)(\State) \leq c \text{ for all } \State \in \States~.$$
\end{itemize}
Then,
\[
	I \expleq \wp{\CC}(\expa)~.	
\]
\medskip

\noindent
\textbf{Encoding Input:} 
\begin{itemize}
	\item pGCL loop $\CC = \WHILEDO{\bexpr}{\loopbody}$.
	\item HeyVL encoding $\downTransWp{\loopbody}$ that satisfies $\vc{\downTransWp{\loopbody}} \expleq \wp{\loopbody}$.
	\item HeyVL encoding $\upTransWp{\loopbody}$ that satisfies $\wp{\loopbody} \expleq \vc{\upTransWp{\loopbody}}$.
	\item Potential $\symWp$-subinvariant $I \in \HeyLo$.
	\item Constant $c \in \mathbb{R}_{\geq 0}$. 
	\item Post $\hla \in \HeyLo$.
\end{itemize}
\medskip

\noindent
\textbf{Side Conditions:}
\begin{itemize}
	\item $\CC$ is positively almost-surely terminating (PAST).\footnote{For our ``ost'' example, we show $\ert{\CC}(0) < \infty$ using Park induction (cf.~\cref{sec:app-encodings-ert}) to show that $\CC$ is PAST.}
\end{itemize}
\medskip

\noindent
\textbf{Encoding Output:}
\begin{itemize}
	\item HeyVL encoding $\stmt$ that verifies only if $I \expleq \wp{\CC}(\hla)$.
\end{itemize}
\medskip

\noindent
Let $\typeof{x_1}{\typevar_1},\ldots,\typeof{x_n}{\typevar_n}$ be the variables that are free in $\CC$ with their types. \\
Let $I_0 = I\substBy{x_1}{x_1^0}\ldots\substBy{x_n}{x_n^0}$.
\medskip

\noindent
Multiple procedures are generated to check the various conditions.
\smallskip

\noindent
The first procedure checks that $I$ is a $\symWp$-subinvariant with respect to post $\hla$:
\begin{appendixprogram}
	&\proc{subinvariant}{\typeof{x_1^0}{\typevar_1},\ldots,\typeof{x_n^0}{\typevar_n}}{\typeof{x_1}{\typevar_1},\ldots,\typeof{x_n}{\typevar_n}} \\
	&\quad\Requires{I_0} \\
	&\quad\Ensures{\hla} \\
	&\blockStart \\
	&\quad \stmtDeclInit{x_1}{\typevar_1}{x_1^0}\symSemi \ldots\symSemi \stmtDeclInit{x_n}{\typevar_n}{x_n^0} \\
	&\quad \symIf~(\bexpr)~\blockStart \\
	&\quad\quad \downTransWp{\loopbody} \\
	&\quad\quad \Assert{I} \\
	&\quad\quad \Assume{\embed{\exprFalse}} \\
	&\quad \blockEnd \\
	&\blockEnd
\end{appendixprogram}

\noindent
Next, we check that $I$ harmonizes with $\hla$, i.e. that $\neg b \impl (I = \hla)$ holds.
Formally, we do this using a procedure and a coprocedure.\footnote{Our implementation \emph{Caesar} supports \HeyLo formulae of the more direct form $\embed{\neg b \impl (I = \hla)}$ as well.}
\begin{appendixprogram}
	&\proc{harmonizes\_lower}{\typeof{x_1}{\typevar_1},\ldots,\typeof{x_n}{\typevar_n}}{} \\
	&\quad\Requires{\embed{\neg b} \sqcap I} \\
    &\quad\Ensures{\hla}~\blockStart \blockEnd \\[\appendixprogramsep]
	&\symUp\proc{harmonizes\_upper}{\typeof{x_1}{\typevar_1},\ldots,\typeof{x_n}{\typevar_n}}{} \\
	&\quad\Requires{\embed{b} \sqcup I} \\
	&\quad\Ensures{\hla}~\blockStart \blockEnd
\end{appendixprogram}

\noindent
The next procedure checks that $\wpPhi_{\eval{\hla}}(I)$ is finite:
\begin{appendixprogram}
	&\symUp\proc{phi\_finite}{\typeof{x_1^0}{\typevar_1},\ldots,\typeof{x_n^0}{\typevar_n}}{\typeof{x_1}{\typevar_1},\ldots,\typeof{x_n}{\typevar_n}} \\
	&\quad\Requires{0} \\
	&\quad\Ensures{\hla} \\
	&\blockStart \\
	&\quad \Validate \\
	&\quad \Assume{\infty} \\
	&\quad \stmtDeclInit{x_1}{\typevar_1}{x_1^0}\symSemi \ldots\symSemi \stmtDeclInit{x_n}{\typevar_n}{x_n^0} \\
	&\quad \symIf~(\bexpr)~\blockStart \\
	&\quad\quad \upTransWp{\loopbody} \\
	&\quad\quad \Assert{I} \\
	&\quad\quad \Assume{\embed{\exprFalse}} \\
	&\quad \blockEnd \\
	&\blockEnd
\end{appendixprogram}

\noindent
The last procedure checks the conditional difference boundedness property:
\begin{appendixprogram}
	&\symUp\proc{cdb}{\typeof{x_1^0}{\typevar_1},\ldots,\typeof{x_n^0}{\typevar_n}}{\typeof{x_1}{\typevar_1},\ldots,\typeof{x_n}{\typevar_n}} \\
	&\quad\Requires{c} \\
	&\quad\Ensures{\exprIte{I_0 \leq I}{I - I_0}{I_0 - I}} \\
	&\blockStart \\
	&\quad \stmtDeclInit{x_1}{\typevar_1}{x_1^0}\symSemi \ldots\symSemi \stmtDeclInit{x_n}{\typevar_n}{x_n^0} \\
	&\quad \upTransWp{\loopbody} \\
	&\blockEnd
\end{appendixprogram}

\subsection{Proof Rules for $\symErt$}
\label{sec:app-encodings-ert}

The $\symErt$ calculus~\cite{kaminskiWeakestPreconditionReasoning2016,kaminskiWeakestPreconditionReasoning2018} is similar to the $\symWp$ calculus.
For loop-free pGCL programs, we obtain encodings similar to \cref{fig:encodings:wp}.
We add $\stmtTick{1}$ statements to track run-times and resolve nondeterministic choice by taking the maximum of the branch runtimes.

\begin{center}
	\begin{tabular}{ll}
		$\pcc$ & $\downTransErt{\pcc}$ \\
		\hline \hline
		$\pSkip$ & $\stmtTick{1}$ \\
		$\pDiverge$ & $\Assume{0}$ \\
		$\pAssign{x}{t}$ & $\stmtRasgn{x}{t}\symSemi \stmtTick{1}\symSemi$ \\
		$\pcc_1;\pcc_2$ & $\downTransErt{\pcc_1};\downTransErt{\pcc_2}$ \\
		$\stmtIfStart{b}~\pcc_1~\} $ & $\stmtTick{1}\symSemi \stmtDemonicStart~\Assume{\embed{b}};\downTransErt{\pcc_1}~\}$ \\
		$\quad\stmtElseStart~\pcc_2 \}$ & $\stmtElseStart~\Assume{\embed{\neg b}};\downTransErt{\pcc_2} \}$ \\
		$\pChoice{\pcc_1}{p}{\pcc_2}$ & $\stmtDeclInit{\mathit{tmp}}{\bool}{\exprFlip{p}}\symSemi$ \\
		& $\downTransErt{\pIte{\mathit{tmp}}{\pcc_1}{\pcc_2}}$ \\
		$\pNd{\pcc_1}{\pcc_2}$ & $\stmtAngelic{\downTransErt{\pcc_1}}{\downTransErt{\pcc_2}}$
	\end{tabular}
\end{center}

Latticed $k$-induction (cf.~\cref{sec:app-enc-kind-wp}) and $\omega$-invariants (cf.~\cref{sec:app-enc-omega-wp}) can be encoded similarly to $\symWp$ and are implemented in our frontend.
The $k$-induction proof rule for $\symErt$ is a straightforward consequence of the latticed $k$-induction principle~\cite{batzLatticedKInductionApplication2021}.
$\omega$-invariants for $\symErt$ have been described in~\cite{kaminskiWeakestPreconditionReasoning2016}.
 
\subsection{Encoding of ``A New Proof Rule for Almost-Sure Termination''}
\label{sec:app-enc-new-proof-rule-ast}

\textbf{Proof Rule:} \emph{``A New Proof Rule for Almost-Sure Termination''}~\cite{mciverNewProofRule2018}.
Let $\CC = \WHILEDO{\bexpr}{\loopbody}$ be a pGCL loop without nondeterministic choice.
Let $I \colon \States \to \Bools$, let $V \colon \States \to \mathbb{R}_{\geq 0}$, let $p \colon \mathbb{R}_{\geq 0} \to (0,1]$ and $d \colon \mathbb{R}_{\geq 0} \to \mathbb{R}_{> 0}$ where $d$ and $p$ are antitone on positive arguments.
If the following four conditions hold, then $\iverson{I} \expleq \wp{\CC}(1)$ holds:
\begin{itemize}
	\item $\iverson{I}$ is a $\symWp$-subinvariant: $\iverson{I} \expleq \wpPhi_{\iverson{I}}(\iverson{I})$,
	\item $\bexpr \land I \implies V > 0$,
	\item $\wpPhi_V(V) \expleq V$,
	\item $\iverson{I} \cdot \iverson{b} \cdot (p \circ V) \expleq \lambda \State.~ \wp{\loopbody}(\iverson{V \leq V(\State) - d(V(\State))})(\State)$.
\end{itemize}

\medskip

\noindent
\textbf{Encoding Input:} 
\begin{itemize}
	\item pGCL loop $\CC = \WHILEDO{\bexpr}{\loopbody}$ without nondeterministic choice.
	\item HeyVL encoding $\downTransWp{\loopbody}$ that satisfies $\vc{\downTransWp{\loopbody}} \expleq \wp{\loopbody}$.
	\item HeyVL encoding $\upTransWp{\loopbody}$ that satisfies $\wp{\loopbody} \expleq \vc{\upTransWp{\loopbody}}$.
	\item Expressions $\typeof{I}{\Bools}$, $\typeof{V}{\ureal}$, $\typeof{p}{\ureal}$, $\typeof{d}{\ureal}$ with a free variable $x$ each for the parameter of the function that they represent.
	As side conditions, we require $0 < p \leq 1$ and $d > 0$ for every value of $x \in \ureal$.
\end{itemize}
\medskip

\noindent
\textbf{Encoding Output:}
\begin{itemize}
	\item HeyVL encoding $\stmt$ that verifies only if $\iverson{I} \expleq \wp{\CC}(1)$.
\end{itemize}
\medskip

\noindent
Let $\typeof{x_1}{\typevar_1},\ldots,\typeof{x_n}{\typevar_n}$ be the variables that are free in $\CC$ with their types. \\
Let $I_0 = I\substBy{x_1}{x_1^0}\ldots\substBy{x_n}{x_n^0}$ and $V_0 = V\substBy{x_1}{x_1^0}\ldots\substBy{x_n}{x_n^0}$ and $b_0 = b\substBy{x_1}{x_1^0}\ldots\substBy{x_n}{x_n^0}$.
\medskip

\noindent
Multiple procedures are generated to check the various conditions.
\smallskip

\noindent
The first two procedures check that $p$ and $d$ are antitone:
\begin{appendixprogram}
	&\proc{p\_antitone}{\typeof{a}{\ureal},\typeof{b}{\ureal}}{} \\
	&\quad\Requires{\embed{a \leq b}} \\
    &\quad\Ensures{\embed{p\substBy{x}{a} \geq p\substBy{x}{b}}}~\blockStart \blockEnd \\[\appendixprogramsep]
	&\proc{d\_antitone}{\typeof{a}{\ureal},\typeof{b}{\ureal}}{} \\
	&\quad\Requires{\embed{a \leq b}} \\
	&\quad\Ensures{\embed{d\substBy{x}{a} \geq d\substBy{x}{b}}}~\blockStart \blockEnd
\end{appendixprogram}

\noindent
The following procedure checks that $\iverson{I}$ is a $\symWp$-subinvariant with respect to post $\iverson{I}$:
\begin{appendixprogram}
	&\proc{I\_wp\_subinvariant}{\typeof{x_1^0}{\typevar_1},\ldots,\typeof{x_n^0}{\typevar_n}}{\typeof{x_1}{\typevar_1},\ldots,\typeof{x_n}{\typevar_n}} \\
	&\quad\Requires{\iverson{I_0}} \\
	&\quad\Ensures{\iverson{I}} \\
	&\blockStart \\
	&\quad \stmtDeclInit{x_1}{\typevar_1}{x_1^0}\symSemi \ldots\symSemi \stmtDeclInit{x_n}{\typevar_n}{x_n^0} \\
	&\quad \symIf~(\bexpr)~\blockStart \\
	&\quad\quad \downTransWp{\loopbody} \\
	&\quad \blockEnd \\
	&\blockEnd
\end{appendixprogram}

\noindent
The next condition:
\begin{appendixprogram}
	&\proc{termination\_condition}{\typeof{x_1}{\typevar_1},\ldots,\typeof{x_n}{\typevar_n}}{}~\blockStart \\
	&\quad \Assert{\embed{(\bexpr \land I) \rightarrow (V > 0)}} \\
	&\blockEnd
\end{appendixprogram}

\noindent
Then, we check that $\wpPhi_V(V) \expleq V$ holds:
\begin{appendixprogram}
	&\symUp\proc{v\_wp\_superinvariant}{\typeof{x_1^0}{\typevar_1},\ldots,\typeof{x_n^0}{\typevar_n}}{\typeof{x_1}{\typevar_1},\ldots,\typeof{x_n}{\typevar_n}} \\
	&\quad\Requires{V_0} \\
	&\quad\Ensures{V} \\
	&\blockStart \\
	&\quad \stmtDeclInit{x_1}{\typevar_1}{x_1^0}\symSemi \ldots\symSemi \stmtDeclInit{x_n}{\typevar_n}{x_n^0} \\
	&\quad \symIf~(\bexpr)~\blockStart \\
	&\quad\quad \upTransWp{\loopbody} \\
	&\quad \blockEnd \\
	&\blockEnd
\end{appendixprogram}

Finally, the progress condition $\iverson{I} \cdot \iverson{b} \cdot (p \circ V) \expleq \lambda \State.~ \wp{\loopbody}(\iverson{V \leq V(\State) - d(V(\State))})(\State)$:
\begin{appendixprogram}
	&\proc{progress}{\typeof{x_1^0}{\typevar_1},\ldots,\typeof{x_n^0}{\typevar_n}}{\typeof{x_1}{\typevar_1},\ldots,\typeof{x_n}{\typevar_n}} \\
	&\quad\Requires{\iverson{I_0} \cdot \iverson{b_0} \cdot (p \circ V_0)} \\
	&\quad\Ensures{\iverson{V \leq V_0 - d(V_0)}} \\
	&\blockStart \\
	&\quad \stmtDeclInit{x_1}{\typevar_1}{x_1^0}\symSemi \ldots\symSemi \stmtDeclInit{x_n}{\typevar_n}{x_n^0} \\
	&\quad \downTransWp{\loopbody} \\
	&\blockEnd
\end{appendixprogram}

\subsection{PAST Rule}
\label{sec:app-enc-past}

\textbf{Proof Rule}: \emph{PAST from Ranking Superinvariants}~\cite{chakarovProbabilisticProgramAnalysis2013}.
Let $\CC = \WHILEDO{\bexpr}{\loopbody}$ be a pGCL loop with loop-free body $\loopbody$ without nondeterministic choice or $\pDiverge$.
Let $I \colon \States \to \mathbb{R}_{\geq 0}$ and let $\epsilon$ and $K$ be constants such that $0 < \epsilon < K$.
If the following conditions hold, then $\CC$ terminates universally positively almost-surely:
\begin{itemize}
	\item $\iverson{\neg \bexpr} \cdot I \expleq K$,
	\item $\iverson{\bexpr} \cdot K \expleq \iverson{\bexpr} \cdot I + \iverson{\neg \bexpr}$,
	\item $\wpPhi_0(I) \expleq \iverson{\bexpr} \cdot (I - \epsilon)$,
\end{itemize}
where the \emph{loop-characteristic functional} $\wpPhi_\expa$ with respect to post $\expa$ is defined as
\[
	\wpPhi_\expa(\expb) = \iverson{\bexpr} \cdot \wp{\loopbody}(\expb) + \iverson{\neg b} \cdot \expa~.
\]

\noindent
\textbf{Encoding Input:} 
\begin{itemize}
	\item pGCL loop $\CC = \WHILEDO{\bexpr}{\loopbody}$ with loop-free body $\loopbody$ without nondeterministic choice or $\pDiverge$.
	\item HeyVL encoding $\upTransWp{\loopbody}$ that satisfies $\wp{\loopbody} \expleq \vc{\upTransWp{\loopbody}}$.
	\item Potential invariant $I \in \HeyLo$ with $I < \infty$.
	\item Constants $\epsilon$ and $K$ such that $0 < \epsilon < K$.
\end{itemize}
\medskip

\noindent
\textbf{Encoding Output:}
\begin{itemize}
	\item HeyVL encoding $\stmt$ that verifies only if $\CC$ is PAST.
\end{itemize}
\medskip

\noindent
The first two conditions can be encoded easily via assertions:
\begin{appendixprogram}
	&\proc{condition\_1}{\typeof{x_1}{\typevar_1},\ldots,\typeof{x_n}{\typevar_n}}{}~\blockStart \\
	&\quad \Assert{\embed{\iverson{\neg \bexpr} \cdot I \leq K}} \\
    &\blockEnd \\[\appendixprogramsep]
	&\proc{condition\_2}{\typeof{x_1}{\typevar_1},\ldots,\typeof{x_n}{\typevar_n}}{}~\blockStart \\
	&\quad \Assert{\embed{\iverson{\bexpr} \cdot K \leq \iverson{\bexpr} \cdot I + \iverson{\neg \bexpr}}} \\
	&\blockEnd
\end{appendixprogram}

\noindent
The last condition, $\wpPhi_0(I) \expleq \iverson{\bexpr} \cdot (I - \epsilon)$, is encoded as another coprocedure:
\begin{appendixprogram}
	&\symUp\proc{condition\_3}{\typeof{x_1^0}{\typevar_1},\ldots,\typeof{x_n^0}{\typevar_n}}{} \\
	&\quad\Requires{(\iverson{b} \cdot (I - \epsilon))\substBy{x_1}{x_1^0}\ldots\substBy{x_n}{x_n^0}} \\
	&\quad\Ensures{0} \\
	&\blockStart \\
	&\quad \stmtDeclInit{x_1}{\typevar_1}{x_1^0}\symSemi \ldots\symSemi \stmtDeclInit{x_n}{\typevar_n}{x_n^0} \\
	&\quad \symIf~(\bexpr)~\blockStart \\
	&\quad\quad \upTransWp{\loopbody} \\
	&\quad\quad \coAssert{I} \\
	&\quad\quad \coAssume{\embed{\exprTrue}} \\
	&\quad \blockEnd \\
	&\blockEnd
\end{appendixprogram}

%% file: literature.bib
@misc{wiki:cc,
    author = "{Wikipedia}",
    title = "Coupon Collector's Problem",
    year = "2023",
    howpublished = "\url{https://en.wikipedia.org/wiki/Coupon_collector\%27s_problem}",
    note = "[Online; accessed 4-September-2023]"
  }

@misc{wiki:rw,
    author = "{Wikipedia}",
    title = "Random Walk",
    year = "2023",
    howpublished = "\url{https://en.wikipedia.org/wiki/Random_walk\#One-dimensional_random_walk}",
    note = "[Online; accessed 4-September-2023]"
  }

@inproceedings{baazInfinitevaluedGodelLogics1996,
  title = {Infinite-Valued {{G\"odel}} Logics with 0-1-{{Projections}} and Relativizations},
  booktitle = {Proc. {{G\"odel}}'96, Logic Foundations of Mathematics, Computer Science and Physics \textendash{} {{Kurt G\"odel}}'s Legacy},
  author = {Baaz, M.},
  editor = {H{\'a}jek, P.},
  year = {1996},
  series = {Lecture Notes in Logic 6},
  publisher = {{Springer}},
  address = {{Brno, Czech Republic}},
}

@inproceedings{bartheAssertionBasedProgramLogic2018,
  title = {An {{Assertion-Based Program Logic}} for {{Probabilistic Programs}}},
  booktitle = {Programming {{Languages}} and {{Systems}}},
  author = {Barthe, Gilles and Espitau, Thomas and Gaboardi, Marco and Gr{\'e}goire, Benjamin and Hsu, Justin and Strub, Pierre-Yves},
  editor = {Ahmed, Amal},
  year = {2018},
  series = {Lecture {{Notes}} in {{Computer Science}}},
  publisher = {{Springer International Publishing}},
  address = {{Cham}},
  isbn = {978-3-319-89884-1},
  doi = {10.1007/978-3-319-89884-1\_5}
}

@book{bartheFoundationsProbabilisticProgramming2020,
  title = {Foundations of {{Probabilistic Programming}}},
  editor = {Barthe, Gilles and Katoen, Joost-Pieter and Silva, Alexandra},
  year = {2020},
  publisher = {{Cambridge University Press}},
  address = {{Cambridge}},
  isbn = {978-1-108-48851-8},
  doi = {10.1017/9781108770750}
}

@article{batzCalculusAmortizedExpected2023,
	author       = {Kevin Batz and
	Benjamin Lucien Kaminski and
	Joost{-}Pieter Katoen and
	Christoph Matheja and
	Lena Verscht},
	title        = {A Calculus for Amortized Expected Runtimes},
	journal      = {Proc. {ACM} Program. Lang.},
	volume       = {7},
	number       = {{POPL}},
	pages        = {1957--1986},
	year         = {2023},
	doi = {https://doi.org/10.1145/3571260}
}

@article{batzFoundationsEntailmentChecking2022,
  author       = {Kevin Batz and
                  Ira Fesefeldt and
                  Marvin Jansen and
                  Joost{-}Pieter Katoen and
                  Florian Ke{\ss}ler and
                  Christoph Matheja and
                  Thomas Noll},
  title        = {Foundations for Entailment Checking in Quantitative Separation Logic},
  booktitle    = {{ESOP}},
  series       = {Lecture Notes in Computer Science},
  volume       = {13240},
  pages        = {57--84},
  publisher    = {Springer},
  year         = {2022},
  doi = {10.1007/978-3-030-99336-8\_3},
}

@article{batzHowLongBayesian2018,
  author       = {Kevin Batz and
                  Benjamin Lucien Kaminski and
                  Joost{-}Pieter Katoen and
                  Christoph Matheja},
  title        = {How long, {O} Bayesian network, will {I} sample thee? - {A} program
                  analysis perspective on expected sampling times},
  booktitle    = {{ESOP}},
  series       = {Lecture Notes in Computer Science},
  volume       = {10801},
  pages        = {186--213},
  publisher    = {Springer},
  year         = {2018},
  doi = {https://doi.org/10.1007/978-3-319-89884-1\_7}
}

@inproceedings{batzLatticedKInductionApplication2021,
	author       = {Kevin Batz and
	Mingshuai Chen and
	Benjamin Lucien Kaminski and
	Joost{-}Pieter Katoen and
	Christoph Matheja and
	Philipp Schr{\"{o}}er},
	title        = {Latticed k-Induction with an Application to Probabilistic Programs},
	booktitle    = {{CAV} {(2)}},
	series       = {Lecture Notes in Computer Science},
	volume       = {12760},
	pages        = {524--549},
	publisher    = {Springer},
	year         = {2021},
	doi = {https://doi.org/10.1007/978-3-030-81688-9\_25}
}

@article{batzPrIC3PropertyDirected2020,
  author       = {Kevin Batz and
                  Sebastian Junges and
                  Benjamin Lucien Kaminski and
                  Joost{-}Pieter Katoen and
                  Christoph Matheja and
                  Philipp Schr{\"{o}}er},
  title        = {PrIC3: Property Directed Reachability for MDPs},
  booktitle    = {{CAV} {(2)}},
  series       = {Lecture Notes in Computer Science},
  volume       = {12225},
  pages        = {512--538},
  publisher    = {Springer},
  year         = {2020},
  doi          = {10.1007/978-3-030-53291-8\_27},
}

@inproceedings{batzProbabilisticProgramVerification2022,
	author       = {Kevin Batz and
	Mingshuai Chen and
	Sebastian Junges and
	Benjamin Lucien Kaminski and
	Joost{-}Pieter Katoen and
	Christoph Matheja},
	title        = {Probabilistic Program Verification via Inductive Synthesis of Inductive
	Invariants},
	booktitle    = {{TACAS} {(2)}},
	series       = {Lecture Notes in Computer Science},
	volume       = {13994},
	pages        = {410--429},
	publisher    = {Springer},
	year         = {2023},
	doi = {https://doi.org/10.1007/978-3-031-30820-8\_25}
}

@article{batzQuantitativeSeparationLogic2019,
  title = {Quantitative Separation Logic: A Logic for Reasoning about Probabilistic Pointer Programs},
  shorttitle = {Quantitative Separation Logic},
  author = {Batz, Kevin and Kaminski, Benjamin Lucien and Katoen, Joost-Pieter and Matheja, Christoph and Noll, Thomas},
  year = {2019},
  month = jan,
  journal = {Proceedings of the ACM on Programming Languages},
  volume = {3},
  number = {POPL},
  doi = {https://doi.org/10.1145/3290347}
}

@article{batzRelativelyCompleteVerification2021,
	author       = {Kevin Batz and
	Benjamin Lucien Kaminski and
	Joost{-}Pieter Katoen and
	Christoph Matheja},
	title        = {Relatively complete verification of probabilistic programs: an expressive
	language for expectation-based reasoning},
	journal      = {Proc. {ACM} Program. Lang.},
	volume       = {5},
	number       = {{POPL}},
	pages        = {1--30},
	year         = {2021},
	doi = {https://doi.org/10.1145/3434320}
}

@article{batzWeightedProgrammingProgramming2022,
  title = {Weighted Programming: A Programming Paradigm for Specifying Mathematical Models},
  shorttitle = {Weighted Programming},
  author = {Batz, Kevin and Gallus, Adrian and Kaminski, Benjamin Lucien and Katoen, Joost-Pieter and Winkler, Tobias},
  year = {2022},
  month = apr,
  journal = {Proceedings of the ACM on Programming Languages},
  volume = {6},
  number = {OOPSLA1},
  doi = {https://doi.org/10.1145/3527310}
}

@inproceedings{cousotAutomaticInferenceNecessary2013,
  title = {Automatic {{Inference}} of {{Necessary Preconditions}}},
  booktitle = {Verification, {{Model Checking}}, and {{Abstract Interpretation}}},
  author = {Cousot, Patrick and Cousot, Radhia and F{\"a}hndrich, Manuel and Logozzo, Francesco},
  editor = {Giacobazzi, Roberto and Berdine, Josh and Mastroeni, Isabella},
  year = {2013},
  series = {Lecture {{Notes}} in {{Computer Science}}},
  publisher = {{Springer}},
  address = {{Berlin, Heidelberg}},
  isbn = {978-3-642-35873-9},
  doi = {https://doi.org/10.1007/978-3-642-35873-9_10}
}

@incollection{cousotPreconditionInferenceIntermittent2011,
  title = {Precondition {{Inference}} from {{Intermittent Assertions}} and {{Application}} to {{Contracts}} on {{Collections}}},
  booktitle = {Verification, {{Model Checking}}, and {{Abstract Interpretation}}},
  author = {Cousot, Patrick and Cousot, Radhia and Logozzo, Francesco},
  editor = {Jhala, Ranjit and Schmidt, David},
  year = {2011},
  volume = {6538},
  publisher = {{Springer Berlin Heidelberg}},
  address = {{Berlin, Heidelberg}},
  isbn = {978-3-642-18274-7 978-3-642-18275-4},
  doi = {https://doi.org/10.1007/978-3-642-18275-4_12}
}

@inproceedings{demouraZ3EfficientSMT2008,
  title = {Z3: {{An Efficient SMT Solver}}},
  shorttitle = {Z3},
  booktitle = {Tools and {{Algorithms}} for the {{Construction}} and {{Analysis}} of {{Systems}}},
  author = {{de Moura}, Leonardo and Bj{\o}rner, Nikolaj},
  editor = {Ramakrishnan, C. R. and Rehof, Jakob},
  year = {2008},
  series = {Lecture {{Notes}} in {{Computer Science}}},
  publisher = {{Springer}},
  address = {{Berlin, Heidelberg}},
  isbn = {978-3-540-78800-3},
  doi = {https://doi.org/10.1007/978-3-540-78800-3\_24}
}

@inproceedings{gordonProbabilisticProgramming2014,
  title = {Probabilistic {{Programming}}},
  booktitle = {Proceedings of the on {{Future}} of {{Software Engineering}}},
  author = {Gordon, Andrew D. and Henzinger, Thomas A. and Nori, Aditya V. and Rajamani, Sriram K.},
  year = {2014},
  series = {{{FOSE}} 2014},
  publisher = {{ACM}},
  address = {{New York, NY, USA}},
  isbn = {978-1-4503-2865-4},
  doi = {https://doi.org/10.1145/2593882.2593900},
}

@article{harkAimingLowHarder2019,
  title = {Aiming Low Is Harder: Induction for Lower Bounds in Probabilistic Program Verification},
  shorttitle = {Aiming Low Is Harder},
  author = {Hark, Marcel and Kaminski, Benjamin Lucien and Giesl, J{\"u}rgen and Katoen, Joost-Pieter},
  year = {2019},
  month = dec,
  journal = {Proceedings of the ACM on Programming Languages},
  volume = {4},
  number = {POPL},
  doi = {https://doi.org/10.1145/3371105}
}

@article{hoareAxiomaticBasisComputer1969,
  title = {An {{Axiomatic Basis}} for {{Computer Programming}}},
  author = {Hoare, C A R},
  year = {1969},
  journal = {Commun. ACM},
  volume = {12},
  number = {10},
  doi = {https://doi.org/10.1145/363235.363259},
}

@article{ProbabilisticGuardedCommands2005,
  title = {Probabilistic {{Guarded Commands Mechanized}} in {{HOL}}},
  author = {Hurd, J. and McIver, Annabelle and Morgan, Carroll},
  year = {2005},
  journal = {Electron. Notes Theor. Comput. Sci.},
  doi = {https://doi.org/10.1016/j.tcs.2005.08.005}
}

@phdthesis{kaminskiAdvancedWeakestPrecondition2019,
  title = {Advanced {{Weakest Precondition Calculi}} for {{Probabilistic Programs}}},
  author = {Kaminski, Benjamin Lucien},
  year = {2019},
  month = feb,
  school = {RWTH Aachen University},
  doi = {https://doi.org/10.18154/RWTH-2019-01829}
}

@incollection{kaminskiExpectedRuntimeAnalysis2020,
  title = {Expected {{Runtime Analysis}} by {{Program Verification}}},
  booktitle = {Foundations of {{Probabilistic Programming}}},
  author = {Kaminski, Benjamin Lucien and Katoen, Joost-Pieter and Matheja, Christoph},
  editor = {Silva, Alexandra and Barthe, Gilles and Katoen, Joost-Pieter},
  year = {2020},
  publisher = {{Cambridge University Press}},
  address = {{Cambridge}},
  isbn = {978-1-108-48851-8},
  doi = {https://doi.org/10.1017/9781108770750}
}

@article{kaminskiHardnessAnalyzingProbabilistic2019,
  title = {On the Hardness of Analyzing Probabilistic Programs},
  author = {Kaminski, Benjamin Lucien and Katoen, Joost-Pieter and Matheja, Christoph},
  year = {2019},
  month = apr,
  journal = {Acta Informatica},
  volume = {56},
  number = {3},
  doi = {https://doi.org/10.1007/s00236-018-0321-1}
}

@inproceedings{kaminskiWeakestPreconditionReasoning2016,
  title = {Weakest {{Precondition Reasoning}} for {{Expected Run}}\textendash{{Times}} of {{Probabilistic Programs}}},
  booktitle = {Programming {{Languages}} and {{Systems}}},
  author = {Kaminski, Benjamin Lucien and Katoen, Joost-Pieter and Matheja, Christoph and Olmedo, Federico},
  editor = {Thiemann, Peter},
  year = {2016},
  series = {Lecture {{Notes}} in {{Computer Science}}},
  publisher = {{Springer}},
  address = {{Berlin, Heidelberg}},
  isbn = {978-3-662-49498-1},
  doi = {https://doi.org/10.1007/978-3-662-49498-1\_15}
}

@article{kaminskiWeakestPreconditionReasoning2018,
  title = {Weakest {{Precondition Reasoning}} for {{Expected Runtimes}} of {{Randomized Algorithms}}},
  author = {Kaminski, Benjamin Lucien and Katoen, Joost-Pieter and Matheja, Christoph and Olmedo, Federico},
  year = {2018},
  month = aug,
  journal = {Journal of the ACM},
  volume = {65},
  number = {5},
  doi = {https://doi.org/10.1145/3208102},
}

@book{kleeneIntroductionMetamathematics1952,
  title = {Introduction to {{Metamathematics}}},
  author = {Kleene, Stephen Cole},
  year = {1952},
  publisher = {{North Holland}},
  doi = {10.2307/2268620}
}

@article{DBLP:journals/toplas/TakisakaOUH21,
  author       = {Toru Takisaka and
                  Yuichiro Oyabu and
                  Natsuki Urabe and
                  Ichiro Hasuo},
  title        = {Ranking and Repulsing Supermartingales for Reachability in Randomized
                  Programs},
  journal      = {{ACM} Trans. Program. Lang. Syst.},
  volume       = {43},
  number       = {2},
  pages        = {5:1--5:46},
  year         = {2021},
  doi = {https://doi.org/10.1145/3450967}
}

@inproceedings{kozenProbabilisticPDL1983,
  author       = {Dexter Kozen},
  title        = {A Probabilistic {PDL}},
  booktitle    = {{STOC}},
  pages        = {291--297},
  publisher    = {{ACM}},
  year         = {1983},
  doi = {https://doi.org/10.1145/800061.808758}
}

@article{kozenProbabilisticPDL1985,
  author       = {Dexter Kozen},
  title        = {A Probabilistic {PDL}},
  journal      = {J. Comput. Syst. Sci.},
  volume       = {30},
  number       = {2},
  pages        = {162--178},
  year         = {1985},
  doi = {https://doi.org/10.1016/0022-0000(85)90012-1}
}

@article{DBLP:journals/fac/KirchnerKPSY15,
  author       = {Florent Kirchner and
                  Nikolai Kosmatov and
                  Virgile Prevosto and
                  Julien Signoles and
                  Boris Yakobowski},
  title        = {Frama-C: {A} software analysis perspective},
  journal      = {Formal Aspects Comput.},
  volume       = {27},
  number       = {3},
  pages        = {573--609},
  year         = {2015},
  doi = {https://doi.org/10.1007/s00165-014-0326-7}
}

@inproceedings{leinoDafnyAutomaticProgram2010,
  title = {Dafny: {{An Automatic Program Verifier}} for {{Functional Correctness}}},
  shorttitle = {Dafny},
  booktitle = {Logic for {{Programming}}, {{Artificial Intelligence}}, and {{Reasoning}}},
  author = {Leino, K. Rustan M.},
  editor = {Clarke, Edmund M. and Voronkov, Andrei},
  year = {2010},
  series = {Lecture {{Notes}} in {{Computer Science}}},
  publisher = {{Springer}},
  address = {{Berlin, Heidelberg}},
  isbn = {978-3-642-17511-4},
  doi = {https://doi.org/10.1007/978-3-642-17511-4_20}
}

@book{leinoThisBoogie2008,
  title = {This Is {{Boogie}} 2},
  author = {Leino, K. Rustan M.},
  year = {2008}
}

@inproceedings{DBLP:conf/esop/FilliatreP13,
  author       = {Jean{-}Christophe Filli{\^{a}}tre and
                  Andrei Paskevich},
  title        = {Why3 - Where Programs Meet Provers},
  booktitle    = {{ESOP}},
  series       = {Lecture Notes in Computer Science},
  volume       = {7792},
  pages        = {125--128},
  publisher    = {Springer},
  year         = {2013},
  doi = {https://doi.org/10.1007/978-3-642-37036-6_8}
}

@misc{leutgebAutomatedExpectedAmortised2022,
  author       = {Lorenz Leutgeb and
                  Georg Moser and
                  Florian Zuleger},
  title        = {Automated Expected Amortised Cost Analysis of Probabilistic Data Structures},
  booktitle    = {{CAV} {(2)}},
  series       = {Lecture Notes in Computer Science},
  volume       = {13372},
  pages        = {70--91},
  publisher    = {Springer},
  year         = {2022},
  doi = {https://doi.org/10.1007/978-3-031-13188-2_4}
}

@book{mciverAbstractionRefinementProof2005,
  title = {Abstraction, {{Refinement}} and {{Proof}} for {{Probabilistic Systems}}},
  author = {McIver, Annabelle and Morgan, Charles Carroll},
  year = {2005},
  series = {Monographs in {{Computer Science}}},
  publisher = {{Springer-Verlag}},
  address = {{New York}},
  isbn = {978-0-387-40115-7},
  doi = {https://doi.org/10.1007/b138392}
}

@article{mciverNewProofRule2018,
  title = {A New Proof Rule for Almost-Sure Termination},
  author = {McIver, Annabelle and Morgan, Carroll and Kaminski, Benjamin Lucien and Katoen, Joost-Pieter},
  year = {2018},
  month = jan,
  journal = {Proceedings of the ACM on Programming Languages},
  volume = {2},
  number = {POPL},
  doi = {https://doi.org/10.1145/3158121}
}

@article{mullerBuildingDeductiveProgram2019,
  title = {Building {{Deductive Program Verifiers}} - {{Lecture Notes}}},
  author = {M{\"u}ller, Peter},
  year = {2019},
  journal = {Engineering Secure and Dependable Software Systems},
  publisher = {{IOS Press}}
}

@inproceedings{ngoBoundedExpectationsResource2018,
  title = {Bounded Expectations: Resource Analysis for Probabilistic Programs},
  shorttitle = {Bounded Expectations},
  booktitle = {Proceedings of the 39th {{ACM SIGPLAN Conference}} on {{Programming Language Design}} and {{Implementation}}},
  author = {Ngo, Van Chan and Carbonneaux, Quentin and Hoffmann, Jan},
  year = {2018},
  month = jun,
  series = {{{PLDI}} 2018},
  publisher = {{Association for Computing Machinery}},
  address = {{New York, NY, USA}},
  isbn = {978-1-4503-5698-5},
  doi = {https://doi.org/10.1145/3192366.3192394}
}

@article{ohearnIncorrectnessLogic2020,
  title = {Incorrectness Logic},
  author = {O'Hearn, Peter W.},
  year = {2020},
  month = jan,
  journal = {Proceedings of the ACM on Programming Languages},
  volume = {4},
  number = {POPL},
  doi = {https://doi.org/10.1145/3371078}
}

@article{olmedoConditioningProbabilisticProgramming2018,
  title = {Conditioning in {{Probabilistic Programming}}},
  author = {Olmedo, Federico and Gretz, Friedrich and Jansen, Nils and Kaminski, Benjamin Lucien and Katoen, Joost-Pieter and Mciver, Annabelle},
  year = {2018},
  month = jan,
  journal = {ACM Transactions on Programming Languages and Systems},
  volume = {40},
  number = {1},
  doi = {https://doi.org/10.1145/3156018}
}

@inproceedings{olmedoReasoningRecursiveProbabilistic2016,
  title = {Reasoning about {{Recursive Probabilistic Programs}}},
  booktitle = {Proceedings of the 31st {{Annual ACM}}/{{IEEE Symposium}} on {{Logic}} in {{Computer Science}}},
  author = {Olmedo, Federico and Kaminski, Benjamin Lucien and Katoen, Joost-Pieter and Matheja, Christoph},
  year = {2016},
  month = jul,
  series = {{{LICS}} '16},
  publisher = {{Association for Computing Machinery}},
  address = {{New York, NY, USA}},
  isbn = {978-1-4503-4391-6},
  doi = {https://doi.org/10.1145/2933575.2935317}
}

@inproceedings{pardoSpecificationLogicPrograms2022,
	author       = {Ra{\'{u}}l Pardo and
	Einar Broch Johnsen and
	Ina Schaefer and
	Andrzej Wasowski},
	title        = {A Specification Logic for Programs in the Probabilistic Guarded Command
	Language},
	booktitle    = {{ICTAC}},
	series       = {Lecture Notes in Computer Science},
	volume       = {13572},
	pages        = {369--387},
	publisher    = {Springer},
	year         = {2022},
	doi = {https://doi.org/10.1007/978-3-031-17715-6_24}
}

@article{parkFixpointInductionProofs1969,
  title = {Fixpoint Induction and Proofs of Program Properties},
  author = {Park, David},
  year = {1969},
  journal = {Machine Intelligence},
  volume = {5},
  publisher = {{Edinburgh University Press}}
}

@inproceedings{preiningGodelLogicsSurvey2010,
  title = {G\"odel {{Logics}} \textendash{} {{A Survey}}},
  booktitle = {Logic for {{Programming}}, {{Artificial Intelligence}}, and {{Reasoning}}},
  author = {Preining, Norbert},
  editor = {Ferm{\"u}ller, Christian G. and Voronkov, Andrei},
  year = {2010},
  series = {Lecture {{Notes}} in {{Computer Science}}},
  publisher = {{Springer}},
  address = {{Berlin, Heidelberg}},
  isbn = {978-3-642-16242-8},
  doi = {https://doi.org/10.1007/978-3-642-16242-8_4}
}

@inproceedings{sheeranCheckingSafetyProperties2000,
	author       = {Mary Sheeran and
	Satnam Singh and
	Gunnar St{\aa}lmarck},
	editor       = {Warren A. Hunt Jr. and
	Steven D. Johnson},
	title        = {Checking Safety Properties Using Induction and a SAT-Solver},
	booktitle    = {Formal Methods in Computer-Aided Design, Third International Conference,
	{FMCAD} 2000, Austin, Texas, USA, November 1-3, 2000, Proceedings},
	series       = {Lecture Notes in Computer Science},
	volume       = {1954},
	pages        = {108--125},
	publisher    = {Springer},
	year         = {2000},
	url          = {https://doi.org/10.1007/3-540-40922-X\_8},
	doi          = {10.1007/3-540-40922-X\_8},
	bibsource    = {dblp computer science bibliography, https://dblp.org}
}

@article{zhangQuantitativeStrongestPost2022,
  author       = {Linpeng Zhang and
                  Benjamin Lucien Kaminski},
  title        = {Quantitative strongest post: a calculus for reasoning about the flow
                  of quantitative information},
  journal      = {Proc. {ACM} Program. Lang.},
  volume       = {6},
  number       = {{OOPSLA1}},
  pages        = {1--29},
  year         = {2022},
  doi = {https://doi.org/10.1145/3527331}
}

@article{DBLP:journals/corr/abs-1304-1916,
  author       = {J{\'{e}}r{\'{e}}mie O. Lumbroso},
  title        = {Optimal Discrete Uniform Generation from Coin Flips, and Applications},
  journal      = {CoRR},
  volume       = {abs/1304.1916},
  year         = {2013},
  url          = {http://arxiv.org/abs/1304.1916},
  eprinttype    = {arXiv},
  eprint       = {1304.1916},
  bibsource    = {dblp computer science bibliography, https://dblp.org}
}

@inproceedings{DBLP:conf/tacas/DArgenioKRT97,
  author       = {Pedro R. D'Argenio and
                  Joost{-}Pieter Katoen and
                  Theo C. Ruys and
                  Jan Tretmans},
  editor       = {Ed Brinksma},
  title        = {The Bounded Retransmission Protocol Must Be on Time!},
  booktitle    = {Tools and Algorithms for Construction and Analysis of Systems, Third
                  International Workshop, {TACAS} '97, Enschede, The Netherlands, April
                  2-4, 1997, Proceedings},
  series       = {Lecture Notes in Computer Science},
  volume       = {1217},
  pages        = {416--431},
  publisher    = {Springer},
  year         = {1997},
  url          = {https://doi.org/10.1007/BFb0035403},
  doi          = {10.1007/BFb0035403},
  bibsource    = {dblp computer science bibliography, https://dblp.org}
}

@inproceedings{DBLP:conf/types/HelminkSV93,
  author       = {Leen Helmink and
                  M. P. A. Sellink and
                  Frits W. Vaandrager},
  editor       = {Henk Barendregt and
                  Tobias Nipkow},
  title        = {Proof-Checking a Data Link Protocol},
  booktitle    = {Types for Proofs and Programs, International Workshop TYPES'93, Nijmegen,
                  The Netherlands, May 24-28, 1993, Selected Papers},
  series       = {Lecture Notes in Computer Science},
  volume       = {806},
  pages        = {127--165},
  publisher    = {Springer},
  year         = {1993},
  url          = {https://doi.org/10.1007/3-540-58085-9\_75},
  doi          = {10.1007/3-540-58085-9\_75},
  bibsource    = {dblp computer science bibliography, https://dblp.org}
}

@inproceedings{viper,
  author       = {Peter M{\"{u}}ller and
                  Malte Schwerhoff and
                  Alexander J. Summers},
  editor       = {Barbara Jobstmann and
                  K. Rustan M. Leino},
  title        = {Viper: {A} Verification Infrastructure for Permission-Based Reasoning},
  booktitle    = {Verification, Model Checking, and Abstract Interpretation - 17th International
                  Conference, {VMCAI} 2016, St. Petersburg, FL, USA, January 17-19,
                  2016. Proceedings},
  series       = {Lecture Notes in Computer Science},
  volume       = {9583},
  pages        = {41--62},
  publisher    = {Springer},
  year         = {2016},
  url          = {https://doi.org/10.1007/978-3-662-49122-5\_2},
  doi          = {10.1007/978-3-662-49122-5\_2},
  bibsource    = {dblp computer science bibliography, https://dblp.org}
}

@online{viperappendix,
  author = {Peter M{\"{u}}ller and
                  Malte Schwerhoff and
                  Alexander J. Summers},
  title = {Online appendix to Viper: {A} Verification Infrastructure for Permission-Based Reasoning},
  year = 2016,
  url = {http://viper.ethz.ch/examples/vmcai16/index.html},
}

@inproceedings{kushilevitzRandomizedMutualExclusion1992,
  author       = {Eyal Kushilevitz and
                  Michael O. Rabin},
  editor       = {Norman C. Hutchinson},
  title        = {Randomized Mutual Exclusion Algorithms Revisited},
  booktitle    = {Proceedings of the Eleventh Annual {ACM} Symposium on Principles of
                  Distributed Computing, Vancouver, British Columbia, Canada, August
                  10-12, 1992},
  pages        = {275--283},
  publisher    = {{ACM}},
  year         = {1992},
  url          = {https://doi.org/10.1145/135419.135468},
  doi          = {10.1145/135419.135468},
  bibsource    = {dblp computer science bibliography, https://dblp.org}
}

@inproceedings{hoelzlert,
	author       = {Johannes H{\"{o}}lzl},
	editor       = {Jasmin Christian Blanchette and
	Stephan Merz},
	title        = {Formalising Semantics for Expected Running Time of Probabilistic Programs},
	booktitle    = {Interactive Theorem Proving - 7th International Conference, {ITP}
	2016, Nancy, France, August 22-25, 2016, Proceedings},
	series       = {Lecture Notes in Computer Science},
	volume       = {9807},
	pages        = {475--482},
	publisher    = {Springer},
	year         = {2016},
	url          = {https://doi.org/10.1007/978-3-319-43144-4\_30},
	doi          = {10.1007/978-3-319-43144-4\_30},
	bibsource    = {dblp computer science bibliography, https://dblp.org}
}

@inproceedings{easycrypt1,
	author       = {Gilles Barthe and
	Fran{\c{c}}ois Dupressoir and
	Benjamin Gr{\'{e}}goire and
	C{\'{e}}sar Kunz and
	Benedikt Schmidt and
	Pierre{-}Yves Strub},
	editor       = {Alessandro Aldini and
	Javier L{\'{o}}pez and
	Fabio Martinelli},
	title        = {EasyCrypt: {A} Tutorial},
	booktitle    = {Foundations of Security Analysis and Design {VII} - {FOSAD} 2012/2013
	Tutorial Lectures},
	series       = {Lecture Notes in Computer Science},
	volume       = {8604},
	pages        = {146--166},
	publisher    = {Springer},
	year         = {2013},
	url          = {https://doi.org/10.1007/978-3-319-10082-1\_6},
	doi          = {10.1007/978-3-319-10082-1\_6},
	bibsource    = {dblp computer science bibliography, https://dblp.org}
}

@inproceedings{easycrypt2,
	author       = {Gilles Barthe and
	Benjamin Gr{\'{e}}goire and
	Sylvain Heraud and
	Santiago Zanella B{\'{e}}guelin},
	editor       = {Phillip Rogaway},
	title        = {Computer-Aided Security Proofs for the Working Cryptographer},
	booktitle    = {Advances in Cryptology - {CRYPTO} 2011 - 31st Annual Cryptology Conference,
	Santa Barbara, CA, USA, August 14-18, 2011. Proceedings},
	series       = {Lecture Notes in Computer Science},
	volume       = {6841},
	pages        = {71--90},
	publisher    = {Springer},
	year         = {2011},
	url          = {https://doi.org/10.1007/978-3-642-22792-9\_5},
	doi          = {10.1007/978-3-642-22792-9\_5},
	bibsource    = {dblp computer science bibliography, https://dblp.org}
}

@inproceedings{koat,
	author       = {Fabian Meyer and
	Marcel Hark and
	J{\"{u}}rgen Giesl},
	editor       = {Jan Friso Groote and
	Kim Guldstrand Larsen},
	title        = {Inferring Expected Runtimes of Probabilistic Integer Programs Using
	Expected Sizes},
	booktitle    = {Tools and Algorithms for the Construction and Analysis of Systems
	- 27th International Conference, {TACAS} 2021, Held as Part of the
	European Joint Conferences on Theory and Practice of Software, {ETAPS}
	2021, Luxembourg City, Luxembourg, March 27 - April 1, 2021, Proceedings,
	Part {I}},
	series       = {Lecture Notes in Computer Science},
	volume       = {12651},
	pages        = {250--269},
	publisher    = {Springer},
	year         = {2021},
	url          = {https://doi.org/10.1007/978-3-030-72016-2\_14},
	doi          = {10.1007/978-3-030-72016-2\_14},
	bibsource    = {dblp computer science bibliography, https://dblp.org}
}

@inproceedings{moments1,
	author       = {Daneshvar Amrollahi and
	Ezio Bartocci and
	George Kenison and
	Laura Kov{\'{a}}cs and
	Marcel Moosbrugger and
	Miroslav Stankovic},
	editor       = {Gagandeep Singh and
	Caterina Urban},
	title        = {Solving Invariant Generation for Unsolvable Loops},
	booktitle    = {Static Analysis - 29th International Symposium, {SAS} 2022, Auckland,
	New Zealand, December 5-7, 2022, Proceedings},
	series       = {Lecture Notes in Computer Science},
	volume       = {13790},
	pages        = {19--43},
	publisher    = {Springer},
	year         = {2022},
	url          = {https://doi.org/10.1007/978-3-031-22308-2\_3},
	doi          = {10.1007/978-3-031-22308-2\_3},
	bibsource    = {dblp computer science bibliography, https://dblp.org}
}

@inproceedings{amber,
	author       = {Marcel Moosbrugger and
	Ezio Bartocci and
	Joost{-}Pieter Katoen and
	Laura Kov{\'{a}}cs},
	editor       = {Marieke Huisman and
	Corina S. Pasareanu and
	Naijun Zhan},
	title        = {The Probabilistic Termination Tool Amber},
	booktitle    = {Formal Methods - 24th International Symposium, {FM} 2021, Virtual
	Event, November 20-26, 2021, Proceedings},
	series       = {Lecture Notes in Computer Science},
	volume       = {13047},
	pages        = {667--675},
	publisher    = {Springer},
	year         = {2021},
	url          = {https://doi.org/10.1007/978-3-030-90870-6\_36},
	doi          = {10.1007/978-3-030-90870-6\_36},
	bibsource    = {dblp computer science bibliography, https://dblp.org}
}

@inproceedings{automatedterm,
	author       = {Marcel Moosbrugger and
	Ezio Bartocci and
	Joost{-}Pieter Katoen and
	Laura Kov{\'{a}}cs},
	editor       = {Nobuko Yoshida},
	title        = {Automated Termination Analysis of Polynomial Probabilistic Programs},
	booktitle    = {Programming Languages and Systems - 30th European Symposium on Programming,
	{ESOP} 2021, Held as Part of the European Joint Conferences on Theory
	and Practice of Software, {ETAPS} 2021, Luxembourg City, Luxembourg,
	March 27 - April 1, 2021, Proceedings},
	series       = {Lecture Notes in Computer Science},
	volume       = {12648},
	pages        = {491--518},
	publisher    = {Springer},
	year         = {2021},
	url          = {https://doi.org/10.1007/978-3-030-72019-3\_18},
	doi          = {10.1007/978-3-030-72019-3\_18},
	bibsource    = {dblp computer science bibliography, https://dblp.org}
}

@article{moment3,
  author       = {Ezio Bartocci and
                  Laura Kov{\'{a}}cs and
                  Miroslav Stankovic},
  title        = {Mora - Automatic Generation of Moment-Based Invariants},
  booktitle    = {{TACAS} {(1)}},
  series       = {Lecture Notes in Computer Science},
  volume       = {12078},
  pages        = {492--498},
  publisher    = {Springer},
  year         = {2020},
  doi = {https://doi.org/10.1007/978-3-030-45190-5_28}
}

@article{symbolicexec,
	author       = {Zachary Susag and
	Sumit Lahiri and
	Justin Hsu and
	Subhajit Roy},
	title        = {Symbolic execution for randomized programs},
	journal      = {Proc. {ACM} Program. Lang.},
	volume       = {6},
	number       = {{OOPSLA2}},
	pages        = {1583--1612},
	year         = {2022},
	url          = {https://doi.org/10.1145/3563344},
	doi          = {10.1145/3563344},
	bibsource    = {dblp computer science bibliography, https://dblp.org}
}

@inproceedings{invsys1,
	author       = {Yijun Feng and
	Lijun Zhang and
	David N. Jansen and
	Naijun Zhan and
	Bican Xia},
	editor       = {Deepak D'Souza and
	K. Narayan Kumar},
	title        = {Finding Polynomial Loop Invariants for Probabilistic Programs},
	booktitle    = {Automated Technology for Verification and Analysis - 15th International
	Symposium, {ATVA} 2017, Pune, India, October 3-6, 2017, Proceedings},
	series       = {Lecture Notes in Computer Science},
	volume       = {10482},
	pages        = {400--416},
	publisher    = {Springer},
	year         = {2017},
	url          = {https://doi.org/10.1007/978-3-319-68167-2\_26},
	doi          = {10.1007/978-3-319-68167-2\_26},
	bibsource    = {dblp computer science bibliography, https://dblp.org}
}

@inproceedings{invsys2,
	author       = {Yu{-}Fang Chen and
	Chih{-}Duo Hong and
	Bow{-}Yaw Wang and
	Lijun Zhang},
	editor       = {Daniel Kroening and
	Corina S. Pasareanu},
	title        = {Counterexample-Guided Polynomial Loop Invariant Generation by Lagrange
	Interpolation},
	booktitle    = {Computer Aided Verification - 27th International Conference, {CAV}
	2015, San Francisco, CA, USA, July 18-24, 2015, Proceedings, Part
	{I}},
	series       = {Lecture Notes in Computer Science},
	volume       = {9206},
	pages        = {658--674},
	publisher    = {Springer},
	year         = {2015},
	url          = {https://doi.org/10.1007/978-3-319-21690-4\_44},
	doi          = {10.1007/978-3-319-21690-4\_44},
	bibsource    = {dblp computer science bibliography, https://dblp.org}
}

@inproceedings{invsys3,
	author       = {Joost{-}Pieter Katoen and
	Annabelle McIver and
	Larissa Meinicke and
	Carroll C. Morgan},
	editor       = {Radhia Cousot and
	Matthieu Martel},
	title        = {Linear-Invariant Generation for Probabilistic Programs: - Automated
	Support for Proof-Based Methods},
	booktitle    = {Static Analysis - 17th International Symposium, {SAS} 2010, Perpignan,
	France, September 14-16, 2010. Proceedings},
	series       = {Lecture Notes in Computer Science},
	volume       = {6337},
	pages        = {390--406},
	publisher    = {Springer},
	year         = {2010},
	url          = {https://doi.org/10.1007/978-3-642-15769-1\_24},
	doi          = {10.1007/978-3-642-15769-1\_24},
	bibsource    = {dblp computer science bibliography, https://dblp.org}
}

@inproceedings{invsys4,
	author       = {Gilles Barthe and
	Thomas Espitau and
	Luis Mar{\'{\i}}a Ferrer Fioriti and
	Justin Hsu},
	editor       = {Swarat Chaudhuri and
	Azadeh Farzan},
	title        = {Synthesizing Probabilistic Invariants via Doob's Decomposition},
	booktitle    = {Computer Aided Verification - 28th International Conference, {CAV}
	2016, Toronto, ON, Canada, July 17-23, 2016, Proceedings, Part {I}},
	series       = {Lecture Notes in Computer Science},
	volume       = {9779},
	pages        = {43--61},
	publisher    = {Springer},
	year         = {2016},
	url          = {https://doi.org/10.1007/978-3-319-41528-4\_3},
	doi          = {10.1007/978-3-319-41528-4\_3},
	bibsource    = {dblp computer science bibliography, https://dblp.org}
}

@inproceedings{chakarovProbabilisticProgramAnalysis2013,
	author       = {Aleksandar Chakarov and
	Sriram Sankaranarayanan},
	editor       = {Natasha Sharygina and
	Helmut Veith},
	title        = {Probabilistic Program Analysis with Martingales},
	booktitle    = {Computer Aided Verification - 25th International Conference, {CAV}
	2013, Saint Petersburg, Russia, July 13-19, 2013. Proceedings},
	series       = {Lecture Notes in Computer Science},
	volume       = {8044},
	pages        = {511--526},
	publisher    = {Springer},
	year         = {2013},
	url          = {https://doi.org/10.1007/978-3-642-39799-8\_34},
	doi          = {10.1007/978-3-642-39799-8\_34},
	bibsource    = {dblp computer science bibliography, https://dblp.org}
}

@inproceedings{termpos,
	author       = {Krishnendu Chatterjee and
	Hongfei Fu and
	Amir Kafshdar Goharshady},
	editor       = {Swarat Chaudhuri and
	Azadeh Farzan},
	title        = {Termination Analysis of Probabilistic Programs Through Positivstellensatz's},
	booktitle    = {Computer Aided Verification - 28th International Conference, {CAV}
	2016, Toronto, ON, Canada, July 17-23, 2016, Proceedings, Part {I}},
	series       = {Lecture Notes in Computer Science},
	volume       = {9779},
	pages        = {3--22},
	publisher    = {Springer},
	year         = {2016},
	url          = {https://doi.org/10.1007/978-3-319-41528-4\_1},
	doi          = {10.1007/978-3-319-41528-4\_1},
	bibsource    = {dblp computer science bibliography, https://dblp.org}
}

@inproceedings{stochterm,
	author       = {Krishnendu Chatterjee and
	Petr Novotn{\'{y}} and
	Dorde Zikelic},
	editor       = {Giuseppe Castagna and
	Andrew D. Gordon},
	title        = {Stochastic invariants for probabilistic termination},
	booktitle    = {Proceedings of the 44th {ACM} {SIGPLAN} Symposium on Principles of
	Programming Languages, {POPL} 2017, Paris, France, January 18-20,
	2017},
	pages        = {145--160},
	publisher    = {{ACM}},
	year         = {2017},
	url          = {https://doi.org/10.1145/3009837.3009873},
	doi          = {10.1145/3009837.3009873},
	bibsource    = {dblp computer science bibliography, https://dblp.org}
}

@inproceedings{termnondet,
	author       = {Hongfei Fu and
	Krishnendu Chatterjee},
	editor       = {Constantin Enea and
	Ruzica Piskac},
	title        = {Termination of Nondeterministic Probabilistic Programs},
	booktitle    = {Verification, Model Checking, and Abstract Interpretation - 20th International
	Conference, {VMCAI} 2019, Cascais, Portugal, January 13-15, 2019,
	Proceedings},
	series       = {Lecture Notes in Computer Science},
	volume       = {11388},
	pages        = {468--490},
	publisher    = {Springer},
	year         = {2019},
	url          = {https://doi.org/10.1007/978-3-030-11245-5\_22},
	doi          = {10.1007/978-3-030-11245-5\_22},
	bibsource    = {dblp computer science bibliography, https://dblp.org}
}

@inproceedings{termcompletesound,
	author       = {Luis Mar{\'{\i}}a Ferrer Fioriti and
	Holger Hermanns},
	editor       = {Sriram K. Rajamani and
	David Walker},
	title        = {Probabilistic Termination: Soundness, Completeness, and Compositionality},
	booktitle    = {Proceedings of the 42nd Annual {ACM} {SIGPLAN-SIGACT} Symposium on
	Principles of Programming Languages, {POPL} 2015, Mumbai, India, January
	15-17, 2015},
	pages        = {489--501},
	publisher    = {{ACM}},
	year         = {2015},
	url          = {https://doi.org/10.1145/2676726.2677001},
	doi          = {10.1145/2676726.2677001},
	bibsource    = {dblp computer science bibliography, https://dblp.org}
}

@article{rankingsuper,
	author       = {Sheshansh Agrawal and
	Krishnendu Chatterjee and
	Petr Novotn{\'{y}}},
	title        = {Lexicographic ranking supermartingales: an efficient approach to termination
	of probabilistic programs},
	journal      = {Proc. {ACM} Program. Lang.},
	volume       = {2},
	number       = {{POPL}},
	pages        = {34:1--34:32},
	year         = {2018},
	url          = {https://doi.org/10.1145/3158122},
	doi          = {10.1145/3158122},
	bibsource    = {dblp computer science bibliography, https://dblp.org}
}

@inproceedings{cegisupast,
	author       = {Alessandro Abate and
	Mirco Giacobbe and
	Diptarko Roy},
	editor       = {Alexandra Silva and
	K. Rustan M. Leino},
	title        = {Learning Probabilistic Termination Proofs},
	booktitle    = {Computer Aided Verification - 33rd International Conference, {CAV}
	2021, Virtual Event, July 20-23, 2021, Proceedings, Part {II}},
	series       = {Lecture Notes in Computer Science},
	volume       = {12760},
	pages        = {3--26},
	publisher    = {Springer},
	year         = {2021},
	url          = {https://doi.org/10.1007/978-3-030-81688-9\_1},
	doi          = {10.1007/978-3-030-81688-9\_1},
	bibsource    = {dblp computer science bibliography, https://dblp.org}
}

@inproceedings{invsys7,
	author       = {Jialu Bao and
	Nitesh Trivedi and
	Drashti Pathak and
	Justin Hsu and
	Subhajit Roy},
	editor       = {Sharon Shoham and
	Yakir Vizel},
	title        = {Data-Driven Invariant Learning for Probabilistic Programs},
	booktitle    = {Computer Aided Verification - 34th International Conference, {CAV}
	2022, Haifa, Israel, August 7-10, 2022, Proceedings, Part {I}},
	series       = {Lecture Notes in Computer Science},
	volume       = {13371},
	pages        = {33--54},
	publisher    = {Springer},
	year         = {2022},
	url          = {https://doi.org/10.1007/978-3-031-13185-1\_3},
	doi          = {10.1007/978-3-031-13185-1\_3},
	bibsource    = {dblp computer science bibliography, https://dblp.org}
}

@article{modularcost,
	author       = {Martin Avanzini and
	Georg Moser and
	Michael Schaper},
	title        = {A modular cost analysis for probabilistic programs},
	journal      = {Proc. {ACM} Program. Lang.},
	volume       = {4},
	number       = {{OOPSLA}},
	pages        = {172:1--172:30},
	year         = {2020},
	url          = {https://doi.org/10.1145/3428240},
	doi          = {10.1145/3428240},
	bibsource    = {dblp computer science bibliography, https://dblp.org}
}

@book{isbabelle,
	author       = {Tobias Nipkow and
	Lawrence C. Paulson and
	Markus Wenzel},
	title        = {Isabelle/HOL - {A} Proof Assistant for Higher-Order Logic},
	series       = {Lecture Notes in Computer Science},
	volume       = {2283},
	publisher    = {Springer},
	year         = {2002},
	url          = {https://doi.org/10.1007/3-540-45949-9},
	doi          = {10.1007/3-540-45949-9},
	isbn         = {3-540-43376-7},
	bibsource    = {dblp computer science bibliography, https://dblp.org}
}

@inproceedings{r2,
	author       = {Aditya V. Nori and
	Chung{-}Kil Hur and
	Sriram K. Rajamani and
	Selva Samuel},
	editor       = {Carla E. Brodley and
	Peter Stone},
	title        = {{R2:} An Efficient {MCMC} Sampler for Probabilistic Programs},
	booktitle    = {Proceedings of the Twenty-Eighth {AAAI} Conference on Artificial Intelligence,
	July 27 -31, 2014, Qu{\'{e}}bec City, Qu{\'{e}}bec, Canada},
	pages        = {2476--2482},
	publisher    = {{AAAI} Press},
	year         = {2014},
	url          = {http://www.aaai.org/ocs/index.php/AAAI/AAAI14/paper/view/8192},
	bibsource    = {dblp computer science bibliography, https://dblp.org},
	doi = {https://doi.org/10.1609/aaai.v28i1.9060}
}

@phdthesis{mathejadiss,
	author       = {Christoph Matheja},
	title        = {Automated reasoning and randomization in separation logic},
	school       = {{RWTH} Aachen University, Germany},
	year         = {2020},
	doi = {https://doi.org/10.18154/RWTH-2020-00940}
}

@software{philipp_schroer_2023_8146987,
  author       = {Philipp Schroer and
                  Kevin Batz and
                  Benjamin Lucien Kaminski and
                  Joost-Pieter Katoen and
                  Christoph Matheja},
  title        = {{A Deductive Verification Infrastructure for 
                   Probabilistic Programs - Artifact Evaluation}},
  month        = jul,
  year         = 2023,
  publisher    = {Zenodo},
  doi          = {10.5281/zenodo.8146987},
  url          = {https://doi.org/10.5281/zenodo.8146987}
}
